\documentclass[12pt]{article}

\usepackage[T1]{fontenc}
\usepackage[margin=1in]{geometry}
\usepackage{setspace}
\usepackage{amsmath,amsfonts,amssymb,amsthm,bm}
\usepackage{graphicx,color,url,bbm}
\usepackage{booktabs,multirow,dcolumn,array}
\usepackage{placeins,ragged2e}
\usepackage[authoryear,round]{natbib}
\usepackage{etoolbox}
\usepackage[hidelinks]{hyperref}

\DeclareGraphicsExtensions{.pdf,.png,.eps}
\numberwithin{equation}{section}

\theoremstyle{plain}

\newtheorem{prop}{Proposition}[section]

\newtheorem{lemma}{Lemma}[section]
\newtheorem{assumption}{Assumption}[section]

\theoremstyle{definition}

\newtheorem{remark}{Remark}[section]

\newenvironment{conditions}
  {\begin{list}{}{\leftmargin=2.2em \labelwidth=1.7em \labelsep=0.5em
                  \itemsep=0.25\baselineskip \parsep=0pt
                  \topsep=0.25\baselineskip}}
  {\end{list}}
\newcounter{bean}

\DeclareMathOperator*{\argmin}{arg\,min}
\newcolumntype{d}[1]{D{.}{.}{#1}}
\def\ci{\perp\!\!\!\perp}

\makeatletter
\newcommand\midfootnotesize{\@setfontsize\midfootnotesize{9.5}{9}}
\makeatother

\newcommand{\refPropAsym}{\ref{prop: asymptotics}}
\newcommand{\refPropAsymNP}{\ref{prop: asymptotics np}}
\newcommand{\refPropPrim}{\ref{prop: np primitives}}
\newcommand{\refPropEffBound}{\ref{prop: efficiency bound}}
\newcommand{\refPropEffEst}{\ref{prop: efficient estimator}}

\newcommand{\keywords}[1]{%
  \par\medskip\noindent\textbf{Keywords: }#1%
}

\hypersetup{
  pdftitle={Point-identifying semiparametric sample selection models with no excluded variable},
  pdfauthor={Dongwoo Kim and Young Jun Lee},
  pdfkeywords={semiparametric sample selection, exclusion restriction, sieve estimation}
}

\title{Point-identifying semiparametric sample selection\\
       models with no excluded variable}

\author{%
  Dongwoo Kim\thanks{Corresponding author. Simon Fraser University and Korea University.
  Email: \texttt{dongwook@sfu.ca}.}
  \and
  Young Jun Lee\thanks{Hanyang University.
  Email: \texttt{youngjunlee@hanyang.ac.kr}.}
}

\date{July 2026}

\begin{document}

\maketitle

\let\ArxivBibliography\bibliography
\let\ArxivBibliographystyle\bibliographystyle
\renewcommand{\bibliography}[1]{}
\renewcommand{\bibliographystyle}[1]{}

\begin{abstract}
Sample selection is pervasive in applied economic studies. This paper proposes semiparametric selection models that achieve point identification without relying on exclusion restrictions. Our identification conditions require at least one continuously distributed covariate and certain nonlinearity in the selection process. We propose a two-step sieve plug-in estimator that is $\sqrt{n}$-consistent, asymptotically normal, and computationally straightforward, allowing for heteroskedasticity. We further derive the semiparametric efficiency bound for the model and propose a weighted variant of the estimator that attains the bound. Our approach provides a middle ground between Lee (2009)'s nonparametric bounds and Honoré and Hu (2020)'s linear selection bounds, while ensuring point identification. Simulation evidence confirms its excellent finite-sample performance. We apply our method to estimate the racial and gender wage disparities using data from the US Current Population Survey. Our estimates often lie outside the Honoré and Hu bounds.

\end{abstract}

\section{Introduction}

Sample selection poses a fundamental challenge in empirical economics. When workers self-select into employment or patients choose whether to seek medical care, the resulting datasets systematically exclude critical segments of the population, distorting our understanding of economic relationships and leading to misguided policy prescriptions. Even carefully designed random experiments are vulnerable, as systematic attrition introduces selection bias. Sample selection models are widely employed to address this concern. However, they often rely on strong assumptions. In particular, semiparametric selection models have long been thought to require exclusion restrictions, limiting their practical applicability. This paper challenges this conventional wisdom by proposing a class of semiparametric selection models that achieve point identification without exclusion restrictions. Specifically, we relax the linear selection and joint normality assumptions in the \cite{heckman1979sample} selection model. Point identification is established when there exists at least one continuous covariate. Following \cite{escanciano2016identification} (EJL henceforth), we leverage the nonlinearity of the selection probability, which is easily verifiable in practice. Consequently, the parameters in the outcome equation can be estimated via a partial linear regression by plugging nonparametrically estimated selection probabilities into the sieve-approximated nonparametric component. Our two-step sieve plug-in estimator is $\sqrt{n}$-consistent, asymptotically normal, and computationally scalable. Unlike many existing methods, we do not assume that unobserved heterogeneity is independent of regressors, allowing for heteroskedasticity. As we maintain the linearity of the outcome equation, incorporating a large set of covariates is straightforward.

\cite{heckman1979sample} pioneered correction methods for selection bias using a parametric model that assumes linearity in both selection and outcome equations:
\begin{equation}\label{eq: heckman selection model}
Y^* = \alpha + X\beta + V,\quad D=\mathbbm{1}[Z\gamma + \varepsilon \ge 0],\quad Y = D\cdot Y^*,
\end{equation}
where $Y^*$ is the latent outcome, $X$ and $Z$ are row vectors of covariates, $V$ and $\varepsilon$ are mean-zero unobserved heterogeneity terms that are jointly normally distributed and independent of $(X,Z)$, with $\operatorname{Var}(\varepsilon)$ normalized to $1$. Given $X=x$, $Z=z$, and $D=1$, we have
\[
E[Y\mid x,z,D=1] = \alpha + x\beta + \sigma_{V\varepsilon}\phi(z\gamma)/\Phi(z\gamma),
\]
where $\sigma_{V\varepsilon}=\operatorname{Cov}(V,\varepsilon)$, and $\phi(\cdot)$ and $\Phi(\cdot)$ denote the standard normal probability density function (p.d.f.) and cumulative distribution function (c.d.f.), respectively. Although Heckman’s model can identify $(\alpha,\beta)$ when $X=Z$, it is generally recommended to include at least one variable in $Z$ that is excluded from $X$ to strengthen identification. Without such an exclusion restriction, the performance of both the MLE (maximum likelihood estimator) and two-step estimator can be poor, as highlighted in the debates in \cite{hay1984let}, \cite{duan1984choosing}, and \cite{manning1987monte}. When only changes in the observed outcome $Y$ are of interest, a two-part model \citep{mullahy2024transform} or log-like transformation (e.g., $\log(Y+1)$) can serve as alternatives. However, \cite{chen2024logs} recently recommend a sample selection model for non-negative outcomes, particularly when the treatment affects the extensive margin (selection).

By relaxing joint normality, econometricians developed semiparametric \citep{chamberlain1986asymptotic, ahn1993semiparametric, newey2009two} and nonparametric \citep{das2003nonparametric} selection models. In both types of models, exclusion restrictions are widely regarded as essential for identification.\footnote{\cite{lee2009bounds} stated that “\textit{standard parametric or semiparametric methods for correcting for sample selection require exclusion restrictions}” (p. 1072). Similarly, \cite{honore2020selection} and \cite{kroft2024lee} claimed that an exclusion restriction is the \textit{key identifying assumption} in semiparametric selection models.} However, finding an excluded variable is often infeasible. \cite{lee2009bounds} proposed a bounds approach to overcome this challenge without invoking exclusion restrictions. By assuming selection monotonicity, under which individuals who are observed without treatment would also be observed with treatment, a trimming procedure is introduced to adjust for selection. Lee bounds are intuitive and tractable, making them widely used in empirical studies, particularly in experiments where some subjects drop out. However, these bounds are often too wide to yield meaningful insights, and their ability to incorporate covariates is limited. Slightly tighter bounds than unconditional bounds can be obtained by averaging the group-specific effects, weighted by covariate density. However, in practice, discretization is necessary for continuous covariates, and handling a large number of covariates is often infeasible. \cite{semenova2023generalized} generalizes Lee’s approach to a high-dimensional setting.

\cite{honore2020selection} (HH henceforth) later demonstrated that $\beta$ in \eqref{eq: heckman selection model} is set identified without joint normality and exclusion restrictions. The HH model serves as a semiparametric alternative to Lee's and provides tighter bounds than Lee bounds, as it imposes additional assumptions. Despite the growing popularity of partially identifying models, estimating identified sets and conducting inference remain challenging, particularly when the identified set is characterized by a large number of moment inequalities. Our model relaxes the linear selection assumption in HH’s model, providing a middle ground between Lee’s and HH’s approaches.

Nonlinearity in the selection process has been considered as a means to identify the parameters in the latent outcome equation. \cite{ahn1993semiparametric} observed that the nonlinear terms in the selection equation may act as excluded variables, and \cite{newey1993efficiency} derived semiparametric efficiency bounds for sample-selection models of this form, implying that nonlinearity of the selection probability can yield a nonzero bound without an excluded variable. However, this insight has been overlooked in the empirical literature. More recently, EJL introduced a more general model, assuming $E[Y\mid X]=F_0(X\beta_0,p_0(X))$, and demonstrated that $\beta_0$ can be identified up to scale if $X$ contains at least two continuous variables and $p_0(X)$ is nonlinear in $X$. We apply their approach to a simpler set-up, thereby obtaining weaker conditions for identifying the model parameters. Our approach operates on a single continuous covariate without scale normalization. It does not require joint optimization of the parametric and nonparametric components (e.g., \cite{ichimura1993semiparametric}) as in EJL.

\cite{pan2022semiparametric} propose an integrated nonlinear least squares estimator \citep{chen2010integrated, chen2010non, chen2011semiparametric, chen2012semiparametric} for EJL's model, eliminating the need for scale normalization. However, they do not provide an identification argument. Moreover, their estimation procedure relies on multiple layers of kernel regression, each requiring the selection of multiple tuning parameters, along with numerical optimization of an integrated criterion function. This results in a computational burden that scales up quickly with the sample size. An independently developed paper, \cite{pan2024locally}, employs a similar identification strategy but substantially differs from ours in terms of estimation. They employ \cite{robinson1988root}'s partialling-out estimator, in which the first-stage selection probability is estimated via a machine learning approach. They further augment their estimator by an orthogonalized moment and cross-fitting to guard against first-step bias, following \cite{chernozhukov2018double}'s Double/Debiased Machine Learning framework. Our sieve plug-in approach offers practitioners a more user-friendly alternative, implementable in R and Stata without specialized ML tooling.

This paper is organized as follows. Section~\ref{sec: model} introduces our semiparametric selection model and establishes identification. Section~\ref{sec: estimator} presents our estimators and derives their asymptotic properties. Section~\ref{sec: simulations} compares the finite-sample performance of our estimator with that of alternative estimators using simulations. Section~\ref{sec: application} applies our method to estimating gender and racial wage disparities in the US. Section~\ref{sec: conclusion} concludes. The proofs of the identification results are collected in the Appendix. The remaining proofs, additional simulation results, and full regression results in the empirical application are provided in the Online Appendix.

\section{The semiparametric selection model}\label{sec: model}

We consider a sample selection model where the selection procedure is left unspecified:
\begin{equation}\label{eq: semi-np selection model}
Y^* = \alpha_0 + X\beta_0 + V,\quad Y = D\cdot Y^*.
\end{equation}
Let $p_0(x) := P[D=1\mid X=x]$ represent the conditional selection probability given $X=x$. We impose the following assumption to identify $\beta_0$.

\begin{assumption}\label{ass: identification}
    (i) At least one variable ($X_1$ without loss of generality) in $X$ is continuously distributed; (ii) $p_0(x)$ is continuously differentiable with respect to $x_1$ almost everywhere on the support of $X$; (iii) $\partial p_0/\partial x_1 \neq 0$ with probability 1; (iv) $E[Y|X, D=1] = X\beta_0 + \lambda_0(p_0(X))$; (v) $\lambda_0(p)$ is continuously differentiable almost everywhere; (vi) The variables in $X$ are not perfectly multicollinear.
\end{assumption}

The above assumption requires a continuous variable in $X$ and smoothness of $p_0(\cdot)$ and $\lambda_0(\cdot)$. It also restricts the selection bias to depend only on the selection probability through an unknown function, $\lambda_0(\cdot)$. Additionally, it rules out a flat region in $p_0(\cdot)$. The intercept $\alpha_0$ is normalized to 0, as we cannot separately identify it from the selection bias.\footnote{For point identification of $\alpha_0$, see \cite{heckman1990varieties} and \cite{andrews1998semiparametric}. Unlike our paper, both papers rely on the ``identification at infinity'' argument and an exclusion restriction.} Assumption \ref{ass: identification}(iv) is a high-level condition under which the selection probability is a sufficient index to control for selection bias.\footnote{This assumption is satisfied in the following special case:
$D = \mathbbm{1}[h_0(X) \ge U]$, where $U$ is normalized to $\operatorname{Unif}(0,1)$, $X \ci U$, and $E[V|X,U] = E[V|U]$. Now consider a more general non-separable selection process: $D=\mathbbm{1}[m(X,U) \ge 0]$. If $m$ is strictly monotone in $U$, we can rewrite this process as
$D=\mathbbm{1}[U \ge \tau(X)]$, where $\tau(X):= m^{-1}(X,0)$,
which is observationally isomorphic to additive threshold crossing models, and therefore Assumption~\ref{ass: identification}(iv) holds. Without additional structure, however, such index sufficiency does not generally hold in a non-monotone selection process \citep{vytlacil2002independence}. In an arbitrary non-separable model where $U$ is not necessarily scalar, let $R(x):=\{u: m(x,u) \ge 0\}$. Then Assumption~\ref{ass: identification}(iv) holds if and only if, for almost every $x,x'$, $p_0(x) = p_0(x')$ implies $E[V|U \in R(x)] = E[V|U \in R(x')]$. This condition implies that, while not impossible to accommodate non-monotone selection, Assumption~\ref{ass: identification}(iv) admits restrictive forms of selection-relevant heterogeneity.} For some $\beta$ and $\lambda(\cdot)$, define $l(x) := x(\beta_0 - \beta)$ and $b(p) := \lambda_0(p) - \lambda(p)$, both of which represent deviations from the truth. For any observationally equivalent $\beta$ and $\lambda$ such that $E[Y|X, D=1] = X\beta + \lambda(p_0(X))$, $l(x) + b(p) = 0$ identically. Under Assumption~\ref{ass: identification}, $\beta_0$ can be point-identified (by showing that $l(x)=0$ and $b(p)=0$ identically). First, we consider the simplest case in which $X$ consists of only one continuous variable.

\begin{prop}\label{prop: identification - single continuous variable}
    Let $X$ be a scalar, nondegenerate, continuously distributed random variable. Let Assumption~\ref{ass: identification}(iv) hold. If there exist two distinct values $x'$ and $x''$ in the support of $X$ such that $p_0(x') = p_0(x'') = p'$, $\beta_0$ and $\lambda_0$ are identified.
\end{prop}

This proposition shows that more than nonlinearity is required in the selection process, as it rules out monotonicity of $p_0(\cdot)$. The identification argument also works for discrete $X$.\footnote{For instance, when $X$ is binary, $p_0(X) = \gamma_0 +\gamma_1 X$ is fully nonparametric. Identification of $\beta_0$ is obtained only when $p_0(0)=p_0(1)$, implying that $X$ is an excluded variable, as $\gamma_1=0$. Therefore, $\beta_0$ is not point-identified in this case unless $X$ is excluded from the selection equation, as shown in HH.} Now we consider more general cases where $X$ is multidimensional. Suppose two elements of $X$, $X_1$ and $X_2$, are continuous.

\begin{prop}\label{prop: identification - two continuous variables}
Let Assumption \ref{ass: identification} hold. Assume there exists another continuously distributed element $X_2$ of $X$ such that $p_0(\cdot)$ is continuously differentiable with respect to $x_2$ and $\partial p_0/\partial x_2 \neq 0$ with probability 1. If the ratio $(\partial p_0/\partial x_1)/(\partial p_0/\partial x_2)$ is not constant on the support of $X$, $\beta_0$ and $\lambda_0$ are identified.
\end{prop}

In general, the marginal effect of $X_1$ on $p$ is not proportional to that of $X_2$. Identification fails when $p_0(X) = F(X\gamma)$. Therefore, the nonlinearity of $p_0(\cdot)$ enables us to identify the model parameters. Lastly, we consider the case where $X_1$ is the only continuous variable in $X$. Without loss of generality, let $X_j$ be a binary variable for all $j \neq 1$, as any discrete variable can be equivalently expressed as a set of dummy variables.

\begin{prop}\label{prop: identification - one continuous variable and discrete variables}
    Let Assumption \ref{ass: identification} hold. For some $X_j$, let $x'$ be a vector where $x_j = 1$, and let $x''$ denote an otherwise identical vector where $x_j = 0$. Assume that for almost all $x_1'$ in the support of $X_1$, there exists $x_1'''$ in the support of $X_1$ such that $p_0(x''') = p_0(x'')$, where $x'''$ is otherwise identical to $x'$ except in its first component. If the mapping $x_1' \mapsto x_1'''$ is nonlinear, i.e., $(x_1' - x_1''')$ varies with $x_1'$, then $\beta_0$ and $\lambda_0$ are identified.
\end{prop}

\noindent Unless the required change in $x_1$ to achieve $p''$ with $x_j = 0$ from $p'$ with $x_j = 1$ is constant for all values of $x_1'$, the model parameters are identified. This proposition also rules out the case in which $p_0(X) = F(X\gamma)$.

The linear outcome specification we consider here is of primary interest in many applied studies because it is easy to interpret and provides an excellent first-order approximation to the data.\footnote{Classic examples include the Mincer wage equation and its extensions, which relate log earnings to schooling and experience \citep{mincer1974schooling, murphy1990empirical, murphy1992structure, card1999causal, lemieux2006mincer}, and health studies that regress log expenditures on age, income, insurance, and health status \citep{manning2001estimating}.} By contrast, selection processes are typically far less transparent. Entry decisions depend on utility maximization, reservation wages, market frictions, and institutional thresholds. Hence the selection process can exhibit complex nonlinearity while the outcome equation is approximately linear. Researchers can enrich the outcome equation with quadratic or cubic terms and interactions for a better approximation. As long as the nonlinearity that drives selection is not absorbed by outcome regressors, our identification results continue to hold.

\begin{remark}\normalfont
    (Parameter heterogeneity) In applied studies, heterogeneous treatment effects are often of interest. It is straightforward to extend our results to allow for parameter heterogeneity on the policy variable of interest. We can modify the model \eqref{eq: semi-np selection model} by partitioning the covariates into a scalar policy variable $X_i$ and control variables $W_i$:
    \begin{equation}\label{eq: parameter heterogeneity}
    Y^*_i = W_i\gamma_0 + X_i\beta_i + V_i, \quad Y_i = D_i \cdot Y^*_i, \quad p_i = P[D_i=1|X_i, W_i],
    \end{equation}
    where $\beta_i$ is individual-specific (we omit the intercept for expositional purposes). Then,
    \[
    E[Y_i|D_i=1, X_i, W_i] = W_i\gamma_0 + X_i E[\beta_i|D_i=1, X_i, W_i] + \lambda^*(p_i).
    \]
    Suppose that $\beta_i = \beta_0 + \varepsilon_{\beta i}$ where $E[\varepsilon_{\beta i}]=0$, and $E[\varepsilon_{\beta i}|D_i=1, X_i, W_i] = \omega(p_i)$, so that the conditional mean of $\varepsilon_{\beta i}$ depends on the covariates only through the selection probability upon selection. Then,
    \begin{equation}\label{eq: varying coefficient}
    E[Y_i|D_i=1, X_i, W_i] = W_i\gamma_0 + X_i\beta_0 + X_i\omega(p_i) + \lambda^*(p_i),
    \end{equation}
    which is a varying coefficient model. With a suitable normalization such as $\omega(p^*)=0$ at some $p^*$, the parameter $\beta_0 = E[\beta_i]$ is separately identified from $\omega(\cdot)$.\footnote{Identification of the model requires that $(1, X_i)$ has rank 2 on the level sets $\{(x,w): p_0(x,w) = p\}$. This condition is generically satisfied as long as $X_i$ is not a deterministic function of the propensity score alone. Since $W_i$ also enters the selection process, fixing $p_0$ allows $X_i$ to vary freely, compensated by offsetting movements in $W_i$. Therefore, the varying coefficient model \eqref{eq: varying coefficient} is identified under essentially the same conditions as the baseline model \eqref{eq: semi-np selection model}.}
\end{remark}

The identification results in this section show that our semiparametric selection model can point-identify the model parameters without an excluded variable as long as there is at least one continuously distributed covariate. In applied economic studies, continuous variables such as age, income, and price are not uncommon. Hence, our semiparametric model is widely applicable to modern data sets. Furthermore, it is natural for the true selection process to exhibit some degree of nonlinearity. As the conditional selection probability is nonparametrically identified, it is simple to check for nonlinearity in the selection probability. When there is strong empirical evidence of selection nonlinearity, our model becomes a strong alternative to the HH model, as it is more robust and point-identifying.

\section{The estimators}\label{sec: estimator}

We now propose a class of tractable sieve two-step plug-in estimators. First suppose that $p_{0i}=p_{0}(X_i)=E[D_i\mid X_i]$ is known. Conditional on selection, the model is partially linear. Suppose we have an i.i.d.\ sample, $\{W_i\}_{i=1}^n$, where $W_i:=(Y_i, D_i, X_i, p_{0i})$. Approximate $\lambda_0$ by $\Lambda_n=\{b_K(\cdot)'\gamma:\gamma\in\mathbbm{R}^{K}\}$, where $b_K$ is a sieve basis of dimension $K=K_n\to\infty$. The sieve least-squares estimator is
\[
\left(\tilde{\beta}_{n},\tilde{\gamma}_{n}\right)
=\argmin_{\beta\in\mathbbm{R}^{\dim(\beta)},\,\gamma\in\mathbbm{R}^{K}}
\frac{1}{n}\sum_{i=1}^{n}
D_{i}\left(Y_{i}-X_{i}'\beta-b_K\left(p_{0i}\right)'\gamma\right)^{2},
\]
where $\tilde{\lambda}_{n}=b_K\left(\cdot\right)'\tilde{\gamma}_{n}$. For any
$p=(p_1,\ldots,p_n)$ with $p_i\in[0,1]$, define
\[
\begin{gathered}
D=\mathop{\rm diag}(D_1,\ldots,D_n),\quad
X=[X_1,\ldots,X_n]',\quad y=(Y_1,\ldots,Y_n)',\\
p_0=(p_{01},\ldots,p_{0n})',\quad
B_K(p)=D[b_K(p_1),\ldots,b_K(p_n)]',\\
Q_K(p)=B_K(p)\{B_K(p)'B_K(p)\}^{+}B_K(p)'.
\end{gathered}
\]
Here ${}^{+}$ denotes the Moore--Penrose inverse. Then the LS estimator of $\beta$ is
written in the following closed-form expression:
\[
\tilde{\beta}_{n}
=\left(\left(DX\right)'\left(I-Q_K\left(p_{0}\right)\right)\left(DX\right)/n\right)^{+}
\left(DX\right)'\left(I-Q_K\left(p_{0}\right)\right)\left(Dy\right)/n.
\]

Let $v_i:=X_i-E[X_i\mid p_{0i},D_i=1]$, the selected-sample residual,
$\tilde X_i:=D_iv_i$, and $A:=E[\tilde X_i'\tilde X_i]$.
\cite{donald1994series} proved that this estimator has the following asymptotic distribution under their regularity conditions:
\begin{equation}\label{eq: asymptotic distribution when p is known}
	\sqrt{n}(
	\tilde{\beta}_{n}-\beta_{0})\rightarrow_{d}N\left(0,A^{-1}E\left[\sigma_{0}^{2}\left(X_i,D_i\right)
		\tilde{X}_i'\tilde{X}_i\right]A^{-1}
		\right).
\end{equation}
If the error term is homoskedastic, so that $\sigma_{0}(X_i,D_i)$ is constant, the sieve LS estimator $\tilde{\beta}_n$ is semiparametrically efficient, as discussed in \cite{li2007nonparametric}. When the error term exhibits heteroskedasticity, efficient estimation can be achieved through the sieve generalized least squares (SGLS) estimator.

As $p_{0i}$ is never observed in practice, $\tilde{\beta}_n$ is an infeasible estimator. Suppose $p_{0i}$ is consistently estimated by an estimator $\hat{p}_n(X_i)$. Define $\hat{p}= (\hat{p}_n(X_1), \cdots, \hat{p}_n(X_n))'$. Replacing $p_{0}$ with $\hat{p}$ in $\tilde{\beta}_n$, we obtain the following feasible estimator:
\[
  \hat{\beta}_{n}
  =\left(\left(DX\right)'\left(I-Q_K\left(\hat{p}\right)\right)
   \left(DX\right)/n\right)^{+}\left(DX\right)'\left(I-Q_K\left(\hat{p}\right)\right)\left(Dy\right)/n.
\]
Equivalently, on the event that the residualized Gram matrix is nonsingular,
Frisch--Waugh--Lovell gives
$\hat\beta_n=\hat A^{-1}n^{-1}\sum_iD_i\hat v_i'Y_i$, where $\hat v_i$
is the residual from regressing $X_i$ on $b_K(\hat p_i)$ in the selected sample
and $\hat A:=n^{-1}\sum_iD_i\hat v_i'\hat v_i$.
We will show that $\hat{\beta}_n$ remains $\sqrt{n}$-consistent and asymptotically normal, but that its asymptotic variance contains a correction term reflecting the first-stage estimation of $p_0$.

We first introduce the H\"{o}lder class of functions. Let
$\left[m\right]$ be the largest nonnegative integer such that $\left[m\right]
<m$. A real-valued function $\lambda$ on $\left[0,1\right]$ is said to be in the H\"{o}lder space $\Lambda^{m}\left(\left[0,1\right]\right)$ if it is
$\left[m\right]$ times continuously differentiable on $\left[0,1\right]$ and its H\"{o}lder norm
\[
  \lVert\lambda\rVert_{\Lambda^{m}}:=\max_{\ell\leq\left[m\right]}\sup_{p}
        \left|\frac{\partial^{\ell}\lambda\left(p\right)}{\partial p^{\ell}}\right|
  +\sup_{p,p'}
    \left|\frac{\partial^{\left[m\right]}\lambda\left(p\right)}{\partial p^{\left[m\right]}}
    -\frac{\partial^{\left[m\right]}\lambda\left(p'\right)}{\partial p^{\left[m\right]}}\right|
    /\left|p-p'\right|^{m-\left[m\right]}
\]
is finite. We impose the following conditions under which the sieve LS estimator with known $p_{0i}$ has the asymptotic distribution \eqref{eq: asymptotic distribution when p is known}.
\begin{assumption}\label{assu:IID}
(i) $\left\{W_i\right\}_{i=1}^{n}$ are i.i.d.; (ii) the support of $X_i$, $\mathcal{X}$,
is compact; (iii) the density of $p_{0i}$ is bounded and bounded away from zero on its compact support in $\left[0,1\right]$.
\end{assumption}
\begin{assumption}\label{assu:hfunc}
(i) $\lambda_0\in\Lambda^{m}([0,1])$ with $m>1/2$, and $\lambda_0$ is
continuously differentiable with bounded Lipschitz derivative:
$|\lambda_0'(p)-\lambda_0'(q)|\leq\bar L_2|p-q|$ for all $p,q\in[0,1]$;
(ii) $\forall \lambda\in\Lambda^{m}\left(\left[0,1\right]\right),\exists \lambda_{n}\left(p;\gamma\right)
\in\Lambda_{n}$ such that $\lVert \lambda_{n}-\lambda\rVert_{\infty}=O\left(K^{-m}
\right)$, where $K=K_{n}\rightarrow\infty$ is the second-stage sieve dimension with $K\log K/n\rightarrow0$.
\end{assumption}
\begin{assumption}\label{assu:s2}
$\sigma_{0}^{2}\left(X_i,D_i\right):=E\left[D_i\left(Y_i-X_i'\beta_{0}-\lambda_{0}
\left(p_{0i}\right)\right)^{2}|X_i=x,D_i=1\right]$ is positive and bounded uniformly over $x \in \mathcal{X}$.
\end{assumption}

\begin{assumption}\label{assu:Dv*}
(i) $A$ is positive definite.
(ii) Each element of $E\left[D_iX_i\mid p_{0i},D_i=1\right]$ belongs to the
H\"{o}lder space $\Lambda^{m_j}\left([0,1]\right)$, where $m_j>1/2$ for
$j=1,\ldots,\dim(\beta)$. The second-stage sieve approximation errors obey
\begin{align*}
e_{\lambda}(K)&\rightarrow0,\qquad e_X(K)\rightarrow0,\qquad
\sqrt n\,e_{\lambda}(K)e_X(K)\rightarrow0,\\
e_{\lambda}(K)
&:=\inf_{\pi}
\left\lVert\lambda_0-b_K(\cdot)'\pi\right\rVert_{L^2(P_{p_0\mid D=1})},\\
e_X(K)
&:=\max_{1\le j\le\dim(\beta)}\inf_{\pi_j}
\left\lVert E[X_{ij}\mid p_{0i},D_i=1]
-b_K(\cdot)'\pi_j\right\rVert_{L^2(P_{p_0\mid D=1})}.
\end{align*}
\end{assumption}

\noindent Assumptions \ref{assu:IID} and \ref{assu:s2} are standard regularity conditions. Assumption \ref{assu:hfunc} imposes mild H\"{o}lder smoothness on $\lambda$ and approximability by the second-stage sieve. Assumption \ref{assu:Dv*}(i) is the standard identification condition that is satisfied when the assumptions in Propositions \ref{prop: identification - single continuous variable}-\ref{prop: identification - one continuous variable and discrete variables} are met. Assumption \ref{assu:Dv*}(ii) ensures the conditional mean of $X_i$ given $p_{0i}$ and $D_i=1$ is smoothly approximable.\footnote{Its product condition is a joint requirement on $\lambda_{0}$ and $E\left[X_i|p_{0i},D_i=1\right]$. A sufficient primitive is $\sqrt{n}\,K^{-\left(m+\underline{m}\right)}\rightarrow0$ with $\underline{m}:=\min_{j}m_{j}$, provided the second-stage spline order is high enough to deliver $e_{\lambda}(K)=O(K^{-m})$ and $e_{X}(K)=O(K^{-\underline{m}})$. }

The estimation of $p_0$ in the first stage is generically \emph{not} asymptotically negligible. We therefore treat the first stage as a genuinely nonparametric sieve estimator of the conditional mean $p_{0}\left(x\right)=E\left[D_{i}|X_{i}=x\right]$. Let $\psi_{L}\left(\cdot\right)=\left(\psi_{1},\ldots,\psi_{L}\right)'$ be a known sieve basis with dimension $L=L_{n}\rightarrow\infty$, and let
\[
\hat{p}\left(x\right)=\psi_{L}\left(x\right)'\hat{\pi},\qquad
\hat{\pi}=\Bigl(\textstyle\sum_{i}\psi_{L}\left(X_{i}\right)\psi_{L}\left(X_{i}\right)'\Bigr)^{-1}\textstyle\sum_{i}\psi_{L}\left(X_{i}\right)D_{i},
\]
be the sieve LS estimator of $p_{0}$. Define the Riesz-representer weight
\[
\alpha_{0}\left(x\right):=p_{0}\left(x\right)\,\lambda_{0}'\left(p_{0}\left(x\right)\right)\,v\left(x\right)',\qquad
v\left(x\right):=x-E\left[X_i\mid p_{0}\left(X_i\right)=p_{0}\left(x\right),D_i=1\right].
\]
Write $\zeta\left(L\right):=\sup_{x\in\mathcal{X}}\lVert\psi_{L}\left(x\right)\rVert_{e}$ and let
\[
\Delta_{p}\left(L\right):=\inf_{\pi}\lVert p_{0}-\psi_{L}'\pi\rVert_{L^{2}\left(F_{X}\right)},\qquad
\Delta_{\alpha}\left(L\right):=\inf_{\Pi}\lVert \alpha_{0}-\psi_{L}'\Pi\rVert_{L^{2}\left(F_{X}\right)}
\]
denote the $L^{2}$ sieve-approximation errors for $p_{0}$ and for the representer $\alpha_{0}$.

\begin{assumption}\label{assu:np-firststage}
(i) The eigenvalues of $E\left[\psi_{L}\left(X_i\right)\psi_{L}\left(X_i\right)'\right]$ are bounded away from zero and infinity uniformly in $L$, and the leverage condition $\zeta\left(L\right)^{2}L\log L/n\rightarrow0$ holds. (ii) $\Delta_{p}\left(L\right)\rightarrow0$, $\Delta_{\alpha}\left(L\right)\rightarrow0$, $\zeta(L)\Delta_{p}(L)\to0$, the $L^{2}$ rate $\lVert\hat{p}-p_{0}\rVert_{L^{2}}=o_{p}\left(n^{-1/4}\right)$, and $\lVert\hat{p}-p_{0}\rVert_{\infty}=o_{p}(1)$.
(iii) $\sqrt{n}\,\Delta_{p}\left(L\right)\,\Delta_{\alpha}\left(L\right)\rightarrow0$.
(iv) The map $p\mapsto E\left[X_i\mid p\left(X_i\right),D_i=1\right]$ is Lipschitz in $L^{2}\left(F_{X}\right)$ along the first-stage perturbations: there is $C<\infty$ with
$E\bigl[\lVert E[X_i\mid \hat{p}(X_i),D_i=1]-E[X_i\mid p_{0}(X_i),D_i=1]\rVert^{2}\bigr]\le C\,\lVert\hat{p}-p_{0}\rVert_{L^{2}}^{2}+o_{p}(n^{-1})$.
\end{assumption}

\noindent Assumption~\ref{assu:np-firststage}(i) is a standard regularity condition.\footnote{The leverage is mildly stronger than the standard invertibility requirement $\zeta(L)^{2}\log L/n\to0$ \citep{newey1997convergence,belloni2015some}, as the score channel of the Riesz step consumes an additional factor of $L$.} (ii) is the first-stage rate condition, and it is substantially weaker than the parametric-rate first stage required in \cite{newey2009two}.\footnote{Only an $n^{-1/4}$ rate in $L^{2}$ is needed, and not a uniform one, because the quadratic remainder is controlled in $L^{2}$ and is the only channel that consumes an estimation rate. The first-order effect of the first stage is not orthogonalized away but priced, entering the asymptotic variance rather than imposing a faster rate on $\hat{p}$.} (iii) is the key undersmoothing condition: only the \emph{product} of the two approximation errors must be $o(n^{-1/2})$, far weaker than the $\sqrt{n}\,\Delta_{p}\to0$ that would be needed if $\hat{p}$ itself had to be $\sqrt{n}$-consistent. (iv) is the no-exclusion-specific condition: convergence of $\hat{p}$ does not by itself imply that the partialling-out residual $\hat{v}_{i}$ converges to $v_{i}$, because $E[X_i\mid p,D_i=1]$ is taken along the level sets of the index.\footnote{A sufficient primitive condition is that $p_{0}$ has no interior critical point on the (possibly trimmed) support, with $\lvert\nabla p_{0}\rvert$ bounded away from zero on each monotone branch; the level sets $\{x:p_{0}(x)=p\}$ then vary smoothly in $p$ and $\hat{v}_{i}\rightarrow v_{i}$.}

The smoothness conditions below are imposed along an \emph{admissible neighbourhood} of $p_{0}$ rather than the raw sup-norm ball.\footnote{The latter contains perturbations of $p_{0}$ that are constant on sets of positive probability, for which $p(X_i)$ has atoms, so density-based conditions quantified over it would be vacuous.} Write $f_{p}$ for the conditional density of $p(X_i)$ given $D_i=1$, required to exist on an interval support with the bound below holding on its interior, and $m_{Y,p}$, $m_{X,p}$ for the conditional means $E[Y_i\mid p(X_i)=\cdot,D_i=1]$ and $E[X_i\mid p(X_i)=\cdot,D_i=1]$. With $\Psi_{L}$ the linear span of $\psi_{L}$ and constants $\varepsilon>0$, $0<\underline{f}\leq\bar{f}<\infty$, $C_{m}<\infty$ fixed, the admissible class is
\[
\mathcal{P}_{n}:=\bigl\{p\in\Psi_{L}\cup\{p_{0}\}:\ \lVert p-p_{0}\rVert_{\infty}\leq\varepsilon,\ f_{p}\in[\underline{f},\bar{f}],\ \lVert m_{Y,p}\rVert_{\Lambda^{m}}\vee\lVert m_{X,p}\rVert_{\Lambda^{m}}\leq C_{m}\bigr\}.
\]
Coefficient vectors are bounded automatically.\footnote{$p=\psi_{L}'\pi\in\mathcal{P}_{n}$ implies $\lVert\pi\rVert^{2}\leq\underline{\lambda}_{G}^{-1}\lVert p\rVert_{L^{2}}^{2}\leq\underline{\lambda}_{G}^{-1}(\lVert p_{0}\rVert_{\infty}+\varepsilon)^{2}$ by Assumption~\ref{assu:np-firststage}(i).} We impose the following smoothness and regularity conditions.

\begin{assumption}\label{assu:diff} (i) $E\left[Y_i|p_{0}(X_i)=\cdot,D_i=1\right]$ and $E\left[X_i|p_{0}(X_i)=\cdot,D_i=1\right]$ are twice continuously differentiable; (ii) there is $C>0$ such that for all $p\in \mathcal{P}_{n}$,
\[
\begin{split}
E\Bigl[\bigl|E[X_i|p(X_i),D_i=1]-E[X_i|p_{0}(X_i),D_i=1]\bigr|_{e}^{2}&\\
{}+\bigl(E[Y_i|p(X_i),D_i=1]-E[Y_i|p_{0}(X_i),D_i=1]\bigr)^{2}\Bigr]
   &\leq C\lVert p-p_{0}\rVert_{L^{2}}^{2};
\end{split}
\]
(iii) for some $\varepsilon>0$, $P[D_i=1|p_{0}(X_i)=\cdot]>\varepsilon$.
\end{assumption}

\begin{assumption}\label{assu:p}
(i) $p_{0}\in\mathcal{P}_{n}$ for all $n$ large: $p_{0}(X_i)$ is continuously distributed given $D_i=1$ with interval support and conditional density in $[\underline{f},\bar{f}]$ on its interior, and the conditional means $E[Y_i\mid p_{0}(X_i)=\cdot,D_i=1]$, $E[X_i\mid p_{0}(X_i)=\cdot,D_i=1]$ have $\Lambda^{m}$ norms at most $C_{m}$. (ii) The fitted first stage is admissible with probability tending to one: $P\left[\hat{p}\in\mathcal{P}_{n}\right]\rightarrow1$. (iii) The second-stage Gram is uniformly nondegenerate over the class:
\begin{align*}
\inf_{p\in\mathcal{P}_{n}}\lambda_{\min}\bigl(E[D_ib_K(p(X_i))b_K(p(X_i))']\bigr)\geq &\ c_{G}>0\\
\sup_{p\in\mathcal{P}_{n}}\lambda_{\max}\bigl(E[D_ib_K(p(X_i))b_K(p(X_i))']\bigr)\leq &\ \bar{c}_{G}<\infty
\end{align*}
\end{assumption}

Assumption \ref{assu:hfunc}(i) supplies the smoothness of the selection-correction
function. Assumption \ref{assu:diff}(i) imposes smoothness on the conditional
means of $Y_i$ and $X_i$ given the propensity, while part (ii) is an
$L^{2}$-Lipschitz condition along perturbations within the admissible class.
Together with Assumption \ref{assu:np-firststage}(iv), it ensures the
partialling-out residual is stable. Assumption \ref{assu:p}(i) regularizes the
density of the generated regressor at $p_{0}$, as is standard in series
estimation \citep[Assumption 4.4 of][]{newey2009two}. Assumption
\ref{assu:p}(ii), membership of the fitted first stage, is a high-level
condition that accommodates a fully data-driven index.\footnote{A primitive
sufficient condition: if $p_{0}$ has no interior critical point on the
estimation support, with $\lvert\nabla p_{0}\rvert$ bounded away from zero on
each monotone branch, and the first stage is derivative-consistent on those
branches, then $\hat{p}\in\mathcal{P}_{n}$ with probability tending to one.}
Assumption~\ref{assu:p}(iii) is the uniform-over-class version of the
second-stage Gram condition of Assumption~\ref{assu:secondstage}(i), which it
implies at $p=p_{0}$.

Recall $v_i$, $\tilde X_i$, and $A$ defined above, and let
\[
\Omega:=E\left[\sigma_{0}^{2}\left(X_i,D_i\right)\tilde{X}_i'\tilde{X}_i\right],\quad
\chi_{i}:=p_{0i}\,\lambda_{0}'\left(p_{0i}\right)v_{i}'\left(D_{i}-p_{0i}\right),
\]
where $\varepsilon_{i}:=Y_i-X_i'\beta_{0}-\lambda_{0}\left(p_{0i}\right)$ and the first-stage correction matrix is $\Omega_{\chi}:=E\left[\chi_{i}\chi_{i}'\right]=E\bigl[p_{0i}^{3}\left(1-p_{0i}\right)\lambda_{0}'\left(p_{0i}\right)^{2}v_{i}'v_{i}\bigr]$, finite under Assumption~\ref{assu:np-hl}(i) (indeed $\Omega_{\chi}\preceq\tfrac14\lVert\alpha_{0}\rVert_{L^{2}}^{2}I$). We state the asymptotic result in two layers. We first give high-level conditions on the two-step estimator under which the feasible estimator has the corrected limiting distribution; we then verify those conditions from the primitive growing-sieve and second-stage assumptions above. Let $\hat{v}_{i}$ be the residual from regressing $X_i$ on the second-stage spline basis $b_K(\hat{p}_i)$ within the selected sample.

\begin{assumption}\label{assu:np-hl}
(i) $\alpha_{0}\in L^{2}\left(F_{X}\right)$ and $A$ is nonsingular.
\begin{align*}
\text{(ii)}\quad
\sqrt n\,\frac1n\sum_iD_i\hat v_i'\lambda_0'(p_{0i})(\hat p_i-p_{0i})
&=\frac1{\sqrt n}\sum_i\chi_i+o_p(1),\\
\text{(iii)}\quad
\sqrt n\,\frac1n\sum_iD_i\hat v_i'
\bigl\{\lambda_0(p_{0i})-\lambda_0(\hat p_i)
+\lambda_0'(p_{0i})(\hat p_i-p_{0i})\bigr\}
&=o_p(1).
\end{align*}
(iv) $\hat A:=n^{-1}\sum_iD_i\hat v_i'\hat v_i\to_p A$, and the oracle term
$n^{-1/2}\sum_iD_i(\hat v_i-v_i)'\varepsilon_i=o_p(1)$. The second-stage
spline-approximation and empirical-projection contributions to the $X$-moment
are $o_p(n^{-1/2})$.
\end{assumption}

\begin{prop}[Asymptotics under high-level conditions]\label{prop: asymptotics np}
    Let Assumptions \ref{assu:IID}(i)--(ii), \ref{assu:s2}, and
    \ref{assu:np-hl} hold. Then
$\sqrt{n}\left(\hat{\beta}_{n}-\beta_{0}\right)=A^{-1}\frac{1}{\sqrt{n}}\sum_{i=1}^{n}\left\{\tilde{X}_i'\varepsilon_{i}-\chi_{i}\right\}+o_{p}\left(1\right)$, and hence $\sqrt{n}(\hat{\beta}_{n}-\beta_{0})\rightarrow_{d}N\left(0,\,V_{\mathrm{NP}}\right)$, where $V_{\mathrm{NP}}=A^{-1}\left(\Omega+\Omega_{\chi}\right)A^{-1}$.
\end{prop}
The two components of the influence function are uncorrelated, so no covariance term appears in $V_{\mathrm{NP}}$. The first component reproduces the oracle variance in \eqref{eq: asymptotic distribution when p is known}; the second, $A^{-1}\Omega_{\chi}A^{-1}$, is the nonparametric first-stage correction, the Riesz representer of the propensity functional in the sense of \cite{newey1994asymptotic}. When $\lambda_{0}'\equiv0$, that is, when there is no selection on unobservables, $\chi_{i}=0$ and $V_{\mathrm{NP}}$ collapses to the oracle variance, so the standard (heteroskedasticity-robust) sandwich formula from the second-stage regression is valid; this restriction is testable, for example through the joint significance of the sieve terms of $\hat{p}$ in the second stage. The fixed-dimensional sieve-index first stage used in our implementation is a parallel implementable specification (Proposition~\ref{prop: asymptotics} in Section~\ref{subsec: practical implementations}), for which the correction takes the score--information form and the full argument is given in Section~S3 of the Online Appendix.

The primitive growing-sieve and second-stage conditions imply the high-level expansion.
\begin{assumption}[Second-stage sieve]\label{assu:secondstage}
		The selection-correction function $\lambda_{0}$ is approximated in the second stage by a cubic spline basis $b_K\left(\cdot\right)$ of dimension $K=K_{n}\rightarrow\infty$, with $\xi\left(K\right):=\sup_{t}\lVert b_K\left(t\right)\rVert_{e}$, satisfying: (i) $\lambda_{\min}\bigl(E[D_{i}b_K(p_{0i})b_K(p_{0i})']\bigr)$ is bounded away from zero; (ii) $\xi\left(K\right)^{2}K\log K/n\rightarrow0$; (iii) the undersmoothing rate $\sqrt{n}\,K^{-m}\rightarrow0$, where $m>1/2$ is the H\"{o}lder smoothness of $\lambda_{0}$ (Assumption~\ref{assu:hfunc}) and of the conditional means $E[X_i\mid\cdot,D_i=1]$, $E[Y_i\mid\cdot,D_i=1]$ (Assumption~\ref{assu:p}(i)); and (iv) the cross-rate condition $(K+L)\log n=o(\sqrt{n})$, that is, $K$ and $L$ grow more slowly than $\sqrt{n}$ up to logarithmic factors, which, together with the other rates of Assumptions~\ref{assu:np-firststage} and~\ref{assu:secondstage}, makes the second-stage cross-remainder $\sqrt{n}\,\lVert\hat{v}-v\rVert_{L^{2}}\,\lVert\hat{p}-p_{0}\rVert_{L^{2}}$ vanish, where $\lVert\hat{v}-v\rVert_{L^{2}}$ carries the uniform-concentration rate of Lemma~S1.2 of the Online Appendix; and (v) the generated-index second-stage regression is stable: $\lambda_{\min}\bigl(n^{-1}\sum_iD_ib_K(\hat p_i)b_K(\hat p_i)'\bigr)$ is bounded away from zero with probability tending to one and $\sup_{1\le i\le n}\lVert\hat{v}_i\rVert=O_{p}(1)$.
\end{assumption}
\begin{prop}[Primitive sufficient conditions]\label{prop: np primitives}
It suffices that\par\noindent
Assumptions \ref{assu:IID}, \ref{assu:hfunc}, \ref{assu:s2},
\ref{assu:Dv*}, \ref{assu:np-firststage}, \ref{assu:diff},
\ref{assu:p}, and \ref{assu:secondstage} hold;
$K$ and $L$ grow polynomially in $n$; and
$p_0\in[\underline c,\bar c]\subset(0,1)$ on the support.
Then Assumption~\ref{assu:np-hl} holds, so
Proposition~\ref{prop: asymptotics np} applies.
\end{prop}
\noindent Assumption~\ref{assu:secondstage}(iii) is stronger than the bias-variance-optimal rate $K\asymp n^{1/(2m+1)}$ admissible under Assumption~\ref{assu:hfunc}: with a nonparametrically generated regressor $\hat{p}$, the second-stage spline-approximation residual is not differenced away at the optimal rate. (iii) and~(iv) are jointly compatible only when $m>1$, since undersmoothing requires $K\gg n^{1/(2m)}$ while (iv) caps $K$ at $\sqrt{n}$ up to logarithms. (v) is the generated-regressor analogue of the usual series-regression stability condition; given the admissible class of Assumption~\ref{assu:p}, both of its parts are in fact implied.\footnote{The generated-index Gram conditioning follows from Lemma~S1.2(b) of the Online Appendix, and the uniform residual bound since $\sup_i\lVert\widehat{\Pi}_K X_i\rVert\le\lVert\Pi_K m_{X,\hat{p}}\rVert_{\infty}+\xi(K)\,O_p\bigl(\sqrt{(K+L)\log n/n}\bigr)=O_p(1)$ under Assumptions~\ref{assu:secondstage}(ii) and~\ref{assu:np-firststage}(i), by the uniformly bounded Lebesgue constant of the spline projection and the boundedness of $m_{X,\hat{p}}$ built into the admissible class. The $L^{4}$ route of Lemma~S1.4 can replace the uniform residual bound if one prefers not to invoke the class construction.} We retain (v) as a stated condition for transparency, since it is the primitive used to control the quadratic remainder in Lemma~S1.4 of the Online Appendix.

\subsection{Semiparametric efficiency}\label{subsec: efficiency}

Because $p_0$ is unknown, the relevant benchmark is the efficiency bound of the feasible model, not the oracle variance. We specialize \cite{newey1993efficiency} to the model $\mathcal P$ in which $p_0$ is unrestricted subject to overlap and density regularity, $E[Y_i\mid X_i,D_i=1]=X_i'\beta+\lambda(p_0(X_i))$, and the remaining selected-sample distribution is unrestricted. Let $\sigma_0^2(x)=\mathrm{Var}(Y_i\mid X_i=x,D_i=1)$ and define
\begin{align*}
	\tilde{u}_{i}&:=D_{i}\varepsilon_{i}
	-p_{0i}\lambda_{0}'(p_{0i})(D_{i}-p_{0i}),\\
	\kappa^{2}(X_i)&:=E[\tilde{u}_{i}^{2}\mid X_i]
	=p_{0i}\sigma_{0}^{2}(X_i)+p_{0i}^{3}(1-p_{0i})\lambda_{0}'(p_{0i})^{2}.
\end{align*}
\[
w_{0}(X_i):=\frac{p_{0i}}{\kappa^{2}(X_i)}=\frac{1}{\sigma_{0}^{2}(X_i)+p_{0i}^{2}(1-p_{0i})\lambda_{0}'(p_{0i})^{2}},
\]
together with $\tilde{E}[X\mid p_0]:=E[(p_0^2/\kappa^2)X\mid p_0]/E[p_0^2/\kappa^2\mid p_0]$ and $\tilde v_i:=X_i-\tilde E[X\mid p_0](p_{0i})$. The composite residual combines the outcome and selection channels with a relative weight pinned by the model.

\begin{assumption}[Efficiency regularity]\label{assu:eff}
	(i) Overlap and the density conditions of Assumption~\ref{assu:p} hold, $\sigma_0^2$ is bounded and bounded away from zero, and $E[(p_0^2/\kappa^2)\tilde v\tilde v']$ is nonsingular. (ii) The range of $w_0$ lies in a known compact interval $[\underline w,\bar w]\subset(0,\infty)$. (iii) Assumptions~\ref{assu:np-firststage}(iv), \ref{assu:diff}(i)--(ii), and~\ref{assu:secondstage} hold for the $w_0$-weighted conditional means and projections, uniformly over $\mathcal P_n$, and the approximation conditions of Assumption~\ref{assu:np-firststage}(ii)--(iii) hold for $\alpha_{w_0}:=p_0w_0\lambda_0'(p_0)\tilde v$. (iv) The fold-specific weight estimate uses outcomes only from its own fold and the full-sample $(X,D)$ data, is truncated to $[\underline w,\bar w]$, and satisfies $\lVert\hat w^{(k)}-w_0\rVert_\infty=o_p(n^{-1/4})$.
\end{assumption}

\begin{prop}[Efficiency]\label{prop: efficiency bound}
	Let Assumptions \ref{assu:IID}, \ref{assu:hfunc}(i), \ref{assu:Dv*}(i), and \ref{assu:eff}(i) hold for $\mathcal P$. Every regular asymptotically linear estimator has influence function $\tilde B(X_i)\tilde u_i$, where
	$E[p_0\tilde BX']=I$ and $E[p_0\tilde B\mid p_0]=0$. The efficient influence function is $\psi_i^*=V^*w_0(X_i)\tilde v_i\tilde u_i$, and the efficiency bound is
	\[
	V^{*}=\Bigl(E\Bigl[\frac{p_{0i}^{2}}{\kappa^{2}(X_i)}\,\tilde{v}_{i}\tilde{v}_{i}'\Bigr]\Bigr)^{-1}.
	\]
	(a) Moreover, $V_{\mathrm{NP}}\succeq V^*$, with equality if and only if
	\[
	(\kappa^2/p_0)v(X)=LX+\mu(p_0(X))\quad\text{a.s.}
	\]
	Constancy of $\kappa^2(X)/p_0(X)$ is sufficient and, when $\sigma_0^2$ depends only on $p_0$, necessary wherever $\mathrm{Var}(X\mid p_0)$ is nondegenerate. (b) Also $V^*\succeq V^\circ$, the known-$p_0$ bound, with equality under $\lambda_0'(p_0(X))=0$. The inequality is strict when $\mathrm{Var}(X\mid p_0)$ has full rank on a positive-measure subset of $\{\lambda_0'(p_0)\neq0\}$.
\end{prop}

The proposition shows that first-stage noise is intrinsic to the model: regular estimators may reweight, but cannot remove, the selection channel. Our unweighted estimator is efficient only under the equality condition in part (a). In particular, homoskedasticity does not generally restore efficiency because $\kappa^2/p_0=\sigma^2+p_0^2(1-p_0)\lambda_0'(p_0)^2$ still varies with $p_0$. The missing weighting is easily supplied. Split the sample into two folds $I_{1},I_{2}$ (of sizes $n_{1},n_{2}$ with $n_{k}/n\rightarrow1/2$); keep the full-sample first stage $\hat{p}$; on each fold $k$ estimate the weight function
\[
\hat{w}^{(k)}(\cdot):=T_{[\underline{w},\bar{w}]}\Bigl(\bigl\{\hat{\sigma}^{2}_{(k)}(\cdot)+\hat{p}^{2}(1-\hat{p})\hat{\lambda}'_{(k)}(\hat{p})^{2}\bigr\}^{-1}\Bigr),
\]
where $\hat{\lambda}'_{(k)}$ and $\hat{\sigma}^{2}_{(k)}$ come from a preliminary (unweighted) second stage on fold $I_{k}$ and $T$ truncates; then run the weighted second stage on each fold with the \emph{other} fold's weight, $\hat{\beta}^{(k)}$ from weighted least squares of $Y_i$ on $(X_i,b_K(\hat{p}_i))$ over $\{i\in I_{k}:D_i=1\}$ with weights $\hat{w}^{(-k)}(X_i)$, and average: $\hat{\beta}^{*}:=\sum_{k}(n_{k}/n)\hat{\beta}^{(k)}$.

\begin{prop}[Efficient weighted estimator]\label{prop: efficient estimator}
	Let the conditions of Proposition~\ref{prop: np primitives} and Assumption~\ref{assu:eff} hold. Then
	\[
	\sqrt{n}\bigl(\hat{\beta}^{*}-\beta_{0}\bigr)=\frac{1}{\sqrt{n}}\sum_{i=1}^{n}\psi^{*}_{i}+o_{p}(1)\rightarrow_{d}N\bigl(0,\,V^{*}\bigr),
	\]
	so $\hat{\beta}^{*}$ attains the semiparametric efficiency bound of Proposition~\ref{prop: efficiency bound}.
\end{prop}

\begin{remark}\normalfont\label{rem: efficient weights}
	(Implementation) The efficient weight requires only $\hat p$, the derivative of the fitted second-stage spline, and an estimate of $\sigma_0^2$. When $\sigma_0^2$ depends only on $p_0$, the variance estimate requires one additional spline regression of squared residuals on $\hat p$. A composite sandwich estimator, detailed in the Online Appendix, consistently estimates $V^*$ under $m>2$ and $K^4\log K/n\to0$. A misspecified limiting weight preserves consistency and asymptotic normality under the corresponding stability conditions, but loses efficiency.
\end{remark}

\subsection{Practical implementations}\label{subsec: practical implementations}
In implementation we estimate $p_{0}$ by a fixed-dimensional sieve \emph{index} model rather than by series least squares, because the index form is $\sqrt{n}$-asymptotically linear with a closed-form influence and is available in any statistical software. Let $\phi\left(x\right)=\left(\phi_{1}\left(x\right),\ldots,\phi_{d_{\gamma}}\left(x\right)\right)'$ be a vector of $d_{\gamma}$ known basis functions (including a constant), with $d_{\gamma}$ fixed, and let $F$ be a known link function with derivative $f$.

\begin{assumption}[Fixed-dimensional sieve index]\label{assu:convergence}
	(i) $p_{0}\left(x\right)=F\left(\phi\left(x\right)'\gamma_{0}\right)$ for some $\gamma_{0}\in\mathrm{int}\left(\Gamma\right)$ with $\Gamma\subset\mathbb{R}^{d_{\gamma}}$ compact, where $F$ is strictly increasing and twice continuously differentiable with bounded first and second derivatives; (ii) the first-stage estimator $\hat{\gamma}_{n}$ satisfies $\sqrt{n}\left(\hat{\gamma}_{n}-\gamma_{0}\right)=J^{-1}\frac{1}{\sqrt{n}}\sum_{i=1}^{n}s_{i}+o_{p}\left(1\right)$, where $s_{i}=s\left(D_{i},X_{i}\right)$ with $E\left[s_{i}\right]=0$ and $\Sigma_{s}:=E\left[s_{i}s_{i}'\right]$ finite and positive definite, and $J$ is nonsingular; (iii) there are estimators $\hat{J}\rightarrow_{p}J$ and $\hat{\Sigma}_{s}\rightarrow_{p}\Sigma_{s}$; (iv) $\lambda_{0}$ is $s$ times continuously differentiable on $\left[0,1\right]$ with bounded derivatives for some $s\geq3$, and $F$ is $s$ times continuously differentiable with bounded derivatives; (v) the second-stage spline dimension $K$ of Assumption~\ref{assu:hfunc} satisfies $K^{4}/n\rightarrow0$ and $\sqrt{n}\,K^{-s-1}\rightarrow0$, which strengthen the corresponding rates of Assumption~\ref{assu:hfunc} for this special case and match the spline conditions of \citet[Assumption~4.5]{newey2009two}.
\end{assumption}

\noindent Assumption~\ref{assu:convergence} is satisfied by the sieve probit (or logit) MLE with a fixed number of basis functions, in which case $s_{i}$ is the score, $J$ is the information matrix, and $V_{\gamma}:=J^{-1}\Sigma_{s}J^{-1}=J^{-1}$ under correct specification; more generally any $\sqrt{n}$-consistent asymptotically linear index estimator is admissible \citep{klein1993efficient}. Define $\dot{p}_{\gamma}\left(x\right):=\partial p_{\gamma}\left(x\right)/\partial\gamma=f\left(\phi\left(x\right)'\gamma\right)\phi\left(x\right)$ and $G:=E\left[\tilde{X}_i'\,\lambda_{0}'\left(p_{0i}\right)\dot{p}_{\gamma_{0}}\left(X_i\right)'\right]$. Under this special case the nonparametric correction $\chi_{i}$ of Proposition~\ref{prop: asymptotics np} takes a score--information form.

\begin{prop}[Fixed-dimensional first stage]\label{prop: asymptotics}
		Let Assumptions \ref{assu:IID}, \ref{assu:hfunc}, \ref{assu:s2}, \ref{assu:Dv*}(i), \ref{assu:diff}(i) and (iii), \ref{assu:p}(i), and \ref{assu:convergence} hold. Then,
	\[
	\sqrt{n}\left(\hat{\beta}_{n}-\beta_{0}\right)=A^{-1}\frac{1}{\sqrt{n}}\sum_{i=1}^{n}\left\{\tilde{X}_i'\varepsilon_{i}-GJ^{-1}s_{i}\right\}+o_{p}\left(1\right),
	\]
	and consequently
	\begin{equation}\label{eq: asymptotic distribution when p is estimated}
		\sqrt{n}(\hat{\beta}_{n}-\beta_{0})\rightarrow_{d}N\left(0,\,V\right),\quad V=A^{-1}\left(\Omega+GV_{\gamma}G'\right)A^{-1},\quad V_{\gamma}:=J^{-1}\Sigma_{s}J^{-1}.
	\end{equation}
\end{prop}
\noindent By the chain rule, $G$ is the matrix $H$ of \cite{newey2009two}, so Proposition~\ref{prop: asymptotics} specializes his Theorem~4.1; Section~S3 of the Online Appendix verifies its conditions. Along a growing first-stage sieve, the score--information correction tracks the Riesz correction under \cite{ackerberg2012practical}. At fixed dimension, however, the reported standard errors estimate $V$ in \eqref{eq: asymptotic distribution when p is estimated}, not literally $V_{\mathrm{NP}}$.

In practical implementation, we propose the following simple two-step procedure.
\setcounter{bean}{0}
\begin{list}{\textsc{Step} \arabic{bean}.}{\usecounter{bean}\leftmargin=4em}
		\item Using the full sample, estimate $p_0$ by the sieve MLE, i.e., a probit (or logit) regression of $D_i$ on the basis $\phi(X_i)$, and obtain the fitted values $\hat{p}$, the score vectors $\hat{s}_i$, and the information matrix $\hat{J}$.
		\item Conditional on $D_i=1$, estimate $\beta_0$ and $\lambda_0$ using the SLS estimator, and compute the corrected standard errors from \eqref{eq: asymptotic distribution when p is estimated} as described below.
\end{list}
In the simulations and application, Step 1 uses a sieve probit with cubic regression splines and Step 2 uses cubic B-splines in $\hat p$. Thus these exercises implement Proposition~\ref{prop: asymptotics}. Its corrected variance is estimated by
\begin{equation}\label{eq: corrected variance estimator}
	\hat{V}=\hat{A}^{-1}\left(\frac{1}{n}\sum_{i=1}^{n}\hat{\omega}_{i}\hat{\omega}_{i}'\right)\hat{A}^{-1},\qquad
	\hat{\omega}_{i}:=D_{i}\hat{v}_{i}'\hat{\varepsilon}_{i}-\hat{G}\hat{J}^{-1}\hat{s}_{i},
\end{equation}
where $\hat{v}_{i}$ is the residual from regressing $X_i$ on the second-stage spline basis of $\hat{p}_i$ within the selected sample (the partialling-out residual that standard partial linear routines already compute), $\hat{\varepsilon}_{i}$ is the second-stage residual, $\hat{A}=n^{-1}\sum_{i}D_{i}\hat{v}_{i}'\hat{v}_{i}$, $\hat{\lambda}'\left(\cdot\right)$ is the derivative of the fitted second-stage spline, and $\hat{G}=\frac{1}{n}\sum_{i=1}^{n}D_{i}\hat{v}_{i}'\,\hat{\lambda}'\left(\hat{p}_{i}\right)f\left(\phi\left(X_{i}\right)'\hat{\gamma}_{n}\right)\phi\left(X_{i}\right)'$. The standard error of $\hat\beta_j$ is $(\hat V_{jj}/n)^{1/2}$. All ingredients are standard outputs of the two regressions. Omitting $\hat G\hat J^{-1}\hat s_i$ yields the robust second-stage sandwich; it ignores first-stage uncertainty and is valid when $G=0$, in particular under $\lambda_0'\equiv0$.

\begin{remark}\normalfont\label{rem: nonlinearity test}
    (Nonlinearity test) Nonlinearity in the first stage can be tested empirically. When a sieve MLE is employed in the first stage, let $D_i = \mathbbm{1}\bigl[\,g(X_i) + U_i > 0 \bigr]$, where $g(X_i)$ is approximated by $\sum_{k=1}^{d_{\gamma}} \gamma_k \,\phi_k(X_i)$. Significant coefficients of high-order sieve terms imply nonlinearity. Alternatively, one can also consider linear index specifications such as probit and logit. Let $l_{lin}$ and $l_{sieve}$ denote maximized log likelihoods of linear index and sieve-based models, respectively. Under the null hypothesis $H_0: g(X) = X'\beta$, the distribution of $l_{LR} = 2(l_{sieve}-l_{lin})$ is asymptotically $\chi^2_{d_{\gamma}-d_X}$ (with $d_{\gamma}$ fixed; the degrees of freedom are the difference between the unrestricted and restricted parameter counts). Reject $H_0$ (and thus reject linearity) if $l_{LR} \;>\; \chi^2_{d_{\gamma}-d_X,\,\alpha}$, where $\chi^2_{d_{\gamma}-d_X,\,\alpha}$ is the critical value of the chi-square distribution with $d_{\gamma}-d_X$ degrees of freedom at the significance level $\alpha$.
\end{remark}


\section{Simulations}\label{sec: simulations}

We evaluate the finite-sample performance of our semiparametric estimator using known data-generating processes (DGPs). For each DGP, we repeat
1,000 iterations, in each of which we draw a Monte Carlo sample of size $n = 5,000$.
We consider three designs: a single-covariate design (DGP0), in which identification
of $\beta_1$ turns on whether the selection index is monotone, and two designs with
two covariates (DGP1 and DGP2), in which $p_0(\cdot)$ exhibits sufficient
nonlinearity that $\beta_1$ and $\beta_2$ are point identified. We employ our
semiparametric estimator (denoted KL). For alternative estimators, we consider
the two-part model (referred to as TPM), which is the ordinary least squares (OLS)
estimator conditional on $D=1$, assuming random selection, and Heckman's MLE
(denoted as HSM). We also compare our estimator with the
oracle estimator, which incorporates the true functional forms of $p_0(X)$ and the
selection bias. The full design specifications, the single-covariate
design, the box plots for every design, the comparison with the \cite{lee2009bounds}
and \cite{honore2020selection} bounds, and designs with weak nonlinearity are
reported in Section~S4 of the Online Appendix.

Table~\ref{tab: finite sample performance - DGP1 and 2} displays the RMSE, mean bias, and the coverage probability of the 95\% confidence interval of each estimator. In both DGPs, TPM and Heckman's MLE are misspecified and hence suffer from large biases. In contrast, our semiparametric estimator performs exceptionally well in DGP1 for both parameters, and the oracle estimator outperforms it by a very slight margin. In DGP2, our estimator performs similarly to the oracle estimator for $\beta_1$, but shows a larger RMSE than the oracle estimator for $\beta_2$, possibly due to limited variation in the discrete covariate. The first-stage correction is not negligible, as the theory suggests. With uncorrected robust standard errors, the coverage of our estimator falls below the nominal level, while the corrected standard errors bring it to approximately 95\%.

\begin{table}[ht!]
	\centering\footnotesize
	\setlength{\tabcolsep}{4pt}
		\caption{Finite-sample performance of estimators}\label{tab: finite sample performance - DGP1 and 2}
	\begin{tabular}{lccccccccc}
		\toprule
		& & \multicolumn{4}{c}{DGP1} & \multicolumn{4}{c}{DGP2} \\
		\cmidrule(lr){3-6} \cmidrule(lr){7-10}
		& & TPM & HSM & KL & Oracle & TPM & HSM & KL & Oracle \\
		\midrule
		\multirow{2}{*}{RMSE} &$\beta_1$ &  0.255 & 0.100 &  0.060 &  0.045 &  0.153 &  0.224 &  0.049 &  0.047 \\
		& $\beta_2$ &  0.321 & 0.065 &  0.063 &  0.051 &  0.393 &  0.383 &  0.113 &  0.085 \\
		\multirow{2}{*}{Bias} & $\beta_1$ & -0.252 & 0.088 & -0.011 & -0.003 & -0.148 & -0.078 & -0.001 & -0.003 \\
		& $\beta_2$ & -0.318 & 0.038 & -0.010 & -0.003 & -0.388 & -0.318 & -0.019 & -0.004 \\
		\multirow{2}{*}{Coverage (uncorrected)} & $\beta_1$ & 0.000 & 0.566 & 0.932 & 0.939 &  0.040 & 0.186 & 0.929 & 0.921  \\
		& $\beta_2$ & 0.000 & 0.873 & 0.927 & 0.947 & 0.000 & 0.351 & 0.922 & 0.943 \\
		\multirow{2}{*}{Coverage (corrected)} & $\beta_1$ & -- & -- & 0.952 & -- & -- & -- & 0.956 & -- \\
		& $\beta_2$ & -- & -- & 0.936 & -- & -- & -- & 0.947 & -- \\
		\bottomrule
	\end{tabular}
	\vspace{0.5em}
	\par\footnotesize\renewcommand{\baselineskip}{11pt}\justifying
	\textbf{Note:} TPM, HSM, KL, and Oracle denote OLS under random selection, Heckman’s MLE, our semiparametric sieve estimator, and the oracle estimator, respectively. `Coverage' rows report the coverage probability of 95\% confidence intervals. For the KL estimator, `uncorrected' uses the robust standard errors of the second-stage regression, which treat $\hat{p}$ as known, and `corrected' uses the first-stage-corrected standard errors. TPM, HSM, and Oracle use their conventional standard errors.
\end{table}

These simulation exercises demonstrate the practical usefulness of our semiparametric estimator. When at least one continuous covariate is present and the selection process exhibits nonlinearity, our estimator performs exceptionally well even with a modest sample size. The first-stage selection probability is nonparametrically identified, so assessing nonlinearity in the selection equation is practically easy. In contrast to the HH model, our semiparametric model offers greater flexibility by not imposing linearity in the first stage while still point-identifying the parameters of interest. Consequently, our estimator can serve as a valuable alternative when the Lee bounds are excessively wide to provide meaningful insights. However, if there is no continuous variable or the selection process is genuinely linear, the HH bounds would be an excellent alternative to the Lee bounds.

\section{Application: gender and racial wage gaps in the US}\label{sec: application}

We now demonstrate the empirical usefulness of our semiparametric estimator by estimating the gender and racial wage disparities in the US. The reservation wage varies across gender and ethnic groups. Upon selection into employment, the distribution of unobserved factors can differ from that of the unemployed. Therefore, the effect of sample selection on observed wages should be taken into account to accurately calculate wage gaps. Following \cite{mora2008nonparametric} and \cite{honore2020selection}, who focus only on racial wage gaps, we analyse Current Population Survey (CPS) data on wages from Arizona, California, New Mexico, and Texas.\footnote{The data set covers the years 2003--2016 and includes 127,738 women. Among them, 24,698 are third-generation Mexican-Americans, while 103,040 are non-Hispanic whites. The remaining 118,250 men comprise 21,402 third-generation Mexican-Americans and 96,848 non-Hispanic whites. All individuals in the sample are aged between 25 and 62. In terms of employment, the percentage of women working is 64\% for third-generation Mexican-Americans and 61\% for non-Hispanic whites. The employment rates for men are 71\% for Mexican-Americans and 67\% for non-Hispanic whites, respectively.} The gender wage gap is estimated for Mexican-Americans and non-Hispanic whites separately to nonparametrically control for ethnicity. We use the log of the inflation-adjusted hourly wage as the outcome. In the outcome equation, we estimate the coefficient on the female dummy with age, age squared, experience, experience squared, education dummies (less than high school, some college, college, and advanced degrees such as master's and doctoral degrees, with high school as the base category), dummies for being a veteran and being married, state dummies (New Mexico as a base state), and year dummies as control variables. Age and experience serve as continuously distributed covariates in the selection equation for our semiparametric model. For the racial wage gap, we estimate the coefficient on the Mexican-American dummy with the same set of control variables separately for men and women.

In Table \ref{tab: Lee bounds}, we present the estimated Lee bounds by education group (high school and college) with no additional covariates. Overall, estimated bounds are not very informative about wage gaps. Regarding racial wage disparity, the Lee bounds for college graduates contain a zero value. For high school graduates, the bounds imply a substantial to moderate racial gap for both men and women. In the context of the gender disparity, the Lee bounds suggest significantly lower wages for women in both racial groups.

\begin{table}[ht!]
\centering\footnotesize
\caption{Estimated \cite{lee2009bounds} bounds for racial and gender wage gaps}\label{tab: Lee bounds}
\begin{tabular}{lclc}
\toprule
Category & Racial Gap& Category & Gender Gap \\
\midrule
 Men, High school   & [-0.249, -0.074] & Mexican, High school & [-0.325, -0.105] \\
 Women, High school & [-0.210, -0.041] & White, High school & [-0.372, -0.142] \\
 Men, College      & [-0.205, 0.015] & Mexican, College & [-0.240, -0.152] \\
 Women, College    & [-0.235, 0.035] & White, College & [-0.284, -0.119] \\
\bottomrule
\end{tabular}
\end{table}

In accordance with our simulations, we evaluate the TPM, Heckman’s MLE (HSM), and our estimator (KL). In the first stage, we estimate the selection probability by the sieve probit MLE, using group-specific cubic regression splines of age and experience (five basis functions each) together with the group dummy and the remaining controls. Nested likelihood-ratio tests reject a linear index and, more importantly, an index spanned by the outcome-equation regressors, in favour of the sieve in every subsample, indicating strong nonlinearity in the selection process. The test results for the selection equation are provided in Section~S5 of the Online Appendix (Table~S3). After obtaining $\hat{p}_i$, we estimate a partial linear model approximating $\lambda(\cdot)$ by a cubic B-spline basis of $\hat{p}$ with five degrees of freedom, and we report the corrected standard errors for our estimator. Across all reported specifications the full-sample fitted propensities remain interior, ranging over $[0.098,0.950]$, and both stages are well conditioned.
The estimation results remain stable when we perturb the number of knots. We report the results with seven knots for both stages in Section~S6 of the Online Appendix (Table~S7). We also compute the HH bounds along with point estimates. Unlike the Lee bounds, we can use the full set of covariates for the HH bounds.

\begin{table}[ht!]
	\centering\footnotesize
	\caption{Estimated racial and gender wage gaps}\label{tab:regression-combined}
	\begin{tabular}{llcccc}
		\toprule
		Gap & Group & TPM & HSM & KL & HH bound\\
		\midrule
		Racial & Men & $-0.113$ $(0.005)$ & $-0.113$ $(0.005)$ & $-0.089$ $(0.009)$ & $[-0.114,-0.103]$\\
		& Women & $-0.078$ $(0.005)$ & $-0.078$ $(0.005)$ & $-0.066$ $(0.006)$ & $[-0.089,-0.066]$\\
		Gender & White & $-0.209$ $(0.003)$ & $-0.159$ $(0.004)$ & $-0.211$ $(0.005)$ & $[-0.214,-0.179]$\\
		& Mexican & $-0.193$ $(0.006)$ & $-0.195$ $(0.007)$ & $-0.180$ $(0.012)$ & $[-0.219,-0.145]$\\
		\bottomrule
	\end{tabular}
	\par\footnotesize\renewcommand{\baselineskip}{11pt}\justifying
	\textbf{Note:} Standard errors are in parentheses. KL standard errors use the first-stage correction in \eqref{eq: corrected variance estimator}. HH denotes \cite{honore2020selection} bounds.
\end{table}

The estimation results for racial and gender wage gaps (full regression results with all covariates are provided in Section~S6 of the Online Appendix) are shown in Table~\ref{tab:regression-combined}. For the racial gap, the TPM and Heckman's MLE produce the same estimate for both men (-11.3\%) and women (-7.8\%). They generally produce almost identical estimates for all covariates. In Heckman's approach, the null hypothesis of no selection bias cannot be rejected. The HH bounds are informative and much narrower than the Lee bounds, ranging between -11.4\% and -10.3\% for men and between -8.9\% and -6.6\% for women. The estimates from both TPM and Heckman's MLE are contained in the HH bounds, meaning that the linear selection models fail to capture any selection bias. As the selection process exhibits strong nonlinearity, the linear selection models are misspecified regardless of the assumptions about the error terms. By contrast, our estimator shows a smaller magnitude of the racial wage disparity: the estimate for men (-8.9\%) lies outside the HH bounds, and the estimate for women (-6.6\%) sits at the upper edge of the HH bounds.

The results on the gender wage gap also show interesting patterns. For Mexican-Americans, the TPM and Heckman's MLE again produce the same estimates (around -19.5\%). In contrast, our estimator indicates a smaller magnitude of the gender wage disparity (-18.0\%). The HH bounds (-21.9\% to -14.5\% for Mexican-Americans and -21.4\% to -17.9\% for whites) are narrower than the Lee bounds but not as tight as for the racial gaps. All wage gap estimates are contained in the HH bounds. For non-Hispanic whites, the patterns are quite the opposite. Heckman's MLE seems to over-correct the selection bias, delivering a much smaller magnitude of the gender wage gap (-15.9\%) than TPM (-20.9\%). It also indicates much larger premiums on higher degrees (college and advanced degrees) compared to high school diplomas than TPM. These patterns are reversed in the semiparametric estimation. Our estimator produces an estimate of the gender wage gap (-21.1\%) very similar to that from TPM, while producing lower wage premiums for higher degrees. Interestingly, Heckman's MLE estimate does not lie in the HH bounds, whereas our semiparametric estimate is still contained in the HH bounds.

This empirical application demonstrates that the widely used bounds approach proposed by \cite{lee2009bounds} can yield uninformative bounds in analysing crucial labour market outcomes, such as wages. The HH bounds offer a potential alternative, as they tend to provide tighter bounds. However, even the feasible non-sharp version of the HH bounds (as the sharp characterization relies on an uncountable infinity of moment inequalities) is computationally intensive. Moreover, inference for these bounds hinges on resampling, which can be computationally demanding. In contrast, our semiparametric estimator is easy to implement in standard statistical packages.

\section{Concluding remarks}\label{sec: conclusion}

In this paper, we investigate point identification and efficient estimation of semiparametric selection models without an exclusion restriction. The primary objective of our paper is to challenge the long-held belief that an exclusion restriction is necessary for semiparametric selection models. Bounds approaches for selection models are often motivated by this misconception. We present convenient and practical semiparametric estimators that accommodate nonlinear selection, heteroskedastic errors, multiple control variables, and simple asymptotically valid inference. We also characterize the semiparametric efficiency bound for the model and provide a weighted variant of the estimator that attains the bound. The identifying conditions are readily verifiable in practice, as researchers simply need to ensure the presence of a continuous variable in the data and reject the linearity of the selection process.

In the simulations and empirical application, we demonstrate that our method offers more robust estimates than linear selection models. Therefore, our estimator presents a valuable alternative in cases where the bounds approaches fail to provide informative results. For researchers interested in correcting sample selection bias without resorting to unjustifiable parametric or distributional assumptions, we recommend reporting point estimates from our semiparametric method. Although our approach is more restrictive than Lee’s nonparametric model, it can accommodate restrictive forms of parameter heterogeneity. Extending our results to more flexible parameter heterogeneity would be an intriguing avenue for future research. Another promising research direction would be the identification of semiparametric sample selection models with endogenous regressors.

\section*{Acknowledgements}

We thank Krishna Pendakur and Myung Hwan Seo for their helpful comments and
suggestions. We also benefited greatly from comments by seminar and conference
participants at Seoul National University, Kyung Hee University, Sogang
University, Hanyang University, Korea University, CIREQ 2025 Econometrics
Conference, and Econometric Society World Congress 2025. All errors are our
own. Kim gratefully acknowledges support from the Social Sciences and
Humanities Research Council of Canada under the Insight Development Grant
(430-2022-00841) and the Insight Grant (435-2024-0322). Lee gratefully
acknowledges support from Hanyang University (HY-202600000001438). The authors declare
that they have no conflicts of interest.

\bibliographystyle{chicago}
\bibliography{ref}

\appendix

\section{Proofs for Identification Results}
\begin{proof}[Proof of Proposition \ref{prop: identification - single continuous variable}]
    For any observationally equivalent $\beta$ and $\lambda(\cdot)$, both
    $l(x')+b(p_0(x'))$ and $l(x'')+b(p_0(x''))$ equal zero. Since
    $p_0(x')=p_0(x'')$, subtraction gives
    $(\beta_0-\beta)(x'-x'')=0$. As $x'\neq x''$, $\beta=\beta_0$.
    Given $\beta_0$, $\lambda_0(\cdot)$ is identified by
    $E[Y-X\beta_0\mid X,D=1]=\lambda_0(p_0(X))$.
\end{proof}

\begin{proof}[Proof of Proposition \ref{prop: identification - two continuous variables}]
    Let $d_X:=\dim(\beta)$.
    Differentiating $l(x)+b(p)=0$ with respect to $x_1$ and $x_2$ gives
\begin{align*}
    \beta_{01}-\beta_1+b'(p)p_{x_1}&=0,&
    \beta_{02}-\beta_2+b'(p)p_{x_2}&=0,\\
    \text{hence}\qquad
    \beta_1-\beta_{01}
    &=\frac{p_{x_1}}{p_{x_2}}(\beta_2-\beta_{02}).
\end{align*}
The maintained nonproportionality of $p_{x_1}$ and $p_{x_2}$ implies
$\beta_1=\beta_{01}$ and $\beta_2=\beta_{02}$, and hence $b'(p)=0$.
Thus $b(p)=C$ and
	$\sum_{j=3}^{d_X}(\beta_{0j}-\beta_j)x_j=C$ on
	$\operatorname{Supp}(X_{3:d_X})$. Assumption \ref{ass: identification}(vi)
implies $\beta_{0j}=\beta_j$ for all $j\ge3$ and $C=0$. Therefore
$\beta=\beta_0$, and $\lambda_0(\cdot)$ follows from
$E[Y-X\beta_0\mid X,D=1]=\lambda_0(p_0(X))$.
\end{proof}

\begin{proof}[Proof of Proposition \ref{prop: identification - one continuous variable and discrete variables}]
    Let $p'=p_0(x')$, $p''=p_0(x'')$, and
    $p'''=p_0(x''')=p''$. Observational equivalence at the three points gives
\begin{align*}
    \beta_{0j}-\beta_j+b(p')-b(p'')&=0,\\
    (\beta_{01}-\beta_1)(x_1'-x_1''')+b(p')-b(p'')&=0.
\end{align*}
Consequently,
$\beta_{0j}-\beta_j=(\beta_{01}-\beta_1)(x_1'-x_1''')$.
Since $x_1'-x_1'''$ is not constant across $x_1'$, it follows that
$\beta_{01}=\beta_1$ and $\beta_{0j}=\beta_j$. Differentiating
$l(x)+b(p)=0$ with respect to $x_1$ then gives $b'(p)=0$ because
$p_{x_1}\neq0$. The remaining coefficients and $\lambda_0(\cdot)$ are
identified as in the proof of Proposition
\ref{prop: identification - two continuous variables}.
\end{proof}

\clearpage
\setcounter{section}{0}
\renewcommand{\thesection}{S\arabic{section}}
\renewcommand{\theHsection}{S\arabic{section}}
\setcounter{equation}{0}
\setcounter{table}{0}
\setcounter{figure}{0}
\renewcommand{\thetable}{S\arabic{table}}
\renewcommand{\thefigure}{S\arabic{figure}}
\renewcommand{\theHtable}{S\arabic{table}}
\renewcommand{\theHfigure}{S\arabic{figure}}

\begin{center}
  {\LARGE\bfseries Online Appendix\par}
  \vspace{0.6em}
  {\large Point-identifying semiparametric sample selection models
  with no excluded variable\par}
\end{center}
\vspace{1.5em}

\let\ArxivSection\section
\renewcommand{\section}[1]{%
  \ifstrequal{#1}{Full wage regression tables}{\clearpage}{}%
  \ArxivSection{#1}%
}
\noindent\textit{Note on cross-references.} Sections, tables, figures and
numbered results of this Online Appendix carry the prefix `S'. Numbers without
that prefix, such as Assumption~\ref{assu:p} or Proposition~\ref{prop: asymptotics np},
refer to the paper.

\section{Proofs of Propositions~\refPropAsymNP{} and~\refPropPrim}\label{app: np proof}

This appendix proves both layers in theorem order. We first prove Proposition~\ref{prop: asymptotics np} directly under the high-level conditions of Assumption~\ref{assu:np-hl}. We then prove Proposition~\ref{prop: np primitives}: Lemmas~\ref{lem: np first stage rates}--\ref{lem: np riesz} verify those high-level conditions from the primitive growing-sieve and second-stage assumptions, after which Proposition~\ref{prop: asymptotics np} applies. Throughout we treat $X_i,v_i,\alpha_{0}$ as column vectors in $\mathbb{R}^{d}$, so that $A=E[D_iv_iv_i']$, $\tilde{X}_i=D_iv_i$, $\chi_i=\left(D_i-p_{0i}\right)\alpha_{0}(X_i)\in\mathbb{R}^{d}$, $\Omega=E[\sigma_{0}^{2}(X_i,D_i)D_iv_iv_i']$, and $\Omega_{\chi}=E[\chi_i\chi_i']$ ($d\times d$). The main-text displays of Section~\ref{sec: estimator} write $v_i$ as a row vector, so that their $v_i'$ is the transpose of the column $v_i$ used here; the two definitions of $\chi_i$ and $\Omega_{\chi}$ therefore coincide. We write $\delta_i:=\hat{p}_i-p_{0i}$ and $\phi_i:=D_iv_i\lambda_{0}'(p_{0i})$. Since $v_i=X_i-m_{X,p_0}(p_{0i})$ and $\lambda_{0}'(p_{0i})$ are deterministic functions of $X_i$ (with $m_{X,p_0}(p):=E[X\mid p_{0}(X)=p,D=1]$) and $E[D_i\mid X_i]=p_{0i}$, we have $E[\phi_i\mid X_i]=p_{0i}\lambda_{0}'(p_{0i})v_i=\alpha_{0}(X_i)$. Let $\sigma_{0}^{2}(x):=E[\varepsilon_i^{2}\mid X_i=x,D_i=1]$. For the first-stage series basis $\psi_{L}$ we set $G:=E[\psi_{L}\psi_{L}']$, $\hat{G}:=n^{-1}\sum_i\psi_{L}(X_i)\psi_{L}(X_i)'$, $\zeta:=\zeta(L)$; $\pi_{0}:=G^{-1}E[\psi_{L}D]=G^{-1}E[\psi_{L}p_{0}]$, $p_{0,L}:=\psi_{L}'\pi_{0}$, $r_{L}:=p_{0}-p_{0,L}$, $\Delta_{p}:=\lVert r_{L}\rVert_{L^{2}}$; and $\alpha_{0,L}:=\Gamma G^{-1}\psi_{L}$ with $\Gamma:=E[\alpha_{0}\psi_{L}']$ the $L^{2}(F_{X})$-projection of $\alpha_{0}$ onto $\mathrm{span}(\psi_{L})$, $\Delta_{\alpha}:=\lVert\alpha_{0}-\alpha_{0,L}\rVert_{L^{2}}$. By Assumption~\ref{assu:np-firststage}(i) the eigenvalues of $G$ lie in $[\underline{\lambda}_{G},\bar{\lambda}_{G}]\subset(0,\infty)$ uniformly in $L$. Generic finite constants are denoted $C$ and may change from line to line. The infeasible estimator $\tilde{\beta}_n$ (with $p_{0}$ known) is the unconstrained least-squares projection of $Y$ on $X$ and the spline span $\{b_K(p_{0i})\}_{D_i=1}$ within the selected sample: it is defined by the normal equations, which are exact identities rather than local approximations, so no compactness of the parameter space or interiority of $\beta_{0}$ enters. The generalized inverse in its closed form makes it well defined in every sample; the residualized Gram matrix is nonsingular with probability approaching one under the oracle conditions together with the second-stage Gram convergence, so the convention adopted off that event does not affect the asymptotics. Under the oracle conditions (Assumptions~\ref{assu:IID}, \ref{assu:hfunc}, \ref{assu:s2}, \ref{assu:Dv*}, together with the second-stage Gram conditions of Assumption~\ref{assu:secondstage}(i)--(ii) evaluated at $p_{0}$) it satisfies $\sqrt{n}(\tilde{\beta}_n-\beta_{0})=A^{-1}n^{-1/2}\sum_iD_iv_i\varepsilon_i+o_p(1)$. This is Theorem~2 of \cite{donald1994series} applied to the selection-weighted semilinear model
\[
y^{\dagger}_{i}=D_{i}Y_{i},\quad x^{\dagger}_{i}=D_{i}X_{i},\quad z^{\dagger}_{i}=(D_{i},p_{0i}),\quad g^{\dagger}(z^{\dagger}_{i})=D_{i}\lambda_{0}(p_{0i}),
\]
with series terms $D_{i}b_K(p_{0i})$, for which $x^{\dagger}_{i}-E[x^{\dagger}_{i}\mid z^{\dagger}_{i}]=D_{i}\{X_{i}-E[X_{i}\mid p_{0i},D_{i}=1]\}=D_{i}v_{i}$, so that their $\bar{A}_{n}$ and $\bar{B}_{n}$ are our $A$ and $\Omega$ and the normalisation is over the full sample; their framework admits the discrete component of $z^{\dagger}$. \citet[pp.~S225--S226]{newey2009two} uses the same $D$-weighted formulation before invoking this result. The work below is to show the feasible estimator differs by exactly the Riesz correction $-A^{-1}n^{-1/2}\sum_i\chi_i$.

For matrices, $\lVert\cdot\rVert_{op}$ and $\lVert\cdot\rVert_F$ denote the
operator and Frobenius norms; for vectors, $\lVert\cdot\rVert$ denotes the
Euclidean norm. For a scalar or vector-valued function $g$, write
$\lVert g\rVert_{n,2,D}^{2}:=n^{-1}\sum_iD_i\lVert g_i\rVert^{2}$ and
$\lVert g\rVert_{L^{2}}^{2}:=E[\lVert g(X_i)\rVert^{2}]$ unless a different
measure is displayed.

\subsection{Proof of Proposition~\ref{prop: asymptotics np}}\label{subsec: np assembly}

By the Frisch--Waugh--Lovell structure, $\hat{\beta}_n-\beta_0=\hat{A}^{-1}n^{-1}\sum_iD_i\hat{v}_i(\lambda_{0}(p_{0i})+\varepsilon_i)$. Because $\hat{v}_i$ is sample-orthogonal to the spline span, $n^{-1}\sum_iD_i\hat{v}_i\widehat{\Pi}_K\lambda_{0}(\hat{p}_i)=0$, so writing
\[
\lambda_{0}(p_{0i})=\bigl[\lambda_{0}(p_{0i})-\lambda_{0}(\hat{p}_i)+\lambda_{0}'(p_{0i})\delta_i\bigr]-\lambda_{0}'(p_{0i})\delta_i+\bigl[\lambda_{0}(\hat{p}_i)-\widehat{\Pi}_K\lambda_{0}(\hat{p}_i)\bigr]+\widehat{\Pi}_K\lambda_{0}(\hat{p}_i)
\]
and multiplying by $\sqrt{n}$, Assumption~\ref{assu:np-hl}(ii)--(iv) give $\sqrt{n}\,n^{-1}\sum_iD_i\hat{v}_i\lambda_{0}(p_{0i})=-n^{-1/2}\sum_i\chi_i+o_p(1)$: the first bracket is $o_p(1)$ by~(iii); the term $[\lambda_{0}(\hat{p}_i)-\widehat{\Pi}_K\lambda_{0}(\hat{p}_i)]$, on inserting $\pm\Pi_{K}\lambda_{0}(\hat{p}_i)$, splits into a spline-approximation part and a $\Pi_{K}-\widehat{\Pi}_K$ empirical-projection part, both $o_p(1)$ by~(iv); the final projection term vanishes exactly by sample orthogonality; and $-\sqrt{n}\,n^{-1}\sum_iD_i\hat{v}_i\lambda_{0}'(p_{0i})\delta_i=-n^{-1/2}\sum_i\chi_i+o_p(1)$ by~(ii). For the noise term, Assumption~\ref{assu:np-hl}(iv) gives $n^{-1/2}\sum_iD_i\hat{v}_i\varepsilon_i=n^{-1/2}\sum_iD_iv_i\varepsilon_i+o_p(1)=n^{-1/2}\sum_i\tilde{X}_i\varepsilon_i+o_p(1)$. With $\hat{A}\to_p A$ nonsingular (Assumption~\ref{assu:np-hl}(i),(iv)), Slutsky's theorem yields
\[
\sqrt{n}(\hat{\beta}_n-\beta_{0})=A^{-1}\frac1{\sqrt n}\sum_i\bigl\{\tilde{X}_i\varepsilon_i-\chi_i\bigr\}+o_{p}(1).
\]
The summands $\omega_i:=D_iv_i\varepsilon_i-\chi_i$ are i.i.d., do not depend on $L$ or $K$, and are mean zero: $E[D_iv_i\varepsilon_i]=E[D_iv_iE[\varepsilon_i\mid X_i,D_i{=}1]]=0$ and $E[\chi_i]=E[\alpha_{0}(X_i)E[D_i-p_{0i}\mid X_i]]=0$. Their covariance is $\Omega+\Omega_{\chi}$: the cross-term vanishes because
\[
E[D_iv_i\varepsilon_i\,\chi_i']=E\bigl[v_i\alpha_{0}(X_i)'\,D_i(D_i-p_{0i})\,\varepsilon_i\bigr]=E\bigl[v_i\alpha_{0}(X_i)'(1-p_{0i})\,E[D_i\varepsilon_i\mid X_i]\bigr]=0,
\]
using $D_i(D_i-p_{0i})=D_i(1-p_{0i})$ and $E[D_i\varepsilon_i\mid X_i]=p_{0i}E[\varepsilon_i\mid X_i,D_i{=}1]=0$; the variance blocks are $E[D_iv_iv_i'\varepsilon_i^{2}]=\Omega$ (finite under Assumptions~\ref{assu:IID}(ii) and~\ref{assu:s2}) and $E[\chi_i\chi_i']=E[(D_i-p_{0i})^{2}\alpha_{0}(X_i)\alpha_{0}(X_i)']=E[p_{0i}^{3}(1-p_{0i})\lambda_{0}'(p_{0i})^{2}v_iv_i']=\Omega_{\chi}$ (using $E[(D-p_{0})^{2}\mid X]=p_{0}(1-p_{0})$ and that $v_i$ is $X_i$-measurable; finite, since $\Omega_{\chi}\preceq\tfrac14\lVert\alpha_{0}\rVert_{L^{2}}^{2}I$, under Assumption~\ref{assu:np-hl}(i)). Since $\omega_i$ is a fixed (non-triangular) i.i.d.\ sequence with finite second moment, the Lindeberg--L\'evy central limit theorem gives $n^{-1/2}\sum_i\omega_i\to_{d}N(0,\Omega+\Omega_{\chi})$, and by Slutsky's theorem (with $A$ symmetric, $A^{-1\prime}=A^{-1}$),
\[
\sqrt{n}(\hat{\beta}-\beta_{0})\to_{d}N\bigl(0,\,A^{-1}(\Omega+\Omega_{\chi})A^{-1}\bigr)=N(0,V_{\mathrm{NP}}),
\]
which is the claim. \hfill$\qed$

\subsection{Proof of Proposition~\ref{prop: np primitives}}

We establish the primitive implication through the following five lemmas and then verify each part of Assumption~\ref{assu:np-hl}.

\subsubsection*{Lemma A: first-stage series rates}

\begin{lemma}\label{lem: np first stage rates}
Under Assumptions~\ref{assu:IID}(i)--(ii) and \ref{assu:np-firststage}(i)--(ii): (a) $\lVert G^{-1/2}\hat{G}G^{-1/2}-I_{L}\rVert_{op}=o_{p}(1)$, hence $\hat{G}$ is invertible with probability tending to one and $\lVert\hat{G}^{-1}\rVert_{op}=O_{p}(1)$; (b) $\lVert\hat{p}-p_{0}\rVert_{L^{2}(F_{X})}=O_{p}(\sqrt{L/n}+\Delta_{p})=o_{p}(n^{-1/4})$, and the empirical-norm version $n^{-1}\sum_i\delta_i^{2}=O_{p}(L/n+\Delta_{p}^{2})=o_{p}(n^{-1/2})$; (c) $\lVert\hat{p}-p_{0}\rVert_{\infty}=o_{p}(1)$.
\end{lemma}

\begin{proof}
(a) Whitening by $G^{-1/2}$, the matrix $G^{-1/2}\hat{G}G^{-1/2}-I_{L}$ is an average of i.i.d.\ mean-zero $L\times L$ matrices $\eta_i:=G^{-1/2}\psi_{L}(X_i)\psi_{L}(X_i)'G^{-1/2}-I_{L}$ with $\lVert\eta_i\rVert_{op}\le \underline{\lambda}_{G}^{-1}\zeta^{2}+1$. By Rudelson's matrix law of large numbers \citep[Lemma~6.2 of][]{belloni2015some} and Markov's inequality, $\lVert G^{-1/2}\hat{G}G^{-1/2}-I_{L}\rVert_{op}=O_{p}\bigl(\sqrt{\zeta^{2}\log L/n}\bigr)$. Assumption~\ref{assu:np-firststage}(i) imposes $\zeta^{2}L\log L/n\to0$, hence a fortiori $\zeta^{2}\log L/n\to0$ (divide by $L\ge1$), so the operator norm is $o_{p}(1)$, and Weyl's inequality gives $\lambda_{\min}(G^{-1/2}\hat{G}G^{-1/2})\ge\tfrac12$ with probability tending to one, whence $\lVert\hat{G}^{-1}\rVert_{op}\le2\underline{\lambda}_{G}^{-1}=O_{p}(1)$.

(b) Write $\hat{p}-p_{0}=(\hat{p}-p_{0,L})-r_{L}$ with $\lVert r_{L}\rVert_{L^{2}}=\Delta_{p}$ deterministic. For the in-span part, $\hat{p}-p_{0,L}=\psi_{L}'(\hat{\pi}-\pi_{0})$ and, as derived in \eqref{eq: np pi expansion} below, $\hat{\pi}-\pi_{0}=\hat{G}^{-1}(S_{n}+T_{n})$ with $S_{n}=n^{-1}\sum_i\psi_{L}(X_i)(D_i-p_{0i})$ and $T_{n}=n^{-1}\sum_i\psi_{L}(X_i)r_{L}(X_i)$. By \eqref{eq: np Sn rate}, $E\lVert S_{n}\rVert^{2}\le\bar{\lambda}_{G}L/(4n)$, so $\lVert\psi_{L}'\hat{G}^{-1}S_{n}\rVert_{L^{2}}^{2}=(S_{n}'\hat{G}^{-1}G\hat{G}^{-1}S_{n})=O_{p}(L/n)$ using part~(a); and $\lVert\psi_{L}'\hat{G}^{-1}T_{n}\rVert_{L^{2}}=O_{p}(\Delta_{p})$ since $\lVert T_{n}\rVert\le(n^{-1}\sum\lVert\psi_{L}(X_i)\rVert^{2})^{1/2}(n^{-1}\sum r_{L}(X_i)^{2})^{1/2}=O_{p}(\zeta\Delta_{p})$ and the in-span $L^{2}$ norm contracts the $\zeta$ factor to $O_{p}(\Delta_{p})$ by part~(a) (equivalently, $\lVert\Pi_{L}r_{L}\rVert_{L^{2}}\le\lVert r_{L}\rVert_{L^{2}}=\Delta_{p}$). Collecting, $\lVert\hat{p}-p_{0}\rVert_{L^{2}}=O_{p}(\sqrt{L/n}+\Delta_{p})$. The empirical-norm bound is identical with the empirical Gram in place of $G$: $n^{-1}\sum_i\delta_i^{2}\le2\,(\hat\pi-\pi_{0})'\hat{G}(\hat\pi-\pi_{0})+2\,n^{-1}\sum_ir_{L}(X_i)^{2}=2\,(S_{n}+T_{n})'\hat{G}^{-1}(S_{n}+T_{n})+O_{p}(\Delta_{p}^{2})=O_{p}(L/n+\Delta_{p}^{2})$, using $\lVert\hat{G}^{-1}\rVert_{op}=O_{p}(1)$, $\lVert S_{n}\rVert=O_{p}(\sqrt{L/n})$, $\lVert T_{n}\rVert=O_{p}(\zeta\Delta_{p}/\sqrt n)$ (here the sharper \emph{centred} rate, valid because $T_{n}$ is a mean-zero average as $E[\psi_{L}r_{L}]=0$, in place of the cruder uncentered $O_{p}(\zeta\Delta_{p})$ used for the $L^{2}$ bound above), and Markov's inequality for the nonnegative deterministic $r_{L}^{2}$, whose mean is exactly $\Delta_{p}^{2}$.

(c) The sup-norm consistency is imposed in Assumption~\ref{assu:np-firststage}(ii). For local stable bases it is delivered by the same normal-equation expansion: for any coefficient vector $c$, $\lVert\psi_{L}'c\rVert_{\infty}\le\zeta\lVert c\rVert$, and $\lVert\hat{\pi}-\pi_{0}\rVert=O_p(\sqrt{L/n}+\zeta\Delta_p/\sqrt n)$ using the centred rate for $T_n$ above. Together with a sup-norm approximation primitive for the basis, $\lVert r_L\rVert_\infty\to0$, a separate condition on the sieve, since a sup-norm approximation error is not implied by the $L^{2}$ error $\Delta_{p}$ alone; this gives $\lVert\hat p-p_0\rVert_\infty\le\zeta\lVert\hat{\pi}-\pi_{0}\rVert+\lVert r_{L}\rVert_{\infty}=O_p(\zeta\sqrt{L/n}+\zeta^2\Delta_p/\sqrt n)+\lVert r_{L}\rVert_{\infty}=o_p(1)$ under Assumption~\ref{assu:np-firststage}(i)--(ii). The $L^{2}$ bound of part~(b) is $o_{p}(n^{-1/4})$ because $L=o(\sqrt{n})$ and $\Delta_{p}=o(n^{-1/4})$ by Assumption~\ref{assu:np-firststage}(ii).
\end{proof}

\subsubsection*{Lemma C0: uniform concentration at generated indices}

The second-stage arguments below repeatedly involve series regressions on the basis $b_K(\hat{p}_i)$, whose index $\hat{p}$ is estimated from the same observations. A pointwise conditioning argument at $\hat{p}$ is not available for the covariate channel: $\hat{p}$ depends on every $(X_j,D_j)$, including observation $i$, so conditional on $\hat{p}$ the array $\{X_i-m_{X,\hat{p}}(\hat{p}_i)\}$ is not a mean-zero regression array. The following lemma replaces that invalid conditioning with uniformity over the admissible class $\mathcal{P}_{n}$ of Assumption~\ref{assu:p}, within which each \emph{fixed} index does deliver mean-zero noise; membership $\hat{p}\in\mathcal{P}_{n}$ (Assumption~\ref{assu:p}(ii)) then transfers every bound to the generated index.

\begin{lemma}\label{lem: np uniform}
Suppose Assumptions~\ref{assu:IID}, \ref{assu:hfunc}, \ref{assu:p}, \ref{assu:np-firststage}(i)--(ii) and \ref{assu:secondstage}(i)--(ii) hold, $K$ and $L$ grow polynomially in $n$, and the spline basis $b_K$ is constructed on a fixed compact interval containing the $\varepsilon$-enlargement of the support of $p_{0}(X_i)$. Let $\mathcal{H}_{n}:=\Lambda^{m}_{C_{m}}([0,1])\cup\mathcal{S}_{K}$, where $\mathcal{S}_{K}:=\{b_K'c:\lVert b_K'c\rVert_{\infty}\leq C_{h}\}$ is a common spline sup-norm ball with $C_{h}:=\Lambda_{\mathrm{Leb}}\max(C_{m},\lVert\lambda_{0}\rVert_{\infty})$ and $\Lambda_{\mathrm{Leb}}$ the uniform Lebesgue-constant bound of the (possibly weighted) spline $L^{2}$-projections; these two \emph{fixed} function balls contain every nuisance the lemmas below evaluate ($m_{g,p}$ and $\lambda_{\omega,p}$ in the H\"{o}lder ball, their spline projections and $\Pi_{K}\lambda_{0}$ in $\mathcal{S}_{K}$), so no link between the nuisance and the index is required. The common knot interval of the basis, the Lebesgue-constant bound, and the uniform Gram bounds of Assumption~\ref{assu:p}(iii) are used explicitly. Set $\rho_{n}:=\sqrt{(K+L)\log n/n}$. Then:
\begin{conditions}
\item[(a)] $\displaystyle\sup_{p\in\mathcal{P}_{n},\,h\in\mathcal{H}_{n}^{d}}\Bigl\lVert\frac1n\sum_i D_i\,b_K(p(X_i))\bigl[g_i-h(p(X_i))\bigr]'-E\bigl[D_i\,b_K(p(X_i))\bigl[g_i-h(p(X_i))\bigr]'\bigr]\Bigr\rVert_{F}=O_{p}(\rho_{n})$ for $g=X$ and for $g=\mu_{0}$, where $\mu_{0}(x):=x'\beta_{0}+\lambda_{0}(p_{0}(x))$ is the fixed, bounded outcome conditional mean, and $\mathcal{H}_{n}^{d}$ collects $\mathbb{R}^{d}$-valued nuisances with each coordinate in $\mathcal{H}_{n}$, so the statement is the scalar bound applied coordinatewise (for $g=\mu_{0}$ the raw-channel envelope uses boundedness of $\mu_{0}$ on the compact support, and $m_{\mu_{0},p}=m_{Y,p}$ lies in the class's H\"{o}lder ball); in particular, since $E[D_ib_K(p_i)u_i^{g}(p)']=0$ for each fixed $p\in\mathcal{P}_{n}$ with $u^{g}_{i}(p):=g_i-m_{g,p}(p(X_i))$ and $m_{g,p}(t):=E[g_i\mid p(X_i)=t,D_i=1]$ (the population normal equations at a fixed index, using $m_{g,p}\in\Lambda^{m}_{C_{m}}$ by the class definition),
\[
\sup_{p\in\mathcal{P}_{n}}\Bigl\lVert\frac1n\sum_i D_i\,b_K(p(X_i))\,u^{g}_{i}(p)'\Bigr\rVert_{F}=O_{p}(\rho_{n})=o_{p}(1);
\]
\item[(b)] $\displaystyle\sup_{p\in\mathcal{P}_{n}}\lVert\widehat{G}_{K}(p)-G_{K}(p)\rVert_{op}=O_{p}\bigl(\xi(K)\rho_{n}\bigr)=o_{p}(1)$, where $\widehat{G}_{K}(p):=n^{-1}\sum_iD_ib_K(p_i)b_K(p_i)'$ and $G_{K}(p):=E[D_ib_K(p_i)b_K(p_i)']$; moreover, by Assumption~\ref{assu:p}(iii), the eigenvalues of $G_{K}(p)$ are bounded and bounded away from zero uniformly over $p\in\mathcal{P}_{n}$, so on the event $\{\hat{p}\in\mathcal{P}_{n}\}$ the generated-index empirical Gram $\widehat{G}_{K}(\hat{p})$ has smallest eigenvalue bounded away from zero with probability tending to one;
\item[(c)] \emph{(Generated-index joint least squares.)} Let $\omega$ be a fixed function of $x$ with $\underline{w}\leq\omega\leq\bar{w}$ (the unweighted case is $\omega\equiv1$) such that the $\omega$-weighted conditional means $m^{\omega}_{g,p}$, $g\in\{X,Y\}$, and the $\omega$-weighted partially linear component $\lambda_{\omega,p}:=m^{\omega}_{Y,p}-m^{\omega\prime}_{X,p}\beta_{\omega}(p)$ lie in a fixed $\Lambda^{m}$ ball uniformly over $p\in\mathcal{P}_{n}$, where $\beta_{\omega}(p)$ is the $\omega$-weighted partially linear coefficient at index $p$. For $Z_i(p):=(X_i',b_K(p(X_i))')'$, let $\hat{\theta}_{\omega}=(\hat{\beta}_{\omega}',\hat{c}_{\omega}')'$ be the $D\omega$-weighted least squares coefficients of $Y_i$ on $Z_i(\hat{p})$, and $\theta_{\omega}(p)=(\beta_{\omega}(p)',c_{\omega,p}')'$ with $c_{\omega,p}$ the coefficients of the $D\omega$-weighted spline projection of $\lambda_{\omega,p}$. Then, for $\varepsilon$ small enough in the definition of $\mathcal{P}_{n}$, on the event $\{\hat{p}\in\mathcal{P}_{n}\}$,
\[
\lVert\hat{\theta}_{\omega}-\theta_{\omega}(\hat{p})\rVert=O_{p}\bigl(\rho_{n}+K^{-m}\bigr).
\]
\end{conditions}
\end{lemma}

\begin{proof}
Every $p\in\mathcal{P}_{n}\setminus\{p_{0}\}$ has the form $p=\psi_{L}'\pi$ with $\lVert\pi\rVert\leq C_{\pi}:=\underline{\lambda}_{G}^{-1/2}(\lVert p_{0}\rVert_{\infty}+\varepsilon)$ (the coefficient-ball bound noted with the class definition). Throughout the proof, the $p$-net is a maximal $\delta$-separated subset of $\mathcal{P}_{n}$ \emph{itself} in the coefficient norm, augmented by $p_{0}$, with $\delta=n^{-4}$: by maximality every $p\in\mathcal{P}_{n}$ lies within $\delta$ of a centre in the coefficient norm, hence within $\zeta(L)\delta$ in sup norm; the centers are admissible, so the density, Gram, and fixed-index normal-equation identities of $\mathcal{P}_{n}$ hold at every centre; and a $\delta$-separated subset of the coefficient ball has at most $(3C_{\pi}/\delta)^{L}$ points, so the net has logarithmic cardinality at most $CL\log n$. All three parts below use this admissible net.

\emph{Part (b).} Fix $p\in\mathcal{P}_{n}$. The summands $D_ib_K(p_i)b_K(p_i)'-G_{K}(p)$ are i.i.d., mean zero, with operator norm at most $2\xi(K)^{2}$ and variance parameter $\lVert E[(D_ib_Kb_K')^{2}]\rVert_{op}\leq\xi(K)^{2}\lVert G_{K}(p)\rVert_{op}\leq C\xi(K)^{2}$ (eigenvalue bound below); the matrix Bernstein inequality \citep[Theorem~6.1 of][]{tropp2012user} gives $\lVert\widehat{G}_{K}(p)-G_{K}(p)\rVert_{op}=O_{p}(\xi(K)\sqrt{(t+\log K)/n}+\xi(K)^{2}(t+\log K)/n)$ with probability at least $1-e^{-t}$. Taking $t=CL\log n$ and a union bound over the net yields the rate $O_{p}(\xi(K)\rho_{n})$, with the Bernstein linear term dominated because $\xi(K)^{2}(K+L)\log n/n\to0$ under Assumptions~\ref{assu:np-firststage}(i) and \ref{assu:secondstage}(ii) with polynomial growth ($\xi(K)^{2}=O(K)$ for the normalized spline basis). Between net points, $\lVert b_K(p_i)b_K(p_i)'-b_K(p'_i)b_K(p'_i)'\rVert_{op}\leq2\xi(K)\sup_t\lVert b_K'(t)\rVert\,\zeta(L)\delta=O(K^{2}\zeta(L)n^{-4})$, negligible under polynomial growth. For the eigenvalue claim, condition on $p(X_i)=t$: $G_{K}(p)=\int b_K(t)b_K(t)'\,q_{p}(t)\,dt$ with $q_{p}(t):=P[D_i=1]\,f_{p(X)\mid D=1}(t)\in[P[D_i=1]\underline{f},\,P[D_i=1]\bar{f}]$ on the support (the class's density property), which shows that the density bounds are the right primitive where the supports align with the knot interval. The uniform two-sided eigenvalue bounds themselves are Assumption~\ref{assu:p}(iii); they cannot be derived from the density bounds alone, because a class member whose support does not cover the knot interval places no mass on some spline components, and they reduce to Assumption~\ref{assu:secondstage}(i) at $p=p_{0}$. The final claim follows from Weyl's inequality on $\{\hat{p}\in\mathcal{P}_{n}\}$.

\emph{Part (a).} Fix an outer coordinate of $g$; the argument is applied $d$ times. Write $V(p,h):=n^{-1}\sum_i\{D_ib_K(p_i)[g_i-h(p_i)]-E[\,\cdot\,]\}\in\mathbb{R}^{K}$ and dualize the norm,
\[
\lVert V(p,h)\rVert=\sup_{\lVert a\rVert=1}a'V(p,h)\leq2\max_{a\in\mathcal{N}}a'V(p,h),
\]
where $\mathcal{N}$ is a $1/2$-net of the unit sphere $S^{K-1}$ with $\lvert\mathcal{N}\rvert\leq5^{K}$. The supremum is then covered \emph{jointly} over the sphere net, the index class, and the nuisance class, which is what delivers complexity of order $K+L$ and the rate $\rho_{n}$; coordinatewise maximal inequalities followed by aggregation would lose a factor $\sqrt{K}$ and are not used. For the index class, use the admissible $p$-net fixed in the opening paragraph. By linearity in $h$ it suffices to treat the raw channel and the two fixed balls of $\mathcal{H}_{n}$.

\emph{Raw and spline-ball channels.} For fixed $(a,p)$ and fixed $h\in\mathcal{S}_{K}$ (or $h=0$ with $g\in\{X,\mu_{0}\}$), the scalar summands $D_i(a'b_K(p_i))[g_i-h(p_i)]$, centred, are i.i.d.\ and bounded by $C\xi(K)$, with variance $E[D_i(a'b_K(p_i))^{2}(g_i-h(p_i))^{2}]\leq C\,a'G_{K}(p)a\leq C\bar{c}_{G}$, dimension free, by boundedness of $g$ (compact support of $X_i$, Assumption~\ref{assu:IID}(ii), for $g=X$; boundedness of $\mu_{0}$ for $g=\mu_{0}$), the ball bound $\lVert h\rVert_{\infty}\leq C_{h}$, and Assumption~\ref{assu:p}(iii). Cover $\mathcal{S}_{K}$ by a sup-norm $n^{-4}$-net (a $K$-dimensional symmetric convex body; logarithmic cardinality at most $CK\log n$ by the volumetric bound). The scalar Bernstein inequality (the $1\times1$ case of Theorem~6.1 of \citealp{tropp2012user}) with $t=C'(K+L)\log n$, together with the union over $\mathcal{N}\times(p\text{-net})\times(h\text{-net})$ of total logarithmic cardinality $O((K+L)\log n)$, gives
\[
\max\,\lvert a'V(p,h)\rvert=O_{p}\Bigl(\rho_{n}+\xi(K)\,\frac{(K+L)\log n}{n}\Bigr)=O_{p}(\rho_{n}),
\]
the Bernstein linear term being dominated because $\xi(K)^{2}(K+L)\log n/n\rightarrow0$ under Assumptions~\ref{assu:np-firststage}(i) and \ref{assu:secondstage}(ii) with polynomial growth.

\emph{H\"{o}lder-ball channel.} For $h\in\Lambda^{m}_{C_{m}}$ split $h=\Pi_{K}h+(h-\Pi_{K}h)$: the projection lies in $\mathcal{S}_{K}$ ($\lVert\Pi_{K}h\rVert_{\infty}\leq\Lambda_{\mathrm{Leb}}C_{m}\leq C_{h}$) and is covered above, while the residual class $\mathcal{R}:=\{h-\Pi_{K}h:h\in\Lambda^{m}_{C_{m}}\}$ has envelope $CC_{m}K^{-m}$ (the Jackson bound) and sup-norm bracketing entropy $\log N_{[\,]}(\epsilon)\leq C\epsilon^{-1/m}$ inherited from the H\"{o}lder ball \citep[Theorem~2.7.1 of][]{vaart1996weak}. At each fixed $(a,p)$, the maximal inequality for bracketing entropy \citep[Lemma~3.4.2 of][]{vaart1996weak} bounds the expected supremum of the centred empirical average over $\mathcal{R}$ by $C\,\xi(K)K^{-m+1/2}/\sqrt{n}=CK^{1-m}/\sqrt{n}=o(\rho_{n})$ (the entropy integral converges for $m>1/2$, Assumption~\ref{assu:hfunc}), and Talagrand's concentration inequality for suprema of bounded empirical processes \citep{bousquet2002bennett}, with variance parameter $\sigma^{2}\leq CK^{-2m}\bar{c}_{G}$ and envelope $C\xi(K)K^{-m}$, converts this into a tail bound at level $e^{-t}$; taking $t=C'(K+L)\log n$ and the union over $\mathcal{N}\times(p\text{-net})$ yields the uniform bound $O_{p}(K^{1-m}/\sqrt{n}+K^{-m}\rho_{n}+\xi(K)K^{-m}(K+L)\log n/n)=o_{p}(\rho_{n})$.

\emph{Off the nets.} Moving $p$ to its net centre $p'$ with $\lVert p-p'\rVert_{\infty}\leq\zeta(L)\delta$, at fixed $h$, perturbs each summand by at most $C\sup_t\lVert b_K'(t)\rVert\,\zeta(L)\delta+\xi(K)\,\lvert h(p(x))-h(p'(x))\rvert$, and the nuisance move is bounded by the H\"{o}lder modulus of continuity, $C_{m}(\zeta(L)\delta)^{m\wedge1}$, on the H\"{o}lder ball, and by the spline inverse inequality $\lVert h'\rVert_{\infty}\leq CK\lVert h\rVert_{\infty}\leq CKC_{h}$ on $\mathcal{S}_{K}$; moving $h$ to its net centre within $\mathcal{S}_{K}$ costs at most $\xi(K)n^{-4}$ directly, since that net is taken in the sup norm. Every increment is polynomial in $n$ times $n^{-4}$, hence negligible. Collecting channels proves the displayed bound; the ``in particular'' statement follows because for each fixed $p\in\mathcal{P}_{n}$ the choice $h=m_{g,p}$ lies in the H\"{o}lder ball by the class definition and makes the population term vanish by the normal equations at the fixed index $p$.

\emph{Part (c).} Write $Q_{\omega}(p):=E[D_i\omega_iZ_i(p)Z_i(p)']$ and $\widehat{Q}_{\omega}(p):=n^{-1}\sum_iD_i\omega_iZ_i(p)Z_i(p)'$. \emph{Joint Gram.} The weight sandwich $\underline{w}\,Q_{1}(p)\preceq Q_{\omega}(p)\preceq\bar{w}\,Q_{1}(p)$ reduces the eigenvalue bounds to the unweighted case. Upper: the $X$-block is bounded by compactness and the $b$-block by Assumption~\ref{assu:p}(iii). Lower, by the Schur complement: using the off-diagonal control $\lVert E[DXb_K(p)']G_{K}(p)^{-1/2}\rVert_{op}^{2}=\lVert E[D\,\Pi^{p}_{K}X(\Pi^{p}_{K}X)']\rVert_{op}\leq\lVert E[DXX']\rVert_{op}\leq C$ (projection contraction and compact support), the block-triangular factorization gives $\lambda_{\min}(Q_{1}(p))\geq c\min\{\lambda_{\min}(G_{K}(p)),\lambda_{\min}(S_{K}(p))\}$ with $c$ depending only on that norm, where $S_{K}(p):=E[D_iv_{K,p,i}v_{K,p,i}']$ and $v_{K,p}:=X-\Pi^{p}_{K}X$ is the spline-partialled residual at index $p$; since $v_{K,p}=v_{p}+(m_{X,p}-\Pi^{p}_{K}m_{X,p})(p)$ with the two pieces $D$-orthogonal, $S_{K}(p)\succeq E[D_iv_{p,i}v_{p,i}']=:A(p)$, and $\lVert A(p)-A\rVert\leq C\lVert p-p_{0}\rVert_{L^{2}}\leq C\varepsilon$ by Assumption~\ref{assu:diff}(ii), so for $\varepsilon\leq\lambda_{\min}(A)/(2C)$ we have $\lambda_{\min}(A(p))\geq\lambda_{\min}(A)/2$ uniformly over $\mathcal{P}_{n}$; $G_{K}(p)$ is bounded below by Assumption~\ref{assu:p}(iii). The empirical Gram concentrates uniformly: the $bb'$ block is part (b) with the bounded weight absorbed, the $Xb'$ block is the part-(a) machinery with $g=X$ and $h\equiv0$, and the $XX'$ block is a fixed-dimensional law of large numbers; hence $\lVert\widehat{Q}_{\omega}(\hat{p})^{-1}\rVert_{op}=O_{p}(1)$ on the membership event. \emph{Score.} At any fixed $p$, decompose $Y_i-Z_i(p)'\theta_{\omega}(p)=\varepsilon_i+q^{\omega}_{p}(X_i)+r^{\omega}_{p,K}(p(X_i))$, where $q^{\omega}_{p}:=\mu_{0}-X'\beta_{\omega}(p)-\lambda_{\omega,p}(p(X))$ and $r^{\omega}_{p,K}:=\lambda_{\omega,p}-\Pi^{\omega,p}_{K}\lambda_{\omega,p}$ with $\lVert r^{\omega}_{p,K}\rVert_{\infty}\leq CK^{-m}$ uniformly, by the $\Lambda^{m}$ ball and the uniformly bounded Lebesgue constant of the weighted spline projection. Three channels. (i) \emph{Outcome shock:} conditional on the first stage (and, in fold applications, on the fold's conditioning field), $Z_i(\hat{p})$ and $\omega_i$ are fixed and $E[D_i\varepsilon_i\mid\cdot]=0$, so the conditional-variance bound gives $\lVert n^{-1}\sum_iD_i\omega_iZ_i(\hat{p})\varepsilon_i\rVert=O_{p}(\sqrt{K/n})\leq\rho_{n}$; this conditioning is valid because the outcome shock does not enter the first stage. (ii) \emph{Design discrepancy:} the population term vanishes at every fixed $p$ by the first-order conditions of the $\omega$-weighted partially linear projection, $E[D\omega Xq^{\omega}_{p}]=0$ and $E[D\omega q^{\omega}_{p}\mid p(X)]=0$ (so $E[D\omega b_K(p)q^{\omega}_{p}]=0$); the centred process is controlled uniformly over $\mathcal{P}_{n}$ by linearity: $q^{\omega}_{p}$ splits into the $g=\mu_{0}$ channel, the $g=X$ channels scaled by $\beta_{\omega}(p)$, which is bounded uniformly over $\mathcal{P}_{n}$ since $\lVert\beta_{\omega}(p)\rVert\leq\lambda_{\min}(A_{\omega}(p))^{-1}\lVert E[D\omega v^{\omega}_{p}\mu_{0}]\rVert=O(1)$ by the Schur lower bound above and boundedness of $\mu_{0}$, and the H\"{o}lder-ball channel of part (a) with $h=\lambda_{\omega,p}$, the bounded weight rescaling envelopes only, each $O_{p}(\rho_{n})$; the $X$-rows obey the same bounds with the fixed-dimensional $X$ in place of $b_K$. (iii) \emph{Spline bias:} the $b$-rows have zero population term by the weighted-projection orthogonality and centred part $O_{p}(K^{-m}\rho_{n})$ by the $K^{-m}$-scaled Bernstein argument; the $X$-rows have population term $E[D\omega\,m^{\omega}_{X,p}(p)\,r^{\omega}_{p,K}(p)']=O(K^{-m})$ and centred part of the same order. Collecting, the score is $O_{p}(\rho_{n}+K^{-m})$ uniformly over $\mathcal{P}_{n}$, and multiplying by $\widehat{Q}_{\omega}(\hat{p})^{-1}=O_{p}(1)$ on the membership event yields the claim.
\end{proof}

\subsubsection*{Lemma C: second-stage reduction to a linear functional}

By Frisch--Waugh--Lovell, with $\hat{v}_i$ the within-selected-sample residual of $X_i$ on $\mathrm{span}\{b_K(\hat{p}_j)\}_{D_j=1}$, $\hat{A}=n^{-1}\sum_iD_i\hat{v}_i\hat{v}_i'$, and $\hat{v}_i$ sample-orthogonal to the spline span ($\sum_{D_i=1}\hat{v}_ib_K(\hat{p}_i)'=0$),
\begin{equation}\label{eq: np FWL}
\hat{\beta}-\beta_{0}=\hat{A}^{-1}\,\frac1n\sum_iD_i\hat{v}_i\bigl(\lambda_{0}(p_{0i})+\varepsilon_i\bigr).
\end{equation}

\begin{lemma}\label{lem: np second stage}
Under Assumptions~\ref{assu:IID}, \ref{assu:hfunc}, \ref{assu:s2},
\ref{assu:Dv*}, \ref{assu:np-firststage}, \ref{assu:diff}, \ref{assu:p},
and \ref{assu:secondstage}, together with the overlap condition of
Proposition~\ref{prop: np primitives}, $\hat{A}\to_{p}A$, the oracle term
$n^{-1/2}\sum_iD_i(\hat{v}_i-v_i)\varepsilon_i=o_{p}(1)$, and the
second-stage spline-approximation and empirical-projection contributions to
the $X$-moment are $o_{p}(n^{-1/2})$; consequently
\[
\sqrt{n}(\hat{\beta}-\beta_{0})=A^{-1}\Bigl[\frac1{\sqrt{n}}\sum_iD_iv_i\varepsilon_i-\frac1{\sqrt{n}}\sum_i\phi_i\delta_i+\sqrt{n}R_{2,n}\Bigr]+o_{p}(1),
\]
with $R_{2,n}$ the first-order Taylor remainder of Lemma~\ref{lem: np quadratic}. This establishes Assumption~\ref{assu:np-hl}(iv) and reduces the first-stage drift to the linear functional $-n^{-1}\sum_i\phi_i\delta_i$.
\end{lemma}

\begin{proof}
\emph{(\S1: $\hat{A}\to_{p}A$.)} Write $m_{X,p}(t):=E[X\mid p(X)=t,D=1]$ for the conditional mean of $X$ given the value of the conditioning function $p$, so that $v_i=X_i-m_{X,p_0}(p_{0i})$, and note that the empirical spline regression $\widehat{\Pi}_{K}X_i$ (of $X$ on $b_K(\hat{p})$ in the selected sample) estimates $m_{X,\hat{p}}$, the conditional mean \emph{along} $\hat{p}$, not $m_{X,p_0}$. Since $\hat{v}_i-v_i=m_{X,p_0}(p_{0i})-\widehat{\Pi}_{K}X_i$, decompose
\[
\hat{v}_i-v_i=\underbrace{[m_{X,p_0}(p_{0i})-m_{X,\hat{p}}(\hat{p}_i)]}_{\mathrm{(I)}}+\underbrace{[m_{X,\hat{p}}(\hat{p}_i)-\Pi_{K}m_{X,\hat{p}}(\hat{p}_i)]}_{\mathrm{(II)}}+\underbrace{[\Pi_{K}m_{X,\hat{p}}(\hat{p}_i)-\widehat{\Pi}_{K}X_i]}_{\mathrm{(III)}}.
\]
In $L^{2}$ on the selected sample: $\lVert\mathrm{(I)}\rVert_{L^{2}}^{2}\le C\lVert\hat{p}-p_{0}\rVert_{L^{2}}^{2}=o_{p}(1)$ by the conditional-mean stability Assumption~\ref{assu:np-firststage}(iv), which compares the two conditional-mean \emph{operators} $m_{X,\hat{p}}(\hat{p}_i)=E[X_i\mid\hat{p}(X_i),D_i{=}1]$ and $m_{X,p_0}(p_{0i})=E[X_i\mid p_0(X_i),D_i{=}1]$ (the bound does not follow from smoothness of the single map $m_{X,p_0}$ evaluated at $\hat{p}_i$ versus $p_{0i}$, since $X_i-m_{X,p_0}(\hat{p}_i)$ is not mean zero given $\hat{p}_i$); $\lVert\mathrm{(II)}\rVert_{L^{2}}=O(K^{-m})=o(1)$ by the H\"{o}lder smoothness of $m_{X,p}$ uniformly over $p\in\mathcal{P}_{n}$, which is part of the class definition (Assumption~\ref{assu:p}); and $\mathrm{(III)}$, the gap between the population and empirical spline projections of $X_i$ onto $\mathrm{span}\{b_K(\hat{p}_j)\}$, is controlled uniformly over the admissible class by Lemma~\ref{lem: np uniform}. At any fixed $p\in\mathcal{P}_{n}$ decompose $\Pi_{K}m_{X,p}(p_i)-\widehat{\Pi}_{K}X_i=[\Pi_{K}-\widehat{\Pi}_{K}]m_{X,p}(p_i)-\widehat{\Pi}_{K}u^{X}_i(p)$, where $u^{X}_{i}(p)=X_i-m_{X,p}(p_i)$. The projected-noise piece has empirical $L^{2}$ norm $O_{p}(\lVert n^{-1}\sum_iD_ib_K(p_i)u^{X}_i(p)'\rVert_{F})=O_{p}(\rho_{n})$ uniformly over $\mathcal{P}_{n}$, by Lemma~\ref{lem: np uniform}(a) with the uniform Gram conditioning of Lemma~\ref{lem: np uniform}(b). For the projection-difference piece, with $c_{p}:=G_{K}(p)^{-1}E[D_ib_K(p_i)m_{X,p}(p_i)']$ the population coefficients ($\lVert b_K'c_{p}\rVert_{\infty}=\lVert\Pi_{K}m_{X,p}\rVert_{\infty}\leq C$ by the uniformly bounded Lebesgue constant of the spline $L^{2}$-projection and $\lVert m_{X,p}\rVert_{\infty}\leq C_{m}$),
\[
[\Pi_{K}-\widehat{\Pi}_{K}]m_{X,p}=-\,b_K'\,\widehat{G}_{K}(p)^{-1}\Bigl[\Bigl(\tfrac1n\textstyle\sum_iD_ib_K(p_i)\,m_{X,p}(p_i)'-E[\cdot]\Bigr)-\Bigl(\tfrac1n\textstyle\sum_iD_ib_K(p_i)\,(b_K'c_{p})(p_i)'-E[\cdot]\Bigr)\Bigr],
\]
and both bracket terms are $O_{p}(\rho_{n})$ uniformly by Lemma~\ref{lem: np uniform}(a), the first with $h=m_{X,p}\in\Lambda^{m}_{C_{m}}$ and the second with $h=b_K'c_{p}\in\mathcal{S}_{K}$ (its sup norm is bounded by the Lebesgue constant); the fixed-vector form is essential here, as the crude operator-norm route $\lVert\widehat{G}_{K}-G_{K}\rVert_{op}\lVert c_{p}\rVert$ would lose a factor $\xi(K)$. On the event $\{\hat{p}\in\mathcal{P}_{n}\}$ of Assumption~\ref{assu:p}(ii), whose probability tends to one, every bound applies at $p=\hat{p}$, so $\lVert\mathrm{(III)}\rVert_{L^{2}}=O_{p}(\rho_{n})=o_{p}(1)$ with $\rho_{n}=\sqrt{(K+L)\log n/n}$. We emphasize that the pointwise argument at the random $\hat{p}$ is unavailable: $\hat{p}$ depends on every $(X_j,D_j)$, including observation $i$, so conditional on the fitted first stage the array $X_i-m_{X,\hat{p}}(\hat{p}_i)$ is not mean zero; uniformity over fixed indices in $\mathcal{P}_{n}$, plus membership, replaces that invalid conditioning. The generated-index Gram conditioning posited in Assumption~\ref{assu:secondstage}(v) is, on the same event, no longer an independent restriction: it follows from Lemma~\ref{lem: np uniform}(b). Hence $n^{-1}\sum_iD_i\lVert\hat{v}_i-v_i\rVert^{2}=o_{p}(1)$, and with $n^{-1}\sum_iD_iv_iv_i'\to_{p}A$ (law of large numbers under Assumption~\ref{assu:IID}(i)--(ii); $A$ nonsingular by Assumption~\ref{assu:Dv*}(i)) the Cauchy--Schwarz cross term vanishes, giving $\hat{A}\to_{p}A$.

\emph{(\S2: oracle replacement.)} We show $n^{-1/2}\sum_iD_i(\hat{v}_i-v_i)\varepsilon_i=o_{p}(1)$. Condition on the first stage $\mathcal{F}_{1}=\{(X_j,D_j)\}_{j=1}^{n}$; then $\hat{v}_i-v_i$ is $\mathcal{F}_{1}$-measurable, and by Assumption~\ref{assu:IID}(i) ($\varepsilon_i\perp\{(X_j,D_j)\}_{j\neq i}\mid(X_i,D_i)$) we have $E[D_i\varepsilon_i\mid\mathcal{F}_{1}]=D_iE[\varepsilon_i\mid X_i,D_i=1]=0$. Hence the sum is, conditionally, mean zero with variance
\[
\frac1n\sum_iD_i\lVert\hat{v}_i-v_i\rVert^{2}\,\sigma_{0}^{2}(X_i)\le\bar{\sigma}^{2}\cdot\frac1n\sum_iD_i\lVert\hat{v}_i-v_i\rVert^{2}=o_{p}(1)
\]
by \S1 and Assumption~\ref{assu:s2}. Chebyshev (conditionally) then unconditionally gives $o_{p}(1)$. Note this step integrates the \emph{noise} $\varepsilon_i$ against the increment and therefore needs only $\lVert\hat{v}-v\rVert_{L^{2}}=o_{p}(1)$, not undersmoothing.

\emph{(\S3: the control term.)} By sample-orthogonality, for the spline projection $\widehat{\Pi}_{K}$,
\[
\frac1n\sum_iD_i\hat{v}_i\lambda_{0}(p_{0i})=\frac1n\sum_iD_i\hat{v}_i\bigl[\lambda_{0}(p_{0i})-\widehat{\Pi}_{K}\lambda_{0}(\hat{p}_i)\bigr].
\]
Insert $\pm\lambda_{0}(\hat{p}_i)$ and $\pm\Pi_{K}\lambda_{0}(\hat{p}_i)$:
\[
\lambda_{0}(p_{0i})-\widehat{\Pi}_{K}\lambda_{0}(\hat{p}_i)=\underbrace{[\lambda_{0}(p_{0i})-\lambda_{0}(\hat{p}_i)]}_{\text{Taylor, to Lemma~\ref{lem: np quadratic}}}+\underbrace{[\lambda_{0}(\hat{p}_i)-\Pi_{K}\lambda_{0}(\hat{p}_i)]}_{\text{spline bias}}+\underbrace{[\Pi_{K}-\widehat{\Pi}_{K}]\lambda_{0}(\hat{p}_i)}_{\text{Gram error}}.
\]
Write
$r_i:=\lambda_{0}(p_{0i})-\lambda_{0}(\hat p_i)
+\lambda_{0}'(p_{0i})\delta_i$, so that
$\lambda_{0}(p_{0i})-\lambda_{0}(\hat p_i)
=-\lambda_{0}'(p_{0i})\delta_i+r_i$ and
$|r_i|\leq(\bar L_2/2)\delta_i^2$ by the Lipschitz continuity of
$\lambda_0'$ in Assumption~\ref{assu:hfunc}(i). The linear piece contributes
$-n^{-1}\sum_iD_i\hat{v}_i\lambda_{0}'(p_{0i})\delta_i$; replacing
$\hat{v}_i$ by $v_i$ here costs
$n^{-1}\sum_iD_i(\hat v_i-v_i)\lambda_0'(p_{0i})\delta_i$, which by
Cauchy--Schwarz and bounded $\lambda_0'$ is
$O_p\bigl(\lVert\hat v-v\rVert_{L^2}\lVert\delta\rVert_{L^2}\bigr)$.
By \S1,
$\lVert\hat v-v\rVert_{L^2}
=O_p(\rho_{n}+K^{-m}+\sqrt{L/n}+\Delta_p)$ with
$\rho_{n}=\sqrt{(K+L)\log n/n}$, and by
Lemma~\ref{lem: np first stage rates},
$\lVert\delta\rVert_{L^2}=o_p(n^{-1/4})$; multiplying by $\sqrt n$, the
$K^{-m}$, $\sqrt{L/n}$ and $\Delta_p$ contributions are $o_p(1)$ (the
first by Assumption~\ref{assu:secondstage}(iii), the latter two because
$L=o(\sqrt n)$ and $\Delta_p=o(n^{-1/4})$ under
Assumption~\ref{assu:np-firststage}(i)--(ii)), while the dominant
$\sqrt n\,\rho_{n}\,\lVert\delta\rVert_{L^2}
=O_p(\sqrt{(K+L)\log n}\cdot o(n^{-1/4}))$ is $o_p(1)$ by the cross-rate
condition $(K+L)\log n=o(\sqrt n)$ of
Assumption~\ref{assu:secondstage}(iv). Hence the substitution error is
$o_p(n^{-1/2})$, and the linear piece equals
$-n^{-1}\sum_i\phi_i\delta_i+o_p(n^{-1/2})$ with
$\phi_i=D_iv_i\lambda_0'(p_{0i})$; the remainder is handled in
Lemma~\ref{lem: np quadratic}. For the spline-bias part,
$\lVert\lambda_{0}(\hat p_\cdot)-\Pi_{K}\lambda_{0}(\hat p_\cdot)\rVert_{\infty}
=O(K^{-m})$ (Assumption~\ref{assu:hfunc}; the cubic-spline $L^{2}$-projection
has uniformly bounded Lebesgue constant, so its sup-norm error matches the
best approximation), so
\[
\Bigl\lVert\frac1{\sqrt n}\sum_iD_i\hat{v}_i[\lambda_{0}(\hat{p}_i)-\Pi_{K}\lambda_{0}(\hat{p}_i)]\Bigr\rVert\le\sqrt{n}\,\Bigl(\frac1n\sum_iD_i\lVert\hat v_i\rVert\Bigr)O(K^{-m})=O_{p}(\sqrt{n}\,K^{-m})=o_{p}(1)
\]
by the undersmoothing rate Assumption~\ref{assu:secondstage}(iii). The Gram-error part is bounded by the uniform machinery of Lemma~\ref{lem: np uniform}: since both $\lambda_{0}\in\Lambda^{m}$ and $\Pi_{K}\lambda_{0}$ (a spline with $\lVert\Pi_{K}\lambda_{0}\rVert_{\infty}\leq C$) lie in $\mathcal{H}_{n}$, the coefficient error of the empirical projection of $\lambda_{0}$ collapses to a single centred average, $[\Pi_{K}-\widehat{\Pi}_{K}]\lambda_{0}=-b_K'\widehat{G}_{K}(\hat p)^{-1}\bigl(n^{-1}\sum_iD_ib_K(\hat p_i)[\lambda_{0}-\Pi_{K}\lambda_{0}](\hat p_i)-E[\cdot]\bigr)$, whose summands carry the small envelope $\xi(K)\,O(K^{-m})$ and variance scaled by $\lVert\lambda_{0}-\Pi_{K}\lambda_{0}\rVert_{\infty}^{2}=O(K^{-2m})$; the same Bernstein-plus-net argument as Lemma~\ref{lem: np uniform}(a), with every bound scaled by $O(K^{-m})$, gives the uniform rate $O_{p}(K^{-m}\rho_{n})$ for its empirical $L^{2}$ norm on $\{\hat p\in\mathcal{P}_{n}\}$. Hence the Gram-error contribution is at most $\sqrt n\cdot O_{p}(1)\cdot O_{p}(K^{-m}\rho_{n})=O_{p}\bigl((\sqrt nK^{-m})\rho_{n}\bigr)=o_{p}(1)$ by Assumptions~\ref{assu:secondstage}(ii)--(iii) together with \ref{assu:np-firststage}(i) for the $L$ part of $\rho_{n}$. Collecting \S1--\S3 and \eqref{eq: np FWL}, and multiplying by $\sqrt{n}$,
\[
\sqrt{n}(\hat\beta-\beta_0)=\hat A^{-1}\Bigl[\frac1{\sqrt n}\sum_iD_iv_i\varepsilon_i-\frac1{\sqrt n}\sum_i\phi_i\delta_i+\sqrt n R_{2,n}\Bigr]+o_p(1),
\]
where $R_{2,n}$ is the first-order Taylor remainder of Lemma~\ref{lem: np quadratic}. With $\hat A\to_p A$ (\S1) and $A$ nonsingular, Slutsky's theorem replaces $\hat A^{-1}$ by $A^{-1}$, completing the lemma once $\sqrt n R_{2,n}=o_p(1)$.
\end{proof}

\subsubsection*{Lemma D: first-order Taylor remainder}

\begin{lemma}\label{lem: np quadratic}
Under Assumptions~\ref{assu:np-firststage}(i)--(ii),
\ref{assu:hfunc}(i), \ref{assu:secondstage}, and the overlap condition of
Proposition~\ref{prop: np primitives}, the first-order remainder
\[
R_{2,n}:=\frac1n\sum_iD_i\hat v_i
\{\lambda_0(p_{0i})-\lambda_0(\hat p_i)
+\lambda_0'(p_{0i})\delta_i\}
\]
satisfies $\sqrt{n}\,R_{2,n}=o_{p}(1)$, using only the $L^{2}$
first-stage rate.
\end{lemma}

\begin{proof}
By the fundamental theorem of calculus and Assumption~\ref{assu:hfunc}(i),
\[
\left|\lambda_0(p_{0i})-\lambda_0(\hat p_i)
+\lambda_0'(p_{0i})\delta_i\right|
=\left|\int_{p_{0i}}^{\hat p_i}
\{\lambda_0'(t)-\lambda_0'(p_{0i})\}\,dt\right|
\leq\frac{\bar L_2}{2}\delta_i^2.
\]
Hence the remainder is controlled by the weighted form
\[
\lVert\sqrt{n}\,R_{2,n}\rVert\le\frac{\bar L_2}{2}\,\sqrt{n}\,\frac1n\sum_iD_i\lVert\hat{v}_i\rVert\delta_i^{2}\le\frac{\bar L_2}{2}\,\Bigl(\sup_i\lVert\hat{v}_i\rVert\Bigr)\,\sqrt{n}\,\frac1n\sum_i\delta_i^{2}.
\]
The empirical squared norm is bounded \emph{directly} from the series normal
equations: by Lemma~\ref{lem: np first stage rates}(b),
$n^{-1}\sum_i\delta_i^{2}=O_p(L/n+\Delta_{p}^{2})=o_p(n^{-1/2})$.
This uses the empirical $L^{2}$ norm, not a conditional-mean argument; such a
conditioning argument is unavailable because $\hat{p}$ is estimated from the
same $D_i$. The residual is uniformly stable,
$\sup_i\lVert\hat v_i\rVert=O_p(1)$, by the generated-index stability imposed
in Assumption~\ref{assu:secondstage}(v). Therefore
$\lVert\sqrt n R_{2,n}\rVert\le\tfrac{\bar L_2}{2}\,O_p(1)\cdot
\sqrt n\,o_p(n^{-1/2})=o_p(1)$. Thus \emph{only the empirical $L^{2}$
first-stage rate is consumed}; no uniform $n^{-1/4}$ rate, hence no
strengthened leverage, is needed for the remainder. If one prefers not to
posit Assumption~\ref{assu:secondstage}(v), the same conclusion follows from
an $L^{4}$ first-stage rate in its place: by Cauchy--Schwarz,
$\sqrt{n}\,n^{-1}\sum_iD_i\lVert\hat{v}_i\rVert\delta_i^{2}
\le\bigl(n^{-1}\sum_iD_i\lVert\hat{v}_i\rVert^{2}\bigr)^{1/2}
\sqrt{n}\,\bigl(n^{-1}\sum_i\delta_i^{4}\bigr)^{1/2}
=O_{p}(1)\cdot\sqrt{n}\,\lVert\delta\rVert_{L^{4}}^{2}$, since
$n^{-1}\sum_iD_i\lVert\hat{v}_i\rVert^{2}
=\operatorname{tr}(\hat{A})=O_{p}(1)$; this is $o_{p}(1)$ whenever
$\lVert\delta\rVert_{L^{4}}=o_{p}(n^{-1/4})$, trading the uniform sup-norm
stability of~(v) for an $L^{4}$ first-stage rate. The first-derivative
substitution error
$n^{-1}\sum_iD_i\hat v_i(\lambda_0'(\hat p_i)-\lambda_0'(p_{0i}))
(p_{0i}-\hat p_i)$ is also bounded by $\bar L_2\delta_i^2$ and is handled
identically. Finiteness of $\Omega_{\chi}$ plays no role here: the remainder
is governed by the Lipschitz modulus of $\lambda_0'$ and the $L^{2}$ rate,
not by the integrated representer norm.
\end{proof}

\subsubsection*{Lemma B: Riesz representation of the linear functional}

\begin{lemma}\label{lem: np riesz}
Let $\phi_i\in\mathbb{R}^{d}$ satisfy $E[\phi_i\mid X_i]=\alpha_{0}(X_i)$ and $\sup_{x}E[\lVert\phi_i\rVert^{2}\mid X_i=x]<\infty$. Under Assumptions~\ref{assu:np-firststage}(i)--(iii),
\[
\sqrt{n}\,\frac1n\sum_i\phi_i\delta_i=\frac1{\sqrt{n}}\sum_i\alpha_{0}(X_i)(D_i-p_{0i})+o_{p}(1)=\frac1{\sqrt{n}}\sum_i\chi_i+o_{p}(1).
\]
\end{lemma}

\begin{proof}
\emph{Normal-equation expansion.} Since $\hat{\pi}=\hat{G}^{-1}n^{-1}\sum_i\psi_{L}(X_i)D_i$, $\pi_{0}=G^{-1}E[\psi_{L}p_{0}]$, and $\psi_{L}'\pi_{0}$ lies in the span,
\begin{equation}\label{eq: np pi expansion}
\hat{\pi}-\pi_{0}=\hat{G}^{-1}\bigl[\,\underbrace{\tfrac1n\textstyle\sum_i\psi_{L}(X_i)(D_i-p_{0i})}_{S_{n}}+\underbrace{\tfrac1n\textstyle\sum_i\psi_{L}(X_i)r_{L}(X_i)}_{T_{n}}\,\bigr].
\end{equation}
The summands of $S_{n}$ are i.i.d.\ mean zero with $\mathrm{Var}(n^{1/2}S_{n})=E[(D-p_{0})^{2}\psi_{L}\psi_{L}']\preceq\tfrac14 G$ (using $E[(D-p_{0})^{2}\mid X]=p_{0}(1-p_{0})\le\tfrac14$), so
\begin{equation}\label{eq: np Sn rate}
E\lVert S_{n}\rVert^{2}=n^{-1}\mathrm{tr}\,E[(D-p_{0})^{2}\psi_{L}\psi_{L}']\le\bar{\lambda}_{G}L/(4n),\qquad\lVert S_{n}\rVert=O_{p}(\sqrt{L/n}).
\end{equation}
With $\hat p(X_i)-p_{0}(X_i)=\psi_{L}(X_i)'(\hat\pi-\pi_{0})-r_{L}(X_i)$ and $\hat\Gamma:=n^{-1}\sum_i\phi_i\psi_{L}(X_i)'$ (so $E[\hat\Gamma]=\Gamma$),
\[
\sqrt n\,\frac1n\sum_i\phi_i\delta_i=\underbrace{\sqrt n\,\hat\Gamma\hat G^{-1}S_{n}}_{\mathrm{(I)}}+\underbrace{\sqrt n\,\hat\Gamma\hat G^{-1}T_{n}}_{\mathrm{(II)}}-\underbrace{\sqrt n\,\frac1n\sum_i\phi_ir_{L}(X_i)}_{\mathrm{(III)}}.
\]

\emph{Term (I): the score channel.} Split $\mathrm{(I)}=\sqrt n\,\Gamma G^{-1}S_{n}+\sqrt n(\hat\Gamma-\Gamma)G^{-1}S_{n}+\sqrt n\,\hat\Gamma(\hat G^{-1}-G^{-1})S_{n}=:\mathrm{(I_a)}+\mathrm{(I_b)}+\mathrm{(I_c)}$.

\noindent$\bullet$ \emph{$\mathrm{(I_a)}$.} The projection identity $\Gamma G^{-1}\psi_{L}=\alpha_{0,L}$ gives $\mathrm{(I_a)}=n^{-1/2}\sum_i\alpha_{0,L}(X_i)(D_i-p_{0i})$. Writing $\alpha_{0,L}=\alpha_{0}-(\alpha_{0}-\alpha_{0,L})$, the first part is $n^{-1/2}\sum_i\chi_i$ and the second is mean zero with variance $E[(\alpha_{0}-\alpha_{0,L})(X)(\alpha_{0}-\alpha_{0,L})(X)'(D-p_{0})^{2}]\preceq\tfrac14\lVert\alpha_{0}-\alpha_{0,L}\rVert_{L^{2}}^{2}I=\tfrac14\Delta_{\alpha}^{2}I\to0$ by Assumption~\ref{assu:np-firststage}(ii). Hence $\mathrm{(I_a)}=n^{-1/2}\sum_i\chi_i+o_{p}(1)$.

\noindent$\bullet$ \emph{$\mathrm{(I_b)}$.} $E\lVert\hat\Gamma-\Gamma\rVert_{F}^{2}\le n^{-1}E[\lVert\phi\rVert^{2}\lVert\psi_{L}\rVert^{2}]\le n^{-1}\zeta^{2}\sup_xE[\lVert\phi\rVert^{2}\mid X{=}x]=O(\zeta^{2}/n)$ by the moment hypothesis on $\phi_{i}$; with $\lVert G^{-1}\rVert_{op}=O(1)$ and $\lVert S_{n}\rVert=O_{p}(\sqrt{L/n})$, $\lVert\mathrm{(I_b)}\rVert\le\sqrt n\,\lVert\hat\Gamma-\Gamma\rVert_{F}\lVert G^{-1}\rVert\lVert S_{n}\rVert=O_{p}(\sqrt n\cdot\zeta/\sqrt n\cdot\sqrt{L/n})=O_{p}(\sqrt{\zeta^{2}L/n})=o_{p}(1)$ by Assumption~\ref{assu:np-firststage}(i).

\noindent$\bullet$ \emph{$\mathrm{(I_c)}$, the Gram-error term.} Using $\hat G^{-1}-G^{-1}=-\hat G^{-1}(\hat G-G)G^{-1}$, inserting $G^{-1/2}G^{1/2}$ between factors and noting $G^{1/2}G^{-1}S_n=G^{-1/2}S_n$,
\[
\lVert\mathrm{(I_c)}\rVert\le\sqrt n\,\lVert\hat\Gamma G^{-1/2}\rVert_{op}\,\lVert G^{1/2}\hat G^{-1}G^{1/2}\rVert_{op}\,\bigl\lVert G^{-1/2}(\hat G-G)G^{-1/2}\bigr\rVert_{op}\,\lVert G^{-1/2}S_{n}\rVert.
\]
Here $\lVert\hat\Gamma G^{-1/2}\rVert_{op}=\lVert\alpha_{0,L}\rVert_{L^{2}}+o_p(1)\le\lVert\alpha_{0}\rVert_{L^{2}}+o_p(1)=O_{p}(1)$ by the moment hypothesis and conditional Jensen; $\lVert G^{1/2}\hat G^{-1}G^{1/2}\rVert_{op}=O_{p}(1)$ by Lemma~\ref{lem: np first stage rates}(a); $\lVert G^{-1/2}S_{n}\rVert=O_{p}(\sqrt{L/n})$ by \eqref{eq: np Sn rate}; and $\lVert G^{-1/2}(\hat G-G)G^{-1/2}\rVert_{op}=O_{p}(\sqrt{\zeta^{2}\log L/n})$ by the matrix Bernstein inequality \citep[Theorem~6.1 of][]{tropp2012user}. Therefore $\lVert\mathrm{(I_c)}\rVert=O_{p}(\sqrt n\cdot\sqrt{\zeta^{2}\log L/n}\cdot\sqrt{L/n})=O_{p}(\sqrt{\zeta^{2}L\log L/n})$, which is $o_{p}(1)$ by the leverage condition $\zeta^{2}L\log L/n\to0$ of Assumption~\ref{assu:np-firststage}(i). This is the step that consumes the leverage: the unweighted score channel carries no approximation factor to damp the Gram error, so the $\sqrt{L}$ from the score norm $\lVert S_n\rVert$ must be absorbed by the leverage rate (it is the sole reason part~(i) is stated with the factor $L$ rather than the bare invertibility condition $\zeta^{2}\log L/n\to0$).

\emph{Terms (II)$-$(III): the bias channel and the double-robustness cancellation.} Replacing $\hat\Gamma\hat G^{-1}$ by $\Gamma G^{-1}$ in (II) changes it by $o_p(1)$, by the same Gram- and sampling-error bounds as $\mathrm{(I_b)},\mathrm{(I_c)}$ applied to $T_n$ in place of $S_n$ (the substitution is further damped because $E[\psi_{L}r_{L}]=0$ gives $\lVert T_n\rVert=O_p(\zeta\Delta_p/\sqrt n)$). The leading part of (II) is $\sqrt n\,\Gamma G^{-1}T_{n}=\sqrt n\,n^{-1}\sum_i\alpha_{0,L}(X_i)r_{L}(X_i)$ and the leading part of (III) is $\sqrt n\,n^{-1}\sum_i\alpha_{0}(X_i)r_{L}(X_i)$. Their population means are $\sqrt n\,E[\alpha_{0,L}r_{L}]$ and $\sqrt n\,E[\alpha_{0}r_{L}]$; since $\alpha_{0,L}\in\mathrm{span}(\psi_L)$ and $r_{L}\perp\mathrm{span}(\psi_L)$, $E[\alpha_{0,L}r_{L}]=0$, while $E[\alpha_{0}r_{L}]=E[(\alpha_{0}-\alpha_{0,L})r_{L}]$, so
\[
\text{(mean of (II)$-$(III))}=-\sqrt n\,E[(\alpha_{0}-\alpha_{0,L})r_{L}],\qquad\bigl\lVert\sqrt n\,E[(\alpha_{0}-\alpha_{0,L})r_{L}]\bigr\rVert\le\sqrt n\,\Delta_{\alpha}\Delta_{p}\to0
\]
by Cauchy--Schwarz and the product undersmoothing Assumption~\ref{assu:np-firststage}(iii). This exact cancellation of the $\sqrt n\langle\alpha_{0,L},r_{L}\rangle$ terms is the double-robustness payoff: only the product $\Delta_{\alpha}\Delta_{p}$, not $\Delta_{p}$ alone, must vanish at the parametric rate. The centred parts of (II) and (III) are mean-zero averages. The centred part of (III) has variance $\mathrm{Var}(\alpha_{0}(X)r_{L}(X))\le\sup_xE[\lVert\phi\rVert^{2}\mid X{=}x]\,\Delta_{p}^{2}=O(\Delta_{p}^{2})\to0$ by the moment hypothesis and Assumption~\ref{assu:np-firststage}(ii); the centred part of (II) has variance $\le\lVert\alpha_{0,L}\rVert_{\infty}^{2}\Delta_{p}^{2}\le\underline{\lambda}_{G}^{-1}\zeta^{2}\lVert\alpha_{0}\rVert_{L^{2}}^{2}\Delta_{p}^{2}=O(\zeta^{2}\Delta_{p}^{2})=o(1)$, where $\lVert\alpha_{0,L}\rVert_{\infty}\le\zeta\lVert G^{-1}E[\psi_{L}\alpha_{0}]\rVert\le\underline{\lambda}_{G}^{-1/2}\zeta\lVert\alpha_{0}\rVert_{L^{2}}$ uses $\lVert\psi_{L}'c\rVert_{\infty}\le\zeta\lVert c\rVert$ and the Gram lower bound, and the final $o(1)$ is the stable-basis approximation condition $\zeta(L)\Delta_p(L)\to0$ in Assumption~\ref{assu:np-firststage}(ii). Hence $\mathrm{(II)}-\mathrm{(III)}=o_{p}(1)$, and combining with Term~(I), $\sqrt n\,n^{-1}\sum_i\phi_i\delta_i=n^{-1/2}\sum_i\chi_i+o_{p}(1)$.
\end{proof}

\subsubsection*{Completion of the proof}

We verify the high-level conditions of Assumption~\ref{assu:np-hl} from the primitives. Part~(i): $\lVert\alpha_{0}(x)\rVert=p_{0}(x)\lvert\lambda_{0}'(p_{0}(x))\rvert\lVert v(x)\rVert\le\sup_{t\in[0,1]}\lvert\lambda_{0}'(t)\rvert\,\sup_{x\in\mathcal{X}}\lVert v(x)\rVert<\infty$ by Assumptions~\ref{assu:hfunc}(i) (bounded $\lambda_{0}'$) and~\ref{assu:IID}(ii) (compact $\mathcal{X}$, so $v$ is bounded), hence $\alpha_{0}\in L^{2}(F_{X})$; $A$ is nonsingular by Assumption~\ref{assu:Dv*}(i). Part~(iv) is established in Lemma~\ref{lem: np second stage}: its \S1--\S2 give $\hat{A}\to_p A$ and the oracle term $n^{-1/2}\sum_iD_i(\hat{v}_i-v_i)\varepsilon_i=o_p(1)$, and its \S3 gives the second-stage spline-approximation and empirical-projection contributions as $o_p(n^{-1/2})$ (Assumption~\ref{assu:secondstage}(ii)--(v)). Part~(iii) is Lemma~\ref{lem: np quadratic}. Part~(ii): by Lemma~\ref{lem: np second stage}\,\S3 the first-stage linear term satisfies $n^{-1}\sum_iD_i\hat{v}_i\lambda_{0}'(p_{0i})\delta_i=n^{-1}\sum_i\phi_i\delta_i+o_p(n^{-1/2})$ (the $\hat{v}_i\to v_i$ substitution), and Lemma~\ref{lem: np riesz}, whose hypotheses hold for $\phi_i=D_iv_i\lambda_{0}'(p_{0i})$ since $E[\lVert\phi_i\rVert^{2}\mid X_i=x]=p_{0}(x)\lambda_{0}'(p_{0}(x))^{2}\lVert v(x)\rVert^{2}$ is bounded as in Part~(i), gives $\sqrt{n}\,n^{-1}\sum_i\phi_i\delta_i=n^{-1/2}\sum_i\chi_i+o_p(1)$, whence $\sqrt{n}\,n^{-1}\sum_iD_i\hat{v}_i\lambda_{0}'(p_{0i})\delta_i=n^{-1/2}\sum_i\chi_i+o_p(1)$. Thus Assumption~\ref{assu:np-hl} holds, and Proposition~\ref{prop: asymptotics np} applies.\hfill$\qed$

\section{Proofs of Propositions~\refPropEffBound{} and~\refPropEffEst}\label{app: efficiency proof}

Throughout this section we use the column-vector convention of Appendix~\ref{app: np proof}: $X_i,v_i,\alpha_{0}\in\mathbb{R}^{d}$ are columns, $A=E[D_iv_iv_i']$, and $\varepsilon_i:=Y_i-X_i'\beta_{0}-\lambda_{0}(p_{0i})$. Write $\mu_{0}(x):=x'\beta_{0}+\lambda_{0}(p_{0}(x))$ for the outcome conditional mean, $\sigma_{0}^{2}(x):=\mathrm{Var}(Y_i\mid X_i=x,D_i=1)$, and recall from Section~\ref{subsec: efficiency} the composite residual, composite variance, efficient weight, and tilted projection
\[
\tilde{u}_i:=D_i\varepsilon_i-p_{0i}\lambda_{0}'(p_{0i})(D_i-p_{0i}),\quad
\kappa^{2}(x):=p_{0}\sigma_{0}^{2}+p_{0}^{3}(1-p_{0})\lambda_{0}'(p_{0})^{2},\quad
w_{0}(x):=\frac{p_{0}(x)}{\kappa^{2}(x)},
\]
\[
\tilde{E}[X\mid p_{0}]:=\frac{E[(p_{0}^{2}/\kappa^{2})X\mid p_{0}(X)]}{E[p_{0}^{2}/\kappa^{2}\mid p_{0}(X)]},\qquad
\tilde{v}_i:=X_i-\tilde{E}[X\mid p_{0}](p_{0i}),\qquad
V^{*}:=\Bigl(E\bigl[\tfrac{p_{0i}^{2}}{\kappa^{2}(X_i)}\,\tilde{v}_i\tilde{v}_i'\bigr]\Bigr)^{-1}.
\]
The bound argument follows the projection approach to semiparametric efficiency \citep{bickel1993efficient,newey1990semiparametric}, and specializes the efficiency-bound analysis of sample-selection models with a nonparametric selection probability in \cite{newey1993efficiency} to the no-exclusion design. Two elementary facts are used repeatedly. First, $E[\tilde{u}_i\mid X_i]=0$ and
\begin{equation}\label{eq: composite variance}
E[\tilde{u}_i^{2}\mid X_i]=E[D_i\varepsilon_i^{2}\mid X_i]+p_{0i}^{2}\lambda_{0}'(p_{0i})^{2}E[(D_i-p_{0i})^{2}\mid X_i]=p_{0i}\sigma_{0}^{2}(X_i)+p_{0i}^{3}(1-p_{0i})\lambda_{0}'(p_{0i})^{2}=\kappa^{2}(X_i),
\end{equation}
since the cross term vanishes: $E[D_i\varepsilon_i(D_i-p_{0i})\mid X_i]=(1-p_{0i})E[D_i\varepsilon_i\mid X_i]=0$. Second, for any bounded weight $\omega(x)>0$ that is a function of $X$,
\begin{equation}\label{eq: weighted projection identity}
\frac{E[D_i\omega(X_i)X_i\mid p_{0i}]}{E[D_i\omega(X_i)\mid p_{0i}]}=\frac{E[p_{0i}\omega(X_i)X_i\mid p_{0i}]}{E[p_{0i}\omega(X_i)\mid p_{0i}]},
\end{equation}
by iterating $E[D_i\mid X_i]=p_{0i}$; in particular $E[X_i\mid p_{0i},D_i=1]=E[X_i\mid p_{0i}]$, and $\tilde{E}[X\mid p_{0}]$ is the $w_{0}$-weighted projection of $X$ on functions of $p_{0}$ in the selected population, because $D w_{0}$ has conditional mean $p_{0}w_{0}=p_{0}^{2}/\kappa^{2}$ given $X$.

\subsection{The tangent set}\label{subsec: tangent}

The model $\mathcal{P}$ of Section~\ref{subsec: efficiency} restricts the distribution $P$ of $(X_i,D_i,D_iY_i)$ by: (M1) $p_{0}=E[D\mid X]$ is unrestricted subject to overlap $p_{0}\in[\underline{c},\bar{c}]\subset(0,1)$; (M2) $E[Y\mid X,D=1]=X'\beta+\lambda(p_{0}(X))$ for some $\beta\in\mathbb{R}^{d}$ and $\lambda\in C^{1}_{b}$, where $C_b^1$ denotes the class of continuously differentiable functions on $[0,1]$ with bounded derivative; (M3) the conditional law of $Y$ given $(X,D=1)$ and the marginal of $X$ are otherwise unrestricted, with $\sigma_{0}^{2}$ bounded and bounded away from zero. $Y$ is unobserved when $D=0$, so the likelihood of one observation factors as
\[
dP=f_{X}(x)\,p_{0}(x)^{d}\left(1-p_{0}(x)\right)^{1-d}\,\bigl[f_{Y\mid X,D=1}(y\mid x)\bigr]^{d}.
\]

\begin{lemma}[Tangent set and pathwise derivative]\label{lem: eff tangent}
Consider regular parametric submodels $t\mapsto P_{t}\in\mathcal{P}$ through $P_{0}=P$ generated by perturbations $p_{t}=p_{0}+th$, $\beta_{t}=\beta_{0}+t\delta$, $\lambda_{t}=\lambda_{0}+tg$, a mean-square differentiable path of $f_{Y\mid X,D=1}$ whose conditional mean tracks $x'\beta_{t}+\lambda_{t}(p_{t}(x))$, and a path of $f_{X}$, where $h$ is bounded (bounded $h$ are dense in $L^{2}(F_{X})$, which suffices for all conclusions below), $\delta\in\mathbb{R}^{d}$, and $g\in C^{1}_{b}$ are otherwise free, and the truth lies interior to the overlap interval so that $p_{t}\in(0,1)$ and identification is preserved for small $t$. The score of such a path is
\begin{equation}\label{eq: tangent score}
s=a(X)(D-p_{0})+D\,\ell(Y,X)+u(X),
\end{equation}
with $a=h/\{p_{0}(1-p_{0})\}$, $E[u(X)]=0$, $E[\ell\mid X,D=1]=0$, and
\begin{equation}\label{eq: mean shift}
E[\ell\,(Y-\mu_{0})\mid X=x,D=1]=\Delta(x):=x'\delta+g(p_{0}(x))+\lambda_{0}'(p_{0}(x))h(x);
\end{equation}
conversely every $(a,\ell,u)$ with $\Delta_{\ell}(x):=E[\ell(Y-\mu_{0})\mid x,D=1]$ of the form \eqref{eq: mean shift} for some $(\delta,g)$ arises from such a path. Along any such path, $\beta(P_{t})$ is differentiable with $\dot{\beta}=\delta$, and $(\delta,g)$ are uniquely determined by $(h,\Delta)$ through the decomposition $\Delta-\lambda_{0}'h=x'\delta+g(p_{0})$, which is unique under Assumption~\ref{assu:Dv*}(i).
\end{lemma}

\begin{proof}
The score decomposition \eqref{eq: tangent score} follows by differentiating the log-likelihood along the path: the Bernoulli factor contributes $h(X)(D-p_{0})/\{p_{0}(1-p_{0})\}$, the outcome factor contributes $D$ times the $L^{2}$-derivative $\ell$ of $\log f_{Y\mid X,D=1}$, and the $f_{X}$ factor contributes $u(X)$. The model constraint (M2) holds along the path if and only if the conditional mean of $f_{t}(\cdot\mid x)$ equals $x'\beta_{t}+\lambda_{t}(p_{t}(x))=\mu_{0}(x)+t\Delta(x)+o(t)$ with $\Delta$ as in \eqref{eq: mean shift}; differentiating the mean identity $\int y f_{t}(y\mid x)dy=\mu_{0}(x)+t\Delta(x)+o(t)$ gives $E[\ell(Y-\mu_{0})\mid x,D=1]=\Delta(x)$. For the converse, given $(a,\ell,u)$ with $\Delta_{\ell}$ of the stated form, take $h=a\,p_{0}(1-p_{0})$, verify $p_{t}\in(\underline{c}',\bar{c}')$ for small $t$, and construct the corrected exponential-tilt path $f_{t}\propto f_{0}\exp\bigl(t\ell+a_{t}(x)(y-\mu_{0}(x))\bigr)$ with $a_{t}=O(t^{2})$ chosen so that the conditional mean tracks $x'\beta_{t}+\lambda_{t}(p_{t}(x))$ exactly (the $O(t^{2})$ adjustment changes neither the score nor $\dot{\beta}$); this realizes the score $\ell$ and the mean shift $\Delta_{\ell}$.

For the pathwise derivative: under $P_{t}$ the observed conditional mean is $\mu_{t}(x)=\mu_{0}(x)+t\Delta(x)+o(t)$ and the propensity is $p_{t}$. The parameter $\beta(P_{t})$ is defined by the unique decomposition $\mu_{t}(x)=x'\beta(P_{t})+\lambda(p_{t}(x))$. Substituting, $x'(\beta(P_{t})-\beta_{0})+\lambda(p_{t}(x))-\lambda_{0}(p_{0}(x))=t\Delta(x)+o(t)$; expanding $\lambda_{0}(p_{0})=\lambda_{0}(p_{t})-t\lambda_{0}'(p_{0})h+o(t)$ and collecting first-order terms in the decomposition $x'\delta+g(p_{0})=\Delta-\lambda_{0}'h$, uniqueness of the decomposition (no nonzero $x'c$ is a.s.\ equal to a function of $p_{0}(x)$, by Assumption~\ref{assu:Dv*}(i) and \eqref{eq: weighted projection identity}) yields $\beta(P_{t})=\beta_{0}+t\delta+o(t)$.
\end{proof}

Two special cases are noteworthy. First, pure first-stage paths, those holding $f_{Y\mid X,D=1}$ completely fixed ($\ell=0$, $\Delta=0$), exist whenever $\lambda_{0}'$ is bounded away from zero on a neighbourhood of the support of $p_{0}$ (so that $\lambda_{t}$ is invertible there for small $t$): for any $(\delta,g)$, setting $p_{t}=\lambda_{t}^{-1}\bigl(\mu_{0}(x)-x'\beta_{t}\bigr)$ with $\beta_{t}=\beta_{0}+t\delta$, $\lambda_{t}=\lambda_{0}+tg$ stays in the model exactly, with $h=-\{x'\delta+g(p_{0})\}/\lambda_{0}'(p_{0})$ and $\dot{\beta}=\delta$; within this family, $\dot{\beta}=0$ if and only if $\delta=0$, i.e.\ $h=-g(p_{0})/\lambda_{0}'(p_{0})$ is a function of the index alone (the index-preserving reparametrizations $\lambda_{t}=\lambda_{0}\circ\varphi_{t}^{-1}$, $\varphi_{t}(p)=p+tH(p)$). Thus, when selection on unobservables is present, $\beta$ is locally estimable off selection-equation variation alone, the selection data are genuinely informative about $\beta$, which is why every influence function below must carry a nondegenerate $(D-p_{0})$-component. Second, for a pre-specified $h$ outside this $d$-dimensional-plus-$\{g\}$ family, the perturbation must be accompanied by an outcome-mean adjustment, and the general score always carries the coupling term $\lambda_{0}'h$ inside $\Delta$; it is this coupling that pins the relative weight of the two channels in every influence function.

\subsection{Characterization of influence functions and the bound}

\begin{lemma}[Influence characterization]\label{lem: eff influence}
Let $\psi\in L^{2}(P)^{d}$ satisfy $E[\psi]=0$ and $E[\psi s]=\dot{\beta}(s)$ for every score $s$ of Lemma~\ref{lem: eff tangent}. Then
\begin{equation}\label{eq: influence form}
\psi=\tilde{B}(X)\,\tilde{u},\qquad \tilde{u}=D(Y-\mu_{0})-p_{0}\lambda_{0}'(p_{0})(D-p_{0}),
\end{equation}
where $\tilde{B}:\mathcal{X}\to\mathbb{R}^{d}$ satisfies
\begin{equation}\label{eq: Btilde constraints}
E[p_{0}\,\tilde{B}(X)X']=I_{d},\qquad E[p_{0}\,\tilde{B}(X)\mid p_{0}(X)]=0.
\end{equation}
Conversely, every $\psi=\tilde{B}\tilde{u}$ with $\tilde{B}$ satisfying \eqref{eq: Btilde constraints} meets $E[\psi s]=\dot{\beta}(s)$ for all tangent scores. The class $\{\tilde{B}\tilde{u}:\eqref{eq: Btilde constraints}\}$ therefore exhausts the influence functions of regular asymptotically linear estimators of $\beta_{0}$.
\end{lemma}

\begin{proof}
Because the observed data are $(X,D,DY)$, every mean-zero $\psi\in L^{2}(P)$ decomposes \emph{exactly}, not merely up to a remainder, into the three mutually orthogonal blocks
\[
\psi=d(X)+C(X)(D-p_{0})+D\,q(Y,X),\qquad E[q\mid X,D=1]=0,
\]
with $d=E[\psi\mid X]$, $C=E[\psi\mid X,D{=}1]-\psi(X,0)$, and $q=\psi(X,1,Y)-E[\psi\mid X,D{=}1]$, where $\psi(X,0)$ and $\psi(X,1,Y)$ denote the values of $\psi$ on $\{D=0\}$ and $\{D=1\}$; the identity is verified by evaluating both sides on each event. It therefore suffices to pin down $(d,C,q)$ from the inner-product conditions.

\emph{Shape directions.} For any $\ell$ with $\Delta_{\ell}=0$ (mean-preserving) the score $D\ell$ is tangent with $\dot\beta=0$ (take $h=0$, $\delta=0$, $g=0$), so $E[\psi D\ell]=E[p_{0}\,q\,\ell]$ must vanish for all such $\ell$; hence $q$ is orthogonal, within $\{E[q\mid X,D{=}1]=0\}$, to the mean-preserving subspace, which forces $q(Y,X)=\tilde{B}(X)(Y-\mu_{0})$ for some $\tilde{B}\in L^{2}$ (the orthocomplement of $\{\ell:\Delta_{\ell}=0\}$ within the centred space is spanned conditionally by $Y-\mu_{0}$).

\emph{Marginal directions.} Paths in $f_{X}$ alone have $\dot\beta=0$ and score $u(X)$, so $E[\psi u]=0$ for all mean-zero $u$, i.e.\ $E[\psi\mid X]=0$; given the other two components are conditionally mean zero, $d\equiv0$.

\emph{Mean and first-stage directions.} Fix $(h,\delta,g)$ and choose $\ell$ with $\Delta_{\ell}=\Delta=x'\delta+g(p_{0})+\lambda_{0}'h$. Using $E[(Y-\mu_{0})\ell\mid X,D=1]=\Delta$, $E[q\ell D\mid X]=p_{0}\tilde{B}\Delta$, and the orthogonality $E[C(D-p_{0})D\ell\mid X]=C(1-p_{0})p_{0}E[\ell\mid X,D{=}1]=0$,
\[
E[\psi s]=E\bigl[p_{0}\tilde{B}(X)\{X'\delta+g(p_{0})+\lambda_{0}'(p_{0})h\}\bigr]+E[C(X)h].
\]
Setting this equal to $\dot\beta=\delta$ for all $(\delta,g,h)$: the $\delta$-terms give $E[p_{0}\tilde{B}X']=I_{d}$; the $g$-terms give $E[p_{0}\tilde{B}g(p_{0})]=0$ for all $g$, i.e.\ $E[p_{0}\tilde{B}\mid p_{0}]=0$; and the $h$-terms give $E[\{p_{0}\lambda_{0}'\tilde{B}+C\}h]=0$ for all $h\in L^{2}(F_{X})$, i.e.\ $C=-p_{0}\lambda_{0}'(p_{0})\tilde{B}$ pointwise. Collecting, $\psi=\tilde{B}(X)\{D(Y-\mu_{0})-p_{0}\lambda_{0}'(p_{0})(D-p_{0})\}=\tilde{B}\tilde{u}$ exactly. The converse direction re-runs the same computations. The exhaustiveness claim follows because the influence function of any regular asymptotically linear estimator is a mean-zero element of $L^{2}(P)^{d}$ satisfying the inner-product conditions, hence of the form \eqref{eq: influence form}.
\end{proof}

\begin{lemma}[Constrained weighted projection]\label{lem: eff GLS}
Let $\kappa^{2}$ be bounded and bounded away from zero, and let $V^{*-1}=E[(p_{0}^{2}/\kappa^{2})\tilde{v}\tilde{v}']$ be nonsingular. Over all $\tilde{B}\in L^{2}$ satisfying \eqref{eq: Btilde constraints}, the matrix $E[\tilde{B}\tilde{B}'\kappa^{2}]$ is minimized (in the positive semidefinite order) uniquely at
\[
\tilde{B}^{*}(X)=V^{*}\,\frac{p_{0}(X)}{\kappa^{2}(X)}\,\tilde{v}(X),\qquad\text{with minimum }E[\tilde{B}^{*}\tilde{B}^{*\prime}\kappa^{2}]=V^{*}.
\]
\end{lemma}

\begin{proof}
First, $\tilde{B}^{*}$ is feasible: $E[p_{0}\tilde{B}^{*}X']=V^{*}E[(p_{0}^{2}/\kappa^{2})\tilde{v}X']=V^{*}E[(p_{0}^{2}/\kappa^{2})\tilde{v}\tilde{v}']=I$ (adding back $\tilde{E}[X\mid p_{0}]'$ contributes zero because $E[(p_{0}^{2}/\kappa^{2})\tilde{v}\mid p_{0}]=0$ by the definition of $\tilde{E}$), and $E[p_{0}\tilde{B}^{*}\mid p_{0}]=V^{*}E[(p_{0}^{2}/\kappa^{2})\tilde{v}\mid p_{0}]=0$. Next, for any feasible $\tilde{B}$ write $\tilde{B}=\tilde{B}^{*}+\rho$; feasibility of both implies $E[p_{0}\rho X']=0$ and $E[p_{0}\rho\mid p_{0}]=0$. Then
\[
E[\rho\,\kappa^{2}\tilde{B}^{*\prime}]=E\bigl[\rho\,p_{0}\tilde{v}'\bigr]V^{*}=E\bigl[\rho\,p_{0}\{X-\tilde{E}[X\mid p_{0}]\}'\bigr]V^{*}=\Bigl(E[p_{0}\rho X']-E\bigl[E[p_{0}\rho\mid p_{0}]\,\tilde{E}[X\mid p_{0}]'\bigr]\Bigr)V^{*}=0,
\]
so $E[\tilde{B}\tilde{B}'\kappa^{2}]=E[\tilde{B}^{*}\tilde{B}^{*\prime}\kappa^{2}]+E[\rho\rho'\kappa^{2}]\succeq E[\tilde{B}^{*}\tilde{B}^{*\prime}\kappa^{2}]$, with equality iff $\rho=0$ a.s.\ ($\kappa^{2}\ge\underline{\kappa}^{2}>0$). Finally $E[\tilde{B}^{*}\tilde{B}^{*\prime}\kappa^{2}]=V^{*}E[(p_{0}^{2}/\kappa^{2})\tilde{v}\tilde{v}']V^{*}=V^{*}$.
\end{proof}

\begin{proof}[Proof of Proposition~\ref{prop: efficiency bound}]
By Lemma~\ref{lem: eff influence}, every influence function of a regular asymptotically linear estimator equals $\tilde{B}\tilde{u}$ for a feasible $\tilde{B}$, with variance $E[\tilde{B}\tilde{B}'\kappa^{2}]$ by \eqref{eq: composite variance}. Hence $\mathrm{Var}(\psi)\succeq V^{*}$ by Lemma~\ref{lem: eff GLS}, and the efficient influence function is $\psi^{*}=\tilde{B}^{*}\tilde{u}=V^{*}w_{0}(X)\tilde{v}(X)\tilde{u}$, which attains $\mathrm{Var}(\psi^{*})=V^{*}$. It remains to verify $\psi^{*}\in\bar{\mathcal{T}}$, so that $V^{*}$ is also the convolution-theorem bound for all regular estimators. Note first that a candidate $B\tilde{u}$ is a tangent score if and only if $(\kappa^{2}/p_{0})B$ is partially linear in $(X,p_{0})$: matching $B'c\,\tilde{u}=D\ell+a(D-p_{0})$ requires $\ell=B'c(Y-\mu_{0})$ and $a=-p_{0}\lambda_{0}'B'c$, i.e.\ $h=-p_{0}^{2}(1-p_{0})\lambda_{0}'B'c$, and the coupling \eqref{eq: mean shift} then demands
\[
x'\delta+g(p_{0})=\Delta_{\ell}-\lambda_{0}'h=B'c\,\sigma_{0}^{2}+p_{0}^{2}(1-p_{0})\lambda_{0}'^{2}\,B'c=\frac{\kappa^{2}}{p_{0}}\,B'c.
\]
For $\tilde{B}^{*}=V^{*}(p_{0}/\kappa^{2})\tilde{v}$ this holds with $\delta=V^{*}c$ and $g=-\tilde{E}[X\mid p_{0}]'V^{*}c$, since $(\kappa^{2}/p_{0})\tilde{B}^{*\prime}c=\tilde{v}'V^{*}c=x'V^{*}c-\tilde{E}[X\mid p_{0}]'V^{*}c$; hence $\psi^{*}\in\bar{\mathcal{T}}$. (For a general feasible $\tilde{B}$ the partial-linearity requirement fails, which is why not every regular estimator is efficient; this membership criterion reappears as the equality condition in part (a).) This proves the bound.

\emph{Part (a).} The two-step estimator of Proposition~\ref{prop: asymptotics np} has influence $A^{-1}(D v\varepsilon-\chi)=A^{-1}v\,\tilde{u}$, i.e.\ $\tilde{B}_{2\mathrm{s}}=A^{-1}v$ with $v=X-E[X\mid p_{0}]$; it is feasible for \eqref{eq: Btilde constraints} since $E[p_{0}vX']=E[p_{0}vv']=A$ and $E[p_{0}v\mid p_{0}]=p_{0}(E[X\mid p_{0}]-E[X\mid p_{0}])=0$. Hence $V_{\mathrm{NP}}=E[\tilde{B}_{2\mathrm{s}}\tilde{B}_{2\mathrm{s}}'\kappa^{2}]=A^{-1}E[vv'\kappa^{2}]A^{-1}\succeq V^{*}$, which also re-derives $V_{\mathrm{NP}}=A^{-1}(\Omega+\Omega_{\chi})A^{-1}$ via \eqref{eq: composite variance}. By the uniqueness in Lemma~\ref{lem: eff GLS}, equality holds iff $A^{-1}v=\tilde{B}^{*}$ a.s., which (multiplying by $\kappa^{2}/p_{0}$ and applying $A$) is equivalent to the existence of a constant matrix $L$ and a function $\mu(\cdot)$ with
\begin{equation}\label{eq: equality FOC}
\frac{\kappa^{2}(X)}{p_{0}(X)}\,v(X)=LX+\mu(p_{0}(X))\quad\text{a.s.,}
\end{equation}
i.e.\ iff the two-step influence $A^{-1}v\,\tilde{u}$ itself lies in the tangent closure (the membership criterion established in the first paragraph of this proof). For the equivalence, the forward direction takes $L=AV^{*}$ and $\mu=-AV^{*}\tilde{E}[X\mid p_{0}]$; conversely, given \eqref{eq: equality FOC}, $A^{-1}v=(p_{0}/\kappa^{2})\{\tilde{L}X+\tilde{\mu}(p_{0})\}$ with $\tilde{L}=A^{-1}L$, and the two feasibility constraints pin $(\tilde{L},\tilde{\mu})$: $E[p_{0}A^{-1}v\mid p_{0}]=0$ forces $\tilde{\mu}=-\tilde{L}\tilde{E}[X\mid p_{0}]$, and then $E[p_{0}A^{-1}vX']=I$ gives $\tilde{L}E[(p_{0}^{2}/\kappa^{2})\tilde{v}\tilde v']=I$, i.e.\ $\tilde{L}=V^{*}$, so $A^{-1}v=\tilde{B}^{*}$. A transparent sufficient condition is that $\kappa^{2}/p_{0}$ be a.s.\ constant, say $=c$: then $p_{0}^{2}/\kappa^{2}=p_{0}/c$, so by \eqref{eq: weighted projection identity} $\tilde{E}[X\mid p_{0}]=E[X\mid p_{0}]$, $\tilde{v}=v$, and $\tilde{B}^{*}=V^{*}c^{-1}v$ with $V^{*}=cA^{-1}$, giving $\tilde{B}^{*}=A^{-1}v$. Constancy is not necessary in complete generality; \eqref{eq: equality FOC} can hold with nonconstant $\kappa^{2}/p_{0}$ in designs whose within-level-set support of $X$ is disconnected in a special way, but it is necessary in the leading case: if $\sigma_{0}^{2}(X)$ is a function of $p_{0}(X)$ alone (as under index-structured errors), so that $\kappa^{2}/p_{0}=\sigma_{0}^{2}(p_{0})+p_{0}^{2}(1-p_{0})\lambda_{0}'(p_{0})^{2}=:c(p_{0})$, then equality reads $\{c(p_{0}(X))A^{-1}-V^{*}\}v(X)=0$ a.s., so $c(p_{0})A^{-1}=V^{*}$, and hence $c(\cdot)$ is constant, on the set of $p_{0}$-values at which $\mathrm{Var}(v\mid p_{0})$ is nondegenerate in every direction. In particular, under homoskedasticity ($\sigma_{0}^{2}$ constant), $c(p)=\sigma^{2}+p^{2}(1-p)\lambda_{0}'(p)^{2}$ is nonconstant, and the unweighted two-step estimator remains inefficient wherever $\mathrm{Var}(X\mid p_{0})$ is nondegenerate, for every $\lambda_{0}$ with $\lambda_{0}'\neq0$ \emph{except} the one-parameter family $\lambda_{0}'(p)^{2}=\text{const}/\{p^{2}(1-p)\}$ on the support of $p_{0}$, the single smooth shape for which the composite variance is proportional to $p_{0}$.

\emph{Part (b).} When $p_{0}$ is known, the tangent set loses the $\{a(X)(D-p_{0})\}$ directions and the $\lambda_{0}'h$ term in \eqref{eq: mean shift}; repeating the argument of Lemma~\ref{lem: eff influence} yields influence functions $\psi=b(X)D(Y-\mu_{0})+C(X)(D-p_{0})$ with $E[p_{0}bX']=I$, $E[p_{0}b\mid p_{0}]=0$, and $C$ now unconstrained (the $(D-p_{0})$-directions are orthogonal to the reduced tangent set), so the $C$-component only adds variance and is zero at the optimum; the effective conditional variance weight is $E[D(Y-\mu_{0})^{2}\mid X]=p_{0}\sigma_{0}^{2}$, and Lemma~\ref{lem: eff GLS} with $\kappa^{2}$ replaced by $p_{0}\sigma_{0}^{2}$ gives the known-$p_{0}$ bound $V^{\circ}=(E[(p_{0}/\sigma_{0}^{2})v^{\circ}v^{\circ\prime}])^{-1}$ with the $(p_{0}/\sigma_{0}^{2})$-tilted residual $v^{\circ}$. For any direction $c\in\mathbb{R}^{d}$, writing each inverse-bound quadratic form as a minimum over tilts,
\[
c'V^{*-1}c=\min_{g}E\bigl[\tfrac{p_{0}^{2}}{\kappa^{2}}\{c'X-g(p_{0})\}^{2}\bigr]\le\min_{g}E\bigl[\tfrac{p_{0}}{\sigma_{0}^{2}}\{c'X-g(p_{0})\}^{2}\bigr]=c'V^{\circ-1}c,
\]
because $p_{0}^{2}/\kappa^{2}\le p_{0}/\sigma_{0}^{2}$ pointwise (equivalent to $\kappa^{2}\ge p_{0}\sigma_{0}^{2}$, immediate from \eqref{eq: composite variance}). Hence $V^{*}\succeq V^{\circ}$. If $\lambda_{0}'(p_{0}(X))=0$ a.s.\ then $\kappa^{2}=p_{0}\sigma_{0}^{2}$ and $V^{*}=V^{\circ}$. Conversely, evaluating the first minimand at the $(p_{0}/\sigma_{0}^{2})$-optimal tilt $g^{\circ}_{c}$, $c'V^{*-1}c\le E[(p_{0}^{2}/\kappa^{2})\{c'X-g^{\circ}_{c}(p_{0})\}^{2}]<E[(p_{0}/\sigma_{0}^{2})\{c'X-g^{\circ}_{c}(p_{0})\}^{2}]=c'V^{\circ-1}c$ whenever the weight gap $p_{0}/\sigma_{0}^{2}-p_{0}^{2}/\kappa^{2}$, which is positive exactly on $\{\lambda_{0}'(p_{0})\neq0\}$ (under overlap), meets a nonvanishing residual $c'X-g^{\circ}_{c}(p_{0})$ there; thus $V^{*}\succ V^{\circ}$ in the full positive-definite order whenever $\mathrm{Var}(X\mid p_{0})$ has full rank on a positive-measure subset of $\{\lambda_{0}'(p_{0})\neq0\}$, and in general $V^{*}=V^{\circ}$ requires the tilted residuals to vanish a.s.\ wherever $\lambda_{0}'\neq0$.
\end{proof}

\subsection{The efficient weighted two-step estimator}\label{subsec: eff estimator proof}

Recall the estimator of Proposition~\ref{prop: efficient estimator}: the sample is split into two folds $I_{1},I_{2}$ of sizes $n_{1},n_{2}$ ($n_{k}/n\to1/2$); the first stage $\hat{p}$ is computed on the full sample as in Assumption~\ref{assu:np-firststage}; the weight function is estimated on each fold and used on the other, $\hat{w}^{(k)}:=T_{[\underline{w},\bar{w}]}\bigl(1/\{\hat{\sigma}^{2}_{(k)}(\cdot)+\hat{p}^{2}(1-\hat{p})\hat{\lambda}'_{(k)}(\hat{p})^{2}\}\bigr)$ where $T_{[\underline{w},\bar{w}]}$ truncates to a fixed interval $0<\underline{w}\le\bar{w}<\infty$ containing the range of $w_{0}$ (Assumption~\ref{assu:eff}(ii)); on fold $k$ the weighted second stage regresses $Y_i$ on $X_i$ and $b_K(\hat{p}_i)$ by weighted least squares over $\{i\in I_{k}:D_i=1\}$ with weights $\hat{w}_i=\hat{w}^{(-k)}(X_i)$, producing $\hat{\beta}^{(k)}$; and $\hat{\beta}^{*}:=\sum_{k}(n_{k}/n)\hat{\beta}^{(k)}$.

For a fixed weight function $w$ with $\underline{w}\le w\le\bar{w}$ define the $w$-weighted population objects: $m^{w}_{X,p}(t):=E[w X\mid p(X)=t,D=1]/E[w\mid p(X)=t,D=1]$, $v^{w}:=X-m^{w}_{X,p_{0}}(p_{0})$, $A_{w}:=E[Dwv^{w}v^{w\prime}]$, $\phi^{w}:=Dwv^{w}\lambda_{0}'(p_{0})$, and $\alpha_{w}:=E[\phi^{w}\mid X]=p_{0}w\lambda_{0}'(p_{0})v^{w}$. Note $A_{w}=E[Dwv^{w}X']$ by the weighted orthogonality $E[Dwv^{w}\mid p_{0}]=0$, and that at $w=w_{0}$, by \eqref{eq: weighted projection identity}, $m^{w_{0}}_{X,p_{0}}=\tilde{E}[X\mid p_{0}]$, $v^{w_{0}}=\tilde{v}$, and $A_{w_{0}}=E[p_{0}w_{0}\tilde{v}\tilde{v}']=V^{*-1}$.

\begin{lemma}[Fixed-weight expansion]\label{lem: eff fixed weight}
Fix a weight function $w$ with $\underline{w}\le w\le\bar{w}$, measurable in $X$, such that (i) Assumptions~\ref{assu:np-firststage}(iv), \ref{assu:diff}(i)--(ii), and~\ref{assu:secondstage} hold with the $w$-weighted conditional means and spline projections in place of the unweighted ones, and (ii) the sieve-approximation conditions of Assumption~\ref{assu:np-firststage}(ii)--(iii) hold with the weighted representer $\alpha_{w}=p_{0}w\lambda_{0}'(p_{0})v^{w}$ in place of $\alpha_{0}$, i.e.\ $\Delta_{\alpha_{w}}(L)\to0$ and $\sqrt{n}\,\Delta_{p}(L)\,\Delta_{\alpha_{w}}(L)\to0$ for the $L^{2}$-projection error $\Delta_{\alpha_{w}}$ of $\alpha_{w}$ onto $\mathrm{span}(\psi_{L})$. (Assumption~\ref{assu:eff}(iii) supplies (i)--(ii) at $w=w_{0}$.) Then, under the remaining conditions of Proposition~\ref{prop: np primitives}, the $w$-weighted second-stage estimator $\hat{\beta}_{w}$ computed on a subsample $I\subset\{1,\dots,n\}$ with $|I|/n\to\tau\in(0,1]$ (with full-sample $\hat{p}$) satisfies
\begin{equation}\label{eq: fixed weight expansion}
\sqrt{|I|}\bigl(\hat{\beta}_{w}-\beta_{0}\bigr)=A_{w}^{-1}\Bigl[\frac{1}{\sqrt{|I|}}\sum_{i\in I}D_iw_iv^{w}_i\varepsilon_i-\sqrt{\tfrac{|I|}{n}}\,\frac{1}{\sqrt{n}}\sum_{j=1}^{n}\alpha_{w}(X_j)(D_j-p_{0j})\Bigr]+o_{p}(1);
\end{equation}
the noise channel runs over the subsample while the first-stage score channel runs over the full sample, because $\hat{p}$ solves the full-sample normal equations.
\end{lemma}

\begin{proof}
The proof re-runs Lemmas~\ref{lem: np first stage rates}--\ref{lem: np riesz} and the assembly of Appendix~\ref{app: np proof} with the following substitutions, none of which alters any rate argument because $w$ is bounded above and below: the second-stage projection $\widehat{\Pi}_{K}$ becomes the $w$-weighted spline projection $\widehat{\Pi}^{w}_{K}$ on $\{b_K(\hat{p}_j)\}_{j\in I,D_j=1}$; the Frisch--Waugh--Lovell identity holds verbatim with weighted residuals $\hat{v}^{w}_i$ and $\hat{A}_{w}:=|I|^{-1}\sum_{i\in I}D_iw_i\hat{v}^{w}_i\hat{v}^{w\prime}_i$; in Lemma~\ref{lem: np second stage} \S1, the population target of $\widehat{\Pi}^{w}_{K}X$ is $m^{w}_{X,\hat{p}}$, the decomposition (I)--(III) applies with $m^{w}$ in place of $m$, term (I) is controlled by the $w$-weighted conditional-mean stability (hypothesis (i); the co-area primitive stated after Assumption~\ref{assu:np-firststage} delivers it because $m^{w}$ is a ratio of two stable conditional means with denominator $E[w\mid\cdot,D=1]\ge\underline{w}$), term (II) by the $w$-weighted H\"older smoothness of $m^{w}_{X,p}$ (the analogue of Assumption~\ref{assu:diff}(i) in hypothesis (i)), and term (III) by the uniform concentration of Lemma~\ref{lem: np uniform} with the bounded weight absorbed into the function class (envelope and entropy are unchanged, since $\underline{w}\le w\le\bar{w}$ is a fixed function of the conditioning data), together with the uniform generated-index Gram conditioning, whose $w$-weighted version follows from Lemma~\ref{lem: np uniform}(b) and the sandwich $\underline{w}\,\widehat{G}_{K}(\hat{p})\preceq\widehat{G}^{w}_{K}(\hat{p})\preceq\bar{w}\,\widehat{G}_{K}(\hat{p})$; \S2 (oracle term) and \S3 (control term) hold verbatim with $\phi^{w}$ in place of $\phi$ and the bounded factor $w$ carried through every display; Lemma~\ref{lem: np quadratic} holds with $\sup_{i}\lVert\hat{v}^{w}_i\rVert=O_{p}(1)$ (same Lebesgue-constant argument, weighted Gram conditioning as above); and Lemma~\ref{lem: np riesz} applies to $\phi^{w}$: its displayed hypotheses hold since $E[\phi^{w}\mid X]=\alpha_{w}\in L^{2}$ and $\sup_{x}E[\lVert\phi^{w}\rVert^{2}\mid X=x]<\infty$ under Assumptions~\ref{assu:IID}(ii) and~\ref{assu:hfunc}(i) and $w\le\bar{w}$, and the sieve-approximation inputs it consumes, $\Delta_{\alpha_{w}}(L)\to0$ in its term $(\mathrm{I_a})$ and $\sqrt{n}\,\Delta_{p}\Delta_{\alpha_{w}}\to0$ in its double-robustness cancellation (II)--(III), are exactly hypothesis (ii). The only modification to Lemma~\ref{lem: np riesz} is that the drift sum runs over $I$ while $\hat{p}$ solves the full-sample normal equations: writing $\hat{\Gamma}^{I}_{w}:=|I|^{-1}\sum_{i\in I}\phi^{w}_i\psi_{L}(X_i)'$ (with $E[\hat{\Gamma}^{I}_{w}]=\Gamma_{w}$ and the same variance bounds as in term $\mathrm{(I_b)}$ there), the score channel becomes $\sqrt{|I|}\,\hat{\Gamma}^{I}_{w}\hat{G}^{-1}S_{n}$ with the \emph{full-sample} $S_{n}$, so the leading term is $\sqrt{|I|}\,\Gamma_{w}G^{-1}S_{n}=\sqrt{|I|/n}\cdot n^{-1/2}\sum_{j=1}^{n}\alpha_{w,L}(X_j)(D_j-p_{0j})$, and the bias channels (II)--(III) are unchanged. Replacing $\alpha_{w,L}$ by $\alpha_{w}$ costs $o_{p}(1)$ exactly as in term $\mathrm{(I_a)}$. Collecting terms as in Section~\ref{subsec: np assembly} yields \eqref{eq: fixed weight expansion}.
\end{proof}

\begin{lemma}[Weight replacement]\label{lem: eff weight replacement}
Let $\hat{w}:=\hat{w}^{(-k)}$ be the cross-fitted weight of Proposition~\ref{prop: efficient estimator}: $\hat{w}$ is measurable with respect to $\mathcal{G}_{k}:=\sigma\bigl(\{(X_j,D_j)\}_{j=1}^{n},\{Y_j\}_{j\in I_{-k}}\bigr)$, $\underline{w}\le\hat{w}\le\bar{w}$, and $\lVert\hat{w}-w_{0}\rVert_{\infty}=o_{p}(n^{-1/4})$ (Assumption~\ref{assu:eff}(iv)). Then $\sqrt{n_{k}}\bigl(\hat{\beta}^{(k)}-\beta_{0}\bigr)$ admits the expansion \eqref{eq: fixed weight expansion} at the \emph{fixed} weight $w_{0}$ with $I=I_{k}$.
\end{lemma}

\begin{proof}
Write $\widehat{\Pi}^{\hat{w}}_{K}$ and $\widehat{\Pi}^{w_{0}}_{K}$ for the $\hat{w}$- and $w_{0}$-weighted spline projections on $\{b_K(\hat{p}_j)\}_{j\in I_{k},D_j=1}$, and $\hat{v}^{\hat{w}}_i,\hat{v}^{w_{0}}_i$ for the corresponding residuals of $X_i$. Because neither projection involves the outcomes, both residuals are $\mathcal{G}_{k}$-measurable. Three preliminary perturbation bounds. First, by the weighted normal equations and $\underline{w}\le\hat{w},w_{0}\le\bar{w}$, the spline-coefficient difference obeys $\lVert\hat{c}_{\hat{w}}-\hat{c}_{w_{0}}\rVert\le C\underline{w}^{-2}\lVert\hat{w}-w_{0}\rVert_{\infty}\,O_{p}(1)$ on the event that the $w_{0}$-weighted generated-index Gram is well conditioned (Assumption~\ref{assu:secondstage}(v) and the sandwich $\underline{w}\widehat{G}_{K}(\hat{p})\preceq\widehat{G}^{w}_{K}(\hat{p})\preceq\bar{w}\widehat{G}_{K}(\hat{p})$), whence
\begin{align}\label{eq: projection perturbation}
\bigl(n_{k}^{-1}\textstyle\sum_{i\in I_{k}}D_i\lVert\hat{v}^{\hat{w}}_i-\hat{v}^{w_{0}}_i\rVert^{2}\bigr)^{1/2}\le C\lVert\hat{w}-w_{0}\rVert_{\infty}&=o_{p}(n^{-1/4}),\\
\sup_{i\in I_{k}}\lVert\hat{v}^{\hat{w}}_i-\hat{v}^{w_{0}}_i\rVert \le\xi(K)\,O_{p}\bigl(\lVert\hat{w}-w_{0}\rVert_{\infty}\bigr)&=o_{p}(1),
\end{align}
the last step using $\xi(K)=O(\sqrt{K})$, $K=o(\sqrt{n})$, so $\xi(K)n^{-1/4}\to0$; in particular $\sup_{i}\lVert\hat{v}^{\hat{w}}_i\rVert=O_{p}(1)$. Second, $\lVert\hat{w}\hat{v}^{\hat{w}}-w_{0}\hat{v}^{w_{0}}\rVert_{\mathrm{emp}}\le\bar{w}\lVert\hat{v}^{\hat{w}}-\hat{v}^{w_{0}}\rVert_{\mathrm{emp}}+\lVert\hat{w}-w_{0}\rVert_{\infty}O_{p}(1)=o_{p}(n^{-1/4})$. Third, $\hat{A}_{\hat{w}}-\hat{A}_{w_{0}}=o_{p}(1)$ by the same bounds.

Now compare the two weighted-least-squares moments channel by channel, exactly as in the decomposition of Lemma~\ref{lem: np second stage}.

\emph{(a) Noise channel.} $n_{k}^{-1/2}\sum_{i\in I_{k}}D_i\{\hat{w}_i\hat{v}^{\hat{w}}_i-w_{0,i}\hat{v}^{w_{0}}_i\}\varepsilon_i$ is, conditionally on $\mathcal{G}_{k}$, a sum of independent mean-zero terms ($\varepsilon_i\perp\mathcal{G}_{k}$ given $(X_i,D_i)$ for $i\in I_{k}$, since $\hat{w}$ uses only fold $I_{-k}$ outcomes) with conditional variance bounded by $\bar{\sigma}^{2}\,n_{k}^{-1}\sum_{i}D_i\lVert\hat{w}_i\hat{v}^{\hat{w}}_i-w_{0,i}\hat{v}^{w_{0}}_i\rVert^{2}=o_{p}(1)$; conditional Chebyshev gives $o_{p}(1)$.

\emph{(b) Linear drift channel, pre-Riesz replacement.} The first-stage drift enters as $n_{k}^{-1/2}\sum_{i\in I_{k}}D_i\{\hat{w}_i\hat{v}^{\hat{w}}_i-w_{0,i}\hat{v}^{w_{0}}_i\}\lambda_{0}'(p_{0i})\delta_i$ with $\delta_i=\hat{p}_i-p_{0i}$; by Cauchy--Schwarz, bounded $\lambda_{0}'$, \eqref{eq: projection perturbation}, and the empirical first-stage rate $n^{-1}\sum_i\delta_i^{2}=o_{p}(n^{-1/2})$ (Lemma~\ref{lem: np first stage rates}(b)), it is bounded by $\sqrt{n_{k}}\cdot o_{p}(n^{-1/4})\cdot o_{p}(n^{-1/4})=o_{p}(1)$. The weight replacement therefore happens \emph{before} any Riesz conversion, and the drift is subsequently converted at the fixed function $w_{0}$ by Lemma~\ref{lem: eff fixed weight}; no expansion at the random weight $\hat{w}$ is ever required, which is what removes any own-observation bias from $\hat{w}$.

\emph{(c) Spline-bias and empirical-projection channels.} The $\hat{w}$-weighted spline-bias term is bounded by $\sqrt{n}\,\bar{w}\,O_{p}(1)\,O(K^{-m})=o_{p}(1)$ under Assumption~\ref{assu:secondstage}(iii), and the $[\Pi_{K}-\widehat{\Pi}^{\hat{w}}_{K}]$ term by the projection contraction in the $\hat{w}$-weighted empirical inner product, $\lVert(I-\widehat{\Pi}^{\hat{w}}_{K})r\rVert_{\mathrm{emp},\hat{w}}\le\lVert r\rVert_{\mathrm{emp},\hat{w}}\le\bar{w}^{1/2}\lVert r\rVert_{\mathrm{emp}}$ applied to the $O(K^{-m})$-small spline residual, giving $o_{p}(1)$ as in Lemma~\ref{lem: np second stage}\,\S3; neither requires any structure on $\hat{w}$ beyond its bounds.

\emph{(d) Quadratic remainder.} The $\hat{w}$-weighted quadratic remainder is bounded, as in Lemma~\ref{lem: np quadratic}, by $\tfrac{\bar{L}_{2}}{2}\bar{w}\,(\sup_{i}\lVert\hat{v}^{\hat{w}}_i\rVert)\sqrt{n}\,n^{-1}\sum_i\delta_i^{2}=o_{p}(1)$, using $\sup_{i}\lVert\hat{v}^{\hat{w}}_i\rVert=O_{p}(1)$ from \eqref{eq: projection perturbation}.

Combining (a)--(d) with $\hat{A}_{\hat{w}}-\hat{A}_{w_{0}}=o_{p}(1)$ and Slutsky, $\sqrt{n_{k}}(\hat{\beta}^{(k)}-\beta_{0})$ has the same expansion as the $w_{0}$-weighted estimator on $I_{k}$, which is \eqref{eq: fixed weight expansion} at $w=w_{0}$ by Lemma~\ref{lem: eff fixed weight}. The same argument with $w_{0}$ replaced by any fixed $w^{\dag}\in[\underline{w},\bar{w}]$ satisfying the hypotheses of Lemma~\ref{lem: eff fixed weight} and $\lVert\hat{w}-w^{\dag}\rVert_{\infty}=o_{p}(n^{-1/4})$ yields the expansion at $w^{\dag}$, which underlies the misspecification-robustness claim in Remark~\ref{rem: efficient weights}.
\end{proof}

\begin{proof}[Proof of Proposition~\ref{prop: efficient estimator}]
By Lemma~\ref{lem: eff weight replacement} and then Lemma~\ref{lem: eff fixed weight} at $w_{0}$ (using $v^{w_{0}}=\tilde{v}$, $\alpha_{w_{0}}=p_{0}w_{0}\lambda_{0}'\tilde{v}$, $A_{w_{0}}=V^{*-1}$),
\[
\sqrt{n_{k}}\bigl(\hat{\beta}^{(k)}-\beta_{0}\bigr)=V^{*}\Bigl[\frac{1}{\sqrt{n_{k}}}\sum_{i\in I_{k}}D_iw_{0,i}\tilde{v}_i\varepsilon_i-\sqrt{\tfrac{n_{k}}{n}}\,\frac{1}{\sqrt{n}}\sum_{j=1}^{n}\alpha_{w_{0}}(X_j)(D_j-p_{0j})\Bigr]+o_{p}(1).
\]
Aggregating, $\sqrt{n}(\hat{\beta}^{*}-\beta_{0})=\sum_{k}\sqrt{n_{k}/n}\,\sqrt{n_{k}}(\hat{\beta}^{(k)}-\beta_{0})$: the noise channels stack to $n^{-1/2}\sum_{i=1}^{n}D_iw_{0,i}\tilde{v}_i\varepsilon_i$, and the score channels carry total weight $\sum_{k}n_{k}/n=1$, giving
\[
\sqrt{n}\bigl(\hat{\beta}^{*}-\beta_{0}\bigr)=V^{*}\frac{1}{\sqrt{n}}\sum_{i=1}^{n}\bigl\{D_iw_{0,i}\tilde{v}_i\varepsilon_i-\alpha_{w_{0}}(X_i)(D_i-p_{0i})\bigr\}+o_{p}(1)=V^{*}\frac{1}{\sqrt{n}}\sum_{i=1}^{n}w_{0}(X_i)\tilde{v}_i\,\tilde{u}_i+o_{p}(1),
\]
the last equality by $Dw_{0}\tilde{v}\varepsilon-p_{0}w_{0}\lambda_{0}'\tilde{v}(D-p_{0})=w_{0}\tilde{v}\{D\varepsilon-p_{0}\lambda_{0}'(D-p_{0})\}=w_{0}\tilde{v}\tilde{u}$. The summands are i.i.d., mean zero ($E[\tilde{u}\mid X]=0$), with variance $E[w_{0}^{2}\tilde{v}\tilde{v}'\kappa^{2}]=E[(p_{0}^{2}/\kappa^{2})\tilde{v}\tilde{v}']=V^{*-1}$ by \eqref{eq: composite variance} and $w_{0}=p_{0}/\kappa^{2}$. By Lindeberg--L\'evy and Slutsky, $\sqrt{n}(\hat{\beta}^{*}-\beta_{0})\to_{d}N(0,V^{*}V^{*-1}V^{*})=N(0,V^{*})$, and the influence function is $V^{*}w_{0}\tilde{v}\tilde{u}=\psi^{*}$, the efficient influence function of Proposition~\ref{prop: efficiency bound}. Consistency of the plug-in variance estimator $\hat{V}^{*}$ of Remark~\ref{rem: efficient weights} is established separately in Lemma~\ref{lem: eff variance consistency} below.
\end{proof}

\begin{lemma}[Consistency of $\hat{V}^{*}$]\label{lem: eff variance consistency}
Suppose, in addition to the conditions of Proposition~\ref{prop: efficient estimator}, that $\sup_{x\in\mathcal X}E[|\varepsilon_i|^4\mid X_i=x,D_i=1]<\infty$, the second-stage spline dimension satisfies the derivative rate $K^{4}\log K/n\rightarrow0$, and $m>2$ in Assumption~\ref{assu:hfunc}; the joint window $n^{1/(2m)}\ll K\ll n^{1/4}$ implied by this rate and the maintained undersmoothing $\sqrt{n}K^{-m}\rightarrow0$ is nonempty precisely when $m>2$, so the composite variance estimator requires a smoother $\lambda_{0}$ than the point estimator, because differentiating the fitted spline tightens the admissible-$K$ cap from $\sqrt{n}$ to $n^{1/4}$. Then the composite plug-in $\hat{V}^{*}$ of Remark~\ref{rem: efficient weights} satisfies $\hat{V}^{*}\rightarrow_{p}V^{*}$.
\end{lemma}

\begin{proof}
Write $\hat{\omega}^{*}_{i}=\hat{w}_i\hat{v}^{\hat{w}}_i\{D_i\hat{\varepsilon}_i-\hat{p}_i\hat{\lambda}'(\hat{p}_i)(D_i-\hat{p}_i)\}$ for the estimated composite summand and $\omega^{*}_{i}=w_{0,i}\tilde{v}_i\tilde{u}_i$ for its population counterpart. Since $\hat{V}^{*}=\hat{A}_{\hat{w}}^{-1}\bigl(n^{-1}\sum_i\hat{\omega}^{*}_{i}\hat{\omega}^{*\prime}_{i}\bigr)\hat{A}_{\hat{w}}^{-1}$ and $\hat{A}_{\hat{w}}\rightarrow_{p}A_{w_{0}}=V^{*-1}$ (Lemma~\ref{lem: eff weight replacement} with the fixed-weight machinery), it suffices to show $n^{-1}\sum_i\lVert\hat{\omega}^{*}_{i}-\omega^{*}_{i}\rVert^{2}=o_{p}(1)$; the law of large numbers for $n^{-1}\sum_i\omega^{*}_{i}\omega^{*\prime}_{i}\rightarrow_{p}E[\omega^{*}\omega^{*\prime}]=V^{*-1}$ (finite under the fourth-moment condition of the lemma and Assumption~\ref{assu:eff}(i)--(ii)) and the Cauchy--Schwarz inequality then complete the argument. The squared gap decomposes into five nuisance errors: (i) $\hat{w}-w_{0}$, $o_{p}(n^{-1/4})$ in sup norm by Assumption~\ref{assu:eff}(iv); (ii) $\hat{v}^{\hat{w}}_i-\tilde{v}_i$, $o_{p}(1)$ in empirical $L^{2}$ norm by \eqref{eq: projection perturbation} and the fixed-weight analogue of Lemma~\ref{lem: np second stage} \S1, routed through Lemma~\ref{lem: np uniform}; (iii) $\hat{\varepsilon}_i-\varepsilon_i=X_i'(\beta_{0}-\hat{\beta})+\lambda_{0}(p_{0i})-\hat{\lambda}(\hat{p}_i)$, $o_{p}(1)$ in empirical $L^{2}$ norm by consistency of $\hat{\beta}$, the spline rates of Assumption~\ref{assu:secondstage}, and Lemma~\ref{lem: np first stage rates}; (iv) $\hat{p}_i-p_{0i}$, $o_{p}(1)$ in sup norm; and (v) the estimated derivative. Only (v) requires an argument beyond those already imposed, because $\hat{\lambda}$ is fitted on the same-sample generated index $\hat{p}$. Two points structure the argument. First, only \emph{empirical} $L^{2}$ consistency is needed: the five errors enter the average $n^{-1}\sum_i\lVert\hat{\omega}^{*}_{i}-\omega^{*}_{i}\rVert^{2}$ multiplied by uniformly bounded cofactors, so by the Cauchy--Schwarz inequality it suffices that $\lVert\hat{\lambda}'(\hat{p})-\lambda_{0}'(p_{0})\rVert_{n,2,D}:=\bigl(n^{-1}\sum_iD_i[\hat{\lambda}'(\hat{p}_i)-\lambda_{0}'(p_{0i})]^{2}\bigr)^{1/2}=o_{p}(1)$; no sup-norm result is required. Second, a pointwise conditioning argument at $\hat{p}$ is unavailable for the same reason as in Lemma~\ref{lem: np second stage}: conditioning on the first stage makes the outcome shock mean zero but leaves the design discrepancy $q^{w_{0}}_{\hat{p}}(X)=\mu_{0}(X)-X'\beta_{w_{0}}(\hat{p})-\lambda_{w_{0},\hat{p}}(\hat{p}(X))$, which satisfies $E[Dw_{0}q^{w_{0}}_{p}\mid p(X)]=0$ only at each fixed $p$, not conditionally on $X$; its empirical projection onto the joint design is controlled uniformly over $\mathcal{P}_{n}$ by Lemma~\ref{lem: np uniform}(c) instead. Throughout, the estimator is exactly as defined: the full-sample series first stage, the joint least-squares second stage of $Y$ on $(X,b_K(\hat{p}))$, and the first-order Riesz-corrected variance; the Frisch--Waugh--Lovell representation remains a proof device only.

\emph{Estimation step (joint-series level rate and inverse inequality).} Lemma~\ref{lem: np uniform}(c), applied within fold $k$ at the fixed bounded weight $\omega=w_{0}$ (its weighted conditional means and partially linear component lie in a fixed $\Lambda^{m}$ ball by Assumption~\ref{assu:eff}(iii)), gives $\lVert\hat{c}_{w_{0}}-c_{w_{0},\hat{p}}\rVert=O_{p}(\rho_{n}+K^{-m})$ with $\rho_{n}=\sqrt{(K+L)\log n/n}$. The estimated weight is transferred by a direct normal-equation perturbation of the \emph{joint} outcome regression, not by \eqref{eq: projection perturbation} (which covers the $X$-on-spline projection only). Write $\hat{\theta}_{\hat{w}}-\hat{\theta}_{w_{0}}=\widehat{Q}_{\hat{w}}(\hat{p})^{-1}s$ with $s:=n_{k}^{-1}\sum_{i\in I_{k}}D_i(\hat{w}_i-w_{0,i})Z_i(\hat{p})\,e_{w_{0},i}$ and $e_{w_{0},i}:=Y_i-Z_i(\hat{p})'\hat{\theta}_{w_{0}}$, and dualize the score in the $\widehat{Q}_{\hat{w}}^{-1/2}$ metric: $\lVert\widehat{Q}_{\hat{w}}^{-1}s\rVert\leq\lVert\widehat{Q}_{\hat{w}}^{-1/2}\rVert_{op}\sup_{\lVert a\rVert=1}\lvert n_{k}^{-1}\sum_iD_i(\hat{w}_i-w_{0,i})(a'\widehat{Q}_{\hat{w}}^{-1/2}Z_i)e_{w_{0},i}\rvert\leq O_{p}(1)\,\lVert\hat{w}-w_{0}\rVert_{\infty}\sup_{\lVert a\rVert=1}\bigl(n_{k}^{-1}\sum_iD_i(a'\widehat{Q}_{\hat{w}}^{-1/2}Z_i)^{2}\bigr)^{1/2}\bigl(n_{k}^{-1}\sum_iD_ie_{w_{0},i}^{2}\bigr)^{1/2}$ by the Cauchy--Schwarz inequality. The middle factor is dimension free: $n_{k}^{-1}\sum_iD_i(a'\widehat{Q}_{\hat{w}}^{-1/2}Z_i)^{2}=a'\widehat{Q}_{\hat{w}}^{-1/2}\widehat{Q}_{1}\widehat{Q}_{\hat{w}}^{-1/2}a\leq\underline{w}^{-1}$, since the weight sandwich gives $\widehat{Q}_{1}\preceq\underline{w}^{-1}\widehat{Q}_{\hat{w}}$. Hence $\lVert\hat{\theta}_{\hat{w}}-\hat{\theta}_{w_{0}}\rVert=O_{p}(1)\cdot o_{p}(n^{-1/4})\cdot O_{p}(1)=o_{p}(n^{-1/4})$ by Assumption~\ref{assu:eff}(iv) and $n_{k}^{-1}\sum_iD_ie_{w_{0},i}^{2}=O_{p}(1)$, so the level rate for the reported coefficients is $\lVert\hat{c}-c_{w_{0},\hat{p}}\rVert=O_{p}(\rho_{n}+K^{-m})+o_{p}(n^{-1/4})$; after the $O_{p}(K)$ conversion below, the transfer term contributes $o_{p}(Kn^{-1/4})=o_{p}(1)$, since $K=o(n^{1/4})$. The conversion to the derivative uses the spline inverse inequality through the \emph{maximum eigenvalue} of the derivative Gram, not its trace (the trace is of order $K^{3}$ for the normalized basis, and a per-summand Cauchy--Schwarz bound would deliver only $O_{p}(K^{3/2})$). For any spline $s=b_K'u$, the $L^{2}$ inverse inequality $\lVert s'\rVert_{L^{2}(dt)}\leq CK\lVert s\rVert_{L^{2}(dt)}$ and the two-sided class density bounds give $u'H_{p}u=\int(s')^{2}q_{p}\,dt\leq\bar{f}\,C^{2}K^{2}\int s^{2}\,dt\leq(\bar{f}/\underline{f})C^{2}K^{2}\,u'G_{K}(p)u$, that is, $H_{p}:=E[D_ib_K'(p_i)b_K'(p_i)']\preceq CK^{2}G_{K}(p)\preceq CK^{2}\bar{c}_{G}I$ uniformly over $\mathcal{P}_{n}$ by Assumption~\ref{assu:p}(iii). The empirical version satisfies $\lambda_{\max}(\widehat{H}_{\hat{p}})=O_{p}(K^{2})$ on the membership event: the deviation $\sup_{p}\lVert\widehat{H}_{p}-H_{p}\rVert_{op}$ is controlled by the matrix Bernstein inequality \citep[Theorem~6.1 of][]{tropp2012user} over the $p$-net of Lemma~\ref{lem: np uniform}, with envelope $O(K^{3})$ and variance proxy $O(K^{5})$, and is $O_{p}(K^{2})$ because $K(K+L)\log n=o(n/\log n)$ under Assumption~\ref{assu:secondstage}(iv); the off-net extension uses the spline derivative bounds $\sup_t\lVert b_K'(t)\rVert=O(K^{3/2})$ and $\sup_t\lVert b_K''(t)\rVert=O(K^{5/2})$, so moving $p$ within $\zeta(L)n^{-4}$ of a net centre perturbs $\widehat{H}_{p}$ in operator norm by at most $2\sup_t\lVert b_K'\rVert\sup_t\lVert b_K''\rVert\,\zeta(L)n^{-4}=O(K^{4}\zeta(L)n^{-4})$, negligible under polynomial growth. Hence $\lVert b_K'(\hat{p})'(\hat{c}-c_{w_{0},\hat{p}})\rVert_{n,2,D}^{2}=(\hat{c}-c_{w_{0},\hat{p}})'\widehat{H}_{\hat{p}}(\hat{c}-c_{w_{0},\hat{p}})\leq O_{p}(K^{2})\lVert\hat{c}-c_{w_{0},\hat{p}}\rVert^{2}$, so the conversion factor is $O_{p}(K)$; together with the $\hat{w}$-transfer term bounded above and the spline simultaneous-approximation error for the derivative, a further $O(K^{-(m-1)})$, this gives
\[
\lVert\hat{\lambda}'(\hat{p})-\lambda_{w_{0},\hat{p}}'(\hat{p})\rVert_{n,2,D}=O_{p}\bigl(K\rho_{n}+K^{-(m-1)}\bigr)+o_{p}\bigl(Kn^{-1/4}\bigr)=o_{p}(1),
\]
since $K^{2}\rho_{n}^{2}=\{K^{2}/\sqrt{n}\}\{(K+L)\log n/\sqrt{n}\}\rightarrow0$ by $K^{4}\log K/n\rightarrow0$ and Assumption~\ref{assu:secondstage}(iv), $Kn^{-1/4}\rightarrow0$ by $K=o(n^{1/4})$, and $m>1$; no additional $K$--$L$ condition is needed. The weight-transfer term is carried explicitly because it need not be dominated by the other two.

\emph{Target-drift step (empirical $L^{2}$).} It remains to show $\lVert\lambda_{w_{0},\hat{p}}'(\hat{p})-\lambda_{0}'(p_{0})\rVert_{n,2,D}=o_{p}(1)$. The argument-shift piece obeys $\lVert\lambda_{0}'(\hat{p})-\lambda_{0}'(p_{0})\rVert_{n,2,D}\leq\bar L_2\bigl(n^{-1}\sum_i\delta_i^{2}\bigr)^{1/2}=o_{p}(n^{-1/4})$ by Assumption~\ref{assu:hfunc}(i) and Lemma~\ref{lem: np first stage rates}(b). For the function-drift piece set $g_{n}:=\lambda_{w_{0},\hat{p}}-\lambda_{0}$, which on the membership event has uniformly bounded first and second derivatives ($m>2$, the class's H\"{o}lder ball, and boundedness of $\beta_{w_{0}}(\hat{p})$). Level smallness: $\lVert\beta_{w_{0}}(\hat{p})-\beta_{0}\rVert=o_{p}(1)$: $\beta_{w_{0}}(p_{0})=\beta_{0}$ exactly, because the conditional mean is exactly partially linear so any bounded-weight projection at the true index recovers $(\beta_{0},\lambda_{0})$, $\beta_{w_{0}}(p)-\beta_{0}=A_{w_{0}}(p)^{-1}E[D w_{0}v^{w_{0}}_{p}\{\lambda_{0}(p_{0})-E^{w_{0}}[\lambda_{0}(p_{0})\mid p,D{=}1]\}]$, is $O_{p}(\lVert\hat{p}-p_{0}\rVert_{L^{2}})$ at $p=\hat{p}$ by the weighted conditional-mean stability; and the weighted conditional-mean stability (Assumption~\ref{assu:eff}(iii)) together with the Lipschitz continuity of $\lambda_{0}$ give $E[D_ig_{n}(\hat{p}(X_i))^{2}]=o_{p}(1)$; since $E[D_ig_{n}(\hat{p}(X_i))^{2}]=P[D_i=1]\int g_{n}^{2}f_{\hat{p}\mid D=1}\,dt\geq P[D_i=1]\,\underline{f}\int g_{n}^{2}\,dt$ by the class's conditional density lower bound on its interval support, this converts to $\lVert g_{n}\rVert_{L^{2}(dt)}=o_{p}(1)$ on that interval. The interval interpolation inequality $\lVert g_{n}'\rVert_{L^{2}(dt)}\leq C\bigl(\lVert g_{n}\rVert_{L^{2}(dt)}^{1/2}\lVert g_{n}''\rVert_{L^{2}(dt)}^{1/2}+\lVert g_{n}\rVert_{L^{2}(dt)}\bigr)=o_{p}(1)$ upgrades the level to the derivative in $L^{2}(dt)$, and the empirical-to-population comparison for the scalar class $\{d\,(g'(p(x)))^{2}:p\in\mathcal{P}_{n},\,g\in C^{2}\text{-ball}\}$ (bounded envelope; the net-and-bracketing argument of Lemma~\ref{lem: np uniform}) together with the density upper bound $\bar{f}$ returns $\lVert g_{n}'(\hat{p})\rVert_{n,2,D}^{2}\leq \bar{f}\,C\lVert g_{n}'\rVert_{L^{2}(dt)}^{2}+o_{p}(1)=o_{p}(1)$. Combining the two steps, $\lVert\hat{\lambda}'(\hat{p})-\lambda_{0}'(p_{0})\rVert_{n,2,D}=o_{p}(1)$, which is exactly what the composite-summand bound consumes. Every error enters $\hat{\omega}^{*}_{i}$ multiplied by factors that are uniformly $O_{p}(1)$: $\hat{w}$ is bounded by construction, $\sup_i\lVert\hat{v}^{\hat{w}}_i\rVert=O_{p}(1)$ by the weighted stability, $\hat{p}_i\in[0,1]$, and $\lambda_{0}'$ is bounded; the product structure and the fourth-moment condition of the lemma therefore give $n^{-1}\sum_i\lVert\hat{\omega}^{*}_{i}-\omega^{*}_{i}\rVert^{2}=o_{p}(1)$, completing the proof.
\end{proof}

\section{Proof of the fixed-dimensional implementation (Proposition~\refPropAsym)}\label{app: fixed K proof}

Proposition~\ref{prop: asymptotics} is Theorem~4.1 of \cite{newey2009two} specialized to the no-exclusion selection estimator. The proof therefore consists of mapping notation and verifying his Assumptions 2.1 and 4.1--4.5 from ours; no separate asymptotic argument is required. Throughout, $d_{\gamma}$ denotes the fixed first-stage dimension of Assumption~\ref{assu:convergence}(i) and $K=K_{n}\rightarrow\infty$ the second-stage spline dimension, as in Assumption~\ref{assu:hfunc}.

\emph{Notation map.} His $(y,x,w,d)$ are our $(Y_i,X_i,X_i,D_i)$: the first-step regressors coincide with the outcome regressors, which is precisely the no-exclusion design. His index $v(w,\alpha)$ is our $\phi(X_i)'\gamma$ with $\alpha=\gamma\in\mathbb{R}^{d_{\gamma}}$ fixed and finite dimensional; his strictly monotone transformation $\tau(v,\eta)$ is our known link $F$, with no nuisance parameter (formally $\hat{\eta}=\eta_{0}$, trivially $\sqrt{n}$-consistent); his second-step approximating functions $p^{K}(\tau)$ are our spline basis evaluated at $\hat{p}_i=F(\phi(X_i)'\hat{\gamma}_{n})$; his unknown correction function $h_{0}(v)$ is $\lambda_{0}(F(v))$, so that $\partial h_{0}(v)/\partial v=\lambda_{0}'(p_{0})f(\phi'\gamma_{0})$ by the chain rule; his selected-sample partialling-out residual $u_i$ is our $D_iv_i$, writing $v_i$ as a column vector in this section (so that $u_i=\tilde{X}_i'$ under the main text's row convention); and his $(M,\Omega,H,V_{\hat{\alpha}})$ are our $(A,\Omega,G,V_{\gamma})$, since
\[
H=E\left[u_i\,\frac{\partial h_{0}(v_i)}{\partial v}\,\frac{\partial v(w_i,\alpha_{0})}{\partial\alpha'}\right]=E\left[D_iv_i\,\lambda_{0}'(p_{0i})\,f(\phi(X_i)'\gamma_{0})\,\phi(X_i)'\right]=G.
\]

\emph{Verification of the conditions.}

\emph{His Assumption 2.1.} Nonsingularity of $M=E[u_iu_i']=A$ is Assumption~\ref{assu:Dv*}(i). His surrounding discussion motivates this condition through an excluded variable, but the assumption itself requires only that no linear combination of $X_i$ coincide with a function of the index on the selected sample; in our design that is delivered by nonlinearity of $p_{0}$ in $X_i$, which is the identification content of Section~\ref{sec: model}.

\emph{His Assumption 4.1.} Root-$n$ asymptotic linearity of the first step with a mean-zero influence function depending only on the first-stage data $(X_i,D_i)$, a finite nonsingular variance, and a consistent variance estimator: Assumption~\ref{assu:convergence}(ii)--(iii), with $\psi_i=J^{-1}s_i$ and $V_{\gamma}=J^{-1}\Sigma_{s}J^{-1}$. That $\psi_i$ is a function of the first-stage data alone is what yields $E[u_i\varepsilon_i\psi_i']=0$ and hence the additive form of $V$ below.

\emph{His Assumption 4.2.} $E[D_i\lVert X_i\rVert^{2+\delta}]<\infty$ and bounded $\mathrm{Var}(X_i\mid v,D_i=1)$ follow from the compact support of $X_i$ (Assumption~\ref{assu:IID}(ii)); bounded $E[\varepsilon_i^{2}\mid v,D_i=1]$ is Assumption~\ref{assu:s2}.

\emph{His Assumption 4.3.} $h_{0}=\lambda_{0}\circ F$ is continuously differentiable of order $s\geq3$ by Assumption~\ref{assu:convergence}(iv); $E[X_i\mid v,D_i=1]$, which equals $E[X_i\mid p_{0}(X_i)=\cdot,D_i=1]$ composed with $F$, is continuously differentiable of order $t=2$ by Assumptions~\ref{assu:diff}(i) and \ref{assu:convergence}(i).

\emph{His Assumption 4.4.} $\sqrt{n}(\hat{\eta}-\eta_{0})=O_{p}(1)$ holds trivially. The distribution of $\tau(v(w,\alpha_{0}),\eta_{0})=p_{0}(X_i)$ is absolutely continuous with density bounded away from zero on a compact support: the support is compact because $p_{0}$ is continuous on the compact $\mathcal{X}$ (Assumption~\ref{assu:IID}(ii)), and by Bayes' rule the unconditional density satisfies $f_{p_{0}}(t)=f_{p_{0}\mid D=1}(t)\,P[D_i=1]/P[D_i=1\mid p_{0i}=t]\in[\underline{f}\,P[D_i=1],\,\bar{f}\,P[D_i=1]/\varepsilon]$, using the conditional density bounds of Assumption~\ref{assu:p}(i) and the overlap bound $P[D_i=1\mid p_{0}(X_i)=\cdot]>\varepsilon$ of Assumption~\ref{assu:diff}(iii). The first two derivatives of $v(w,\alpha)=\phi(w)'\alpha$ with respect to $\alpha$ are $\phi(w)$ and $0$, bounded on the compact support, and those of $\tau=F$ are bounded by Assumption~\ref{assu:convergence}(i).

\emph{His Assumption 4.5, spline case.} The cubic spline satisfies his degree requirement $3\geq t-1=1$; the smoothness order satisfies $s\geq3$; and the rate conditions $K^{4}/n\rightarrow0$ and $\sqrt{n}\,K^{-s-t+1}=\sqrt{n}\,K^{-s-1}\rightarrow0$ are imposed in Assumption~\ref{assu:convergence}(v). The two rates are compatible: $\sqrt{n}\,K^{-s-1}\rightarrow0$ requires $K\gg n^{1/(2(s+1))}$, and $1/(2(s+1))\leq1/8<1/4$.

\emph{Nonsingularity of his $\Omega$.} For any $c\neq0$, $c'\Omega c=E[\sigma_{0}^{2}(X_i)D_i(c'v_i)^{2}]$; since $\sigma_{0}^{2}$ is positive (Assumption~\ref{assu:s2}), $c'\Omega c=0$ would force $D_i(c'v_i)^{2}=0$ almost surely and hence $c'Ac=0$, contradicting Assumption~\ref{assu:Dv*}(i).

\emph{Conclusion.} All conditions of Theorem~4.1 of \cite{newey2009two} hold. His proof establishes, through his equations (A.6)--(A.10), the expansion
\[
\sqrt{n}\left(\hat{\beta}_{n}-\beta_{0}\right)=M^{-1}\Bigl[\frac{1}{\sqrt{n}}\sum_{i=1}^{n}u_i\varepsilon_i-H\,\sqrt{n}\left(\hat{\alpha}-\alpha_{0}\right)\Bigr]+o_{p}\left(1\right),
\]
and substituting the first-stage expansion of Assumption~\ref{assu:convergence}(ii) yields exactly the displayed representation of Proposition~\ref{prop: asymptotics}, $\sqrt{n}(\hat{\beta}_{n}-\beta_{0})=A^{-1}n^{-1/2}\sum_i\{\tilde{X}_i'\varepsilon_{i}-GJ^{-1}s_{i}\}+o_{p}(1)$. (His printed combined score, p.~S227, carries the opposite sign on the first-step influence term relative to his equation (A.8); the discrepancy is immaterial, as the limiting variance depends on it only through the quadratic form $HV_{\hat{\alpha}}H'$.) The theorem then delivers $\sqrt{n}(\hat{\beta}_{n}-\beta_{0})\rightarrow_{d}N(0,V)$ with $V=M^{-1}(\Omega+HV_{\hat{\alpha}}H')M^{-1}=A^{-1}(\Omega+GV_{\gamma}G')A^{-1}$, which is \eqref{eq: asymptotic distribution when p is estimated}, together with consistency of his variance estimator, his equation (3.7). Our \eqref{eq: corrected variance estimator} is the sample outer product of the combined influence summand; expanding the square reproduces his two terms, the White component and the correction $\hat{G}\hat{J}^{-1}\hat{\Sigma}_{s}\hat{J}^{-1}\hat{G}'$, plus sample cross terms that converge in probability to $E[u_i\varepsilon_i\psi_i']=0$ (established on his p.~S228), so the two estimators share the limit $V$ and consistency transfers. In particular, the uniform consistency of the estimated derivative inside $\hat{G}$, $\max_{i\leq n}\lvert\partial\hat{h}(\hat{v}_i)/\partial v-\partial h_{0}(v_i)/\partial v\rvert\rightarrow_{p}0$, is established within his proof (his equation (A.11)) under the conditions verified above and requires no separate argument. This completes the proof. \qed

\section{Detailed simulation results}\label{appendix: simulations}

\subsection{Design specifications}\label{subsec: dgp specs}

We first investigate the single-covariate case using the following DGP (referred to
as DGP0):
\begin{equation}
    Y = D \cdot \left(\beta_0 + X\beta_1 + 2 \cdot V \right),\quad D=\mathbbm{1}[\alpha_0 + \alpha_1 X + \alpha_2 X^2 + \alpha_3 X^3 + U \ge 0],
\end{equation}
\begin{equation}\label{eq: joint normality}
X\sim N(0,1), \quad \begin{bmatrix}
    V\\U
    \end{bmatrix}|X \sim  N\left(\begin{bmatrix}
        0\\0
    \end{bmatrix}, \begin{bmatrix}
        1 & 0.75\\
        0.75& 1
        \end{bmatrix} \right).
\end{equation}
In this instance, $\beta_0$ is not separately identified from $\lambda_0(\cdot)$. The identification of $\beta_1$ is contingent upon the parameter values $\alpha = (\alpha_0, \alpha_1, \alpha_2, \alpha_3)$. This is because $p_0(X)$ must not exhibit strict monotonicity. We employ our semiparametric estimator (denoted KL) for $\beta_1$. For alternative estimators, we consider the two-part model (referred to as TPM), which is the ordinary least squares (OLS) estimator conditional on $D=1$, assuming random selection, and Heckman's MLE (denoted HSM). Both TPM and HSM are misspecified. We also compare our estimator with the oracle estimator, which incorporates the true functional forms of $p_0(X)$ and the selection bias. Specifically, in the oracle estimation, we estimate $\alpha$ using probit regression of $D$ on $Z:=(1, X, X^2, X^3)$. Subsequently, we include $\frac{\phi(Z\hat{\alpha})}{\Phi(Z\hat{\alpha})}$ as a control function in the second step.

We set $(\beta_0, \beta_1) = (0.5, 1)$ and consider two designs for selection parameters:
$$\text{(a) } \alpha = (0.6, 1.50, -0.5, -0.05), \quad \text{(b) } \alpha = (0.4, 1.50, 0.2, 0.05).$$

We next generate Monte Carlo samples from the following DGP (referred to as DGP1
henceforth), where $X$ consists of two continuously distributed variables and the
unobservables are jointly normally distributed as in \eqref{eq: joint normality}:
\begin{align}
    Y &= D \cdot \left(\beta_0 + X_1\beta_1 + X_2\beta_2 + 2 \cdot V \right),\\
    D &=\mathbbm{1}[\alpha_0 + \alpha_1X_1 + \alpha_2X_1^2 + \alpha_3X_1^3 + \alpha_4X_1X_2 + \alpha_5X_2 + \alpha_6X_2^2 + U \ge 0].
\end{align}
$X_1$ and $X_2$ are drawn from the standard normal distribution and independent of each other. The parameter values are set as:
$$\alpha = (1.5, 0.5, -0.5, 0.2, 0.5, 1.0, -0.5), \quad \beta = (0.5, 0.5, 0.25).$$
The average selection probability across Monte Carlo samples is 66\%.

In the last design (DGP2 henceforth), we consider the scenario where $X$ consists of a continuously distributed variable, $X_1 \sim N(0,1)$, and a binary variable, $X_2 \sim \operatorname{Bernoulli}(0.5)$. The remaining elements of DGP2 are otherwise identical to DGP1 except for the selection process:
\begin{equation*}
    D =\mathbbm{1}[\alpha_0 + \alpha_1X_1 + \alpha_2X_1^2 + \alpha_3X_1^3 + \alpha_4 X_1X_2 + \alpha_5 X_2 + \alpha_6X_1^2X_2 + \alpha_7X_1^3X_2 + U \ge 0].
\end{equation*}
The parameter values are set as:
$$\alpha = (0.2, -0.2, -0.5, 0.3, 0.1, 0.5, -0.3, 0.2), \quad \beta = (0.5, 0.5, 0.25).$$
The average selection probability is 52\% under this DGP.

\subsection{Single-covariate design (DGP0)}\label{subsec: dgp0 results}

In both designs, the parameter values are chosen to ensure that the selection probability, $P(D=1)$, is approximately 60\%. Let $h(X) := \alpha_0 + \alpha_1 X + \alpha_2 X^2 + \alpha_3 X^3$ be the selection index. The shape of $h$ is shown in Figure~\ref{fig: simul-DGP0-selection-index}, and the performance of the estimators is reported in Table~\ref{tab: finite sample performance - DGP0} and Figure~\ref{fig: simul-DGP0}. Under design (a), $h(\cdot)$ is not monotone, so our estimator for $\beta_1$ is well centred around the true value. Its root-mean-squared error (RMSE) is close to that of the oracle estimator. Conversely, with design (b), $\beta_1$ remains unidentified, so our estimator suffers from a large RMSE, as expected. In both designs, the TPM exhibits substantial misspecification bias. Heckman’s MLE performs poorly in the non-monotone design due to misspecification but performs very well in the monotone design. This is because the selection index is close to linear in the effective support of $X$ and the error distribution is correctly specified in the monotone design. Coverage probabilities illustrate the practical content of Proposition~\ref{prop: asymptotics}. Selection makes the second-stage error heteroskedastic by construction, even though $(U,V)$ is homoskedastic, so the appropriate first-stage-blind benchmark is the robust sandwich standard error of the second-stage regression, that is, \eqref{eq: corrected variance estimator} with the correction term omitted. Confidence intervals based on these robust but uncorrected standard errors fall short of nominal coverage (0.921 in design (a) and 0.941 in design (b)), whereas the corrected standard errors from \eqref{eq: corrected variance estimator} restore approximately nominal coverage in both designs (0.944 and 0.950). The gap is direct finite-sample evidence that the estimation of $p_0$ contributes to the sampling variance of $\hat{\beta}$, as the theory predicts; the first-stage contribution is largest in design (a), where the correction raises coverage from 0.921 to 0.944.

\begin{table}[!htbp]
\centering\small
\setlength{\tabcolsep}{4pt}
\caption{Finite-sample performance of estimators (single covariate)}\label{tab: finite sample performance - DGP0}
\begin{tabular}{lcccccccc}
\toprule
 & \multicolumn{4}{c}{Non-monotone selection} & \multicolumn{4}{c}{Monotone selection} \\
\cmidrule(lr){2-5} \cmidrule(lr){6-9}
 & TPM & HSM & KL & Oracle & TPM & HSM & KL & Oracle \\
\midrule
 RMSE &  0.524 & 0.109 & 0.080 &  0.071 &  0.693 &  0.059 &  0.231 &  0.096 \\
 Bias & -0.522 & 0.092 & -0.020 & -0.004 & -0.692 & -0.015 & -0.004 & -0.003 \\
 Coverage (uncorrected) & 0.000 & 0.671 & 0.921 & 0.941 & 0.000 & 0.950 & 0.941 & 0.954 \\
 Coverage (corrected) & -- & -- & 0.944 & -- & -- & -- & 0.950 & -- \\
\bottomrule
\end{tabular}
\par\footnotesize\renewcommand{\baselineskip}{11pt}\justifying
\textbf{Note:} TPM, HSM, KL, and Oracle denote OLS under random selection, Heckman’s MLE, our semiparametric sieve estimator, and the oracle estimator, respectively. `Coverage' rows report the coverage probability of 95\% confidence intervals. For the KL estimator, `uncorrected' uses the robust standard errors of the second-stage regression, which treat $\hat{p}$ as known (the correction term in \eqref{eq: corrected variance estimator} omitted), and `corrected' uses the first-stage-corrected standard errors of Proposition \ref{prop: asymptotics} via \eqref{eq: corrected variance estimator}; TPM, HSM, and Oracle use their conventional standard errors.
\end{table}
\FloatBarrier

\begin{figure}[!htbp]
    \centering
    \caption{Selection index designs (DGP0)}\label{fig: simul-DGP0-selection-index}
    \includegraphics[width=\textwidth]{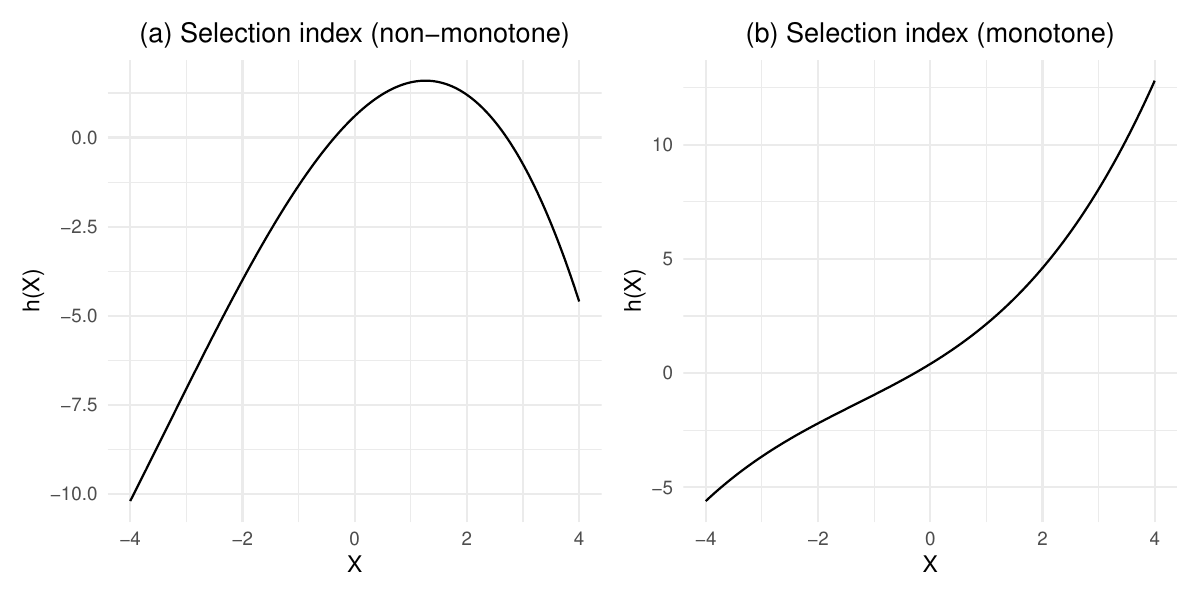}
    \caption{Finite-sample performance of estimators (DGP0)}\label{fig: simul-DGP0}
    \includegraphics[width=\textwidth]{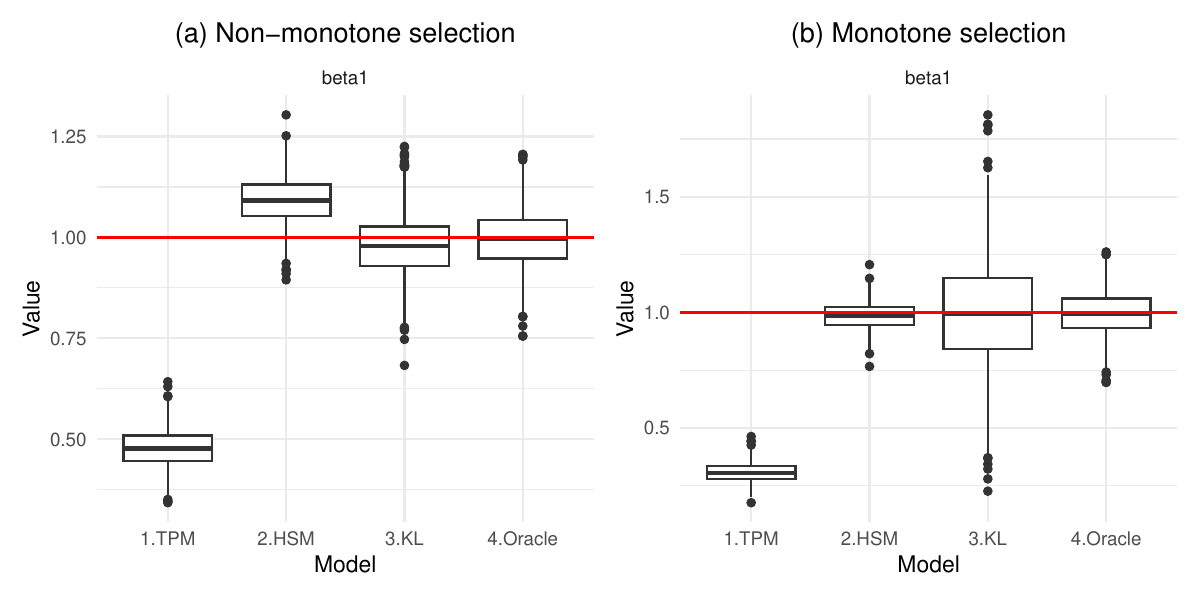}
    \par\footnotesize\renewcommand{\baselineskip}{11pt}\raggedright
\textbf{Note:} Boxplots mark the interquartile range (IQR) with the median (thick bar); whiskers extend to 1.5~×~IQR and dots denote outliers. Red lines indicate true parameter values. TPM, HSM, KL, and Oracle denote OLS under random selection, Heckman’s MLE, our semiparametric sieve estimator, and the oracle estimator, respectively.
\end{figure}
\FloatBarrier

\subsection{Two-covariate designs (DGP1 and DGP2)}\label{subsec: dgp12 results}

Figure~\ref{fig:simul-overall} displays the box plots of the estimates for DGP1 and
DGP2; the corresponding summary statistics are reported in Table~\ref{tab: finite sample performance - DGP1 and 2} of the paper.

\begin{figure}[!htbp]
  \centering
  \caption{Finite-sample performance (DGP1 and DGP2)}
  \label{fig:simul-overall}
  \begin{minipage}[t]{0.78\textwidth}
    \centering
    \textbf{(a) DGP1 (selection probability = 0.66)}\par\smallskip
    \includegraphics[width=\linewidth]{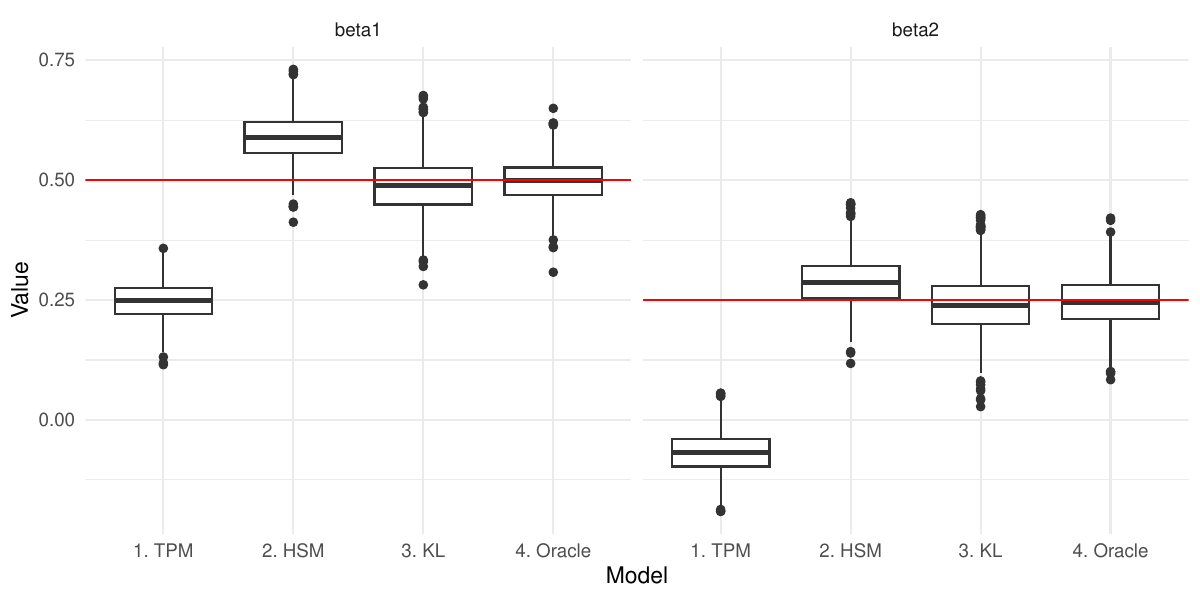}
  \end{minipage}

  \vspace{1em}

  \begin{minipage}[t]{0.78\textwidth}
    \centering
    \textbf{(b) DGP2 (selection probability = 0.52)}\par\smallskip
    \includegraphics[width=\linewidth]{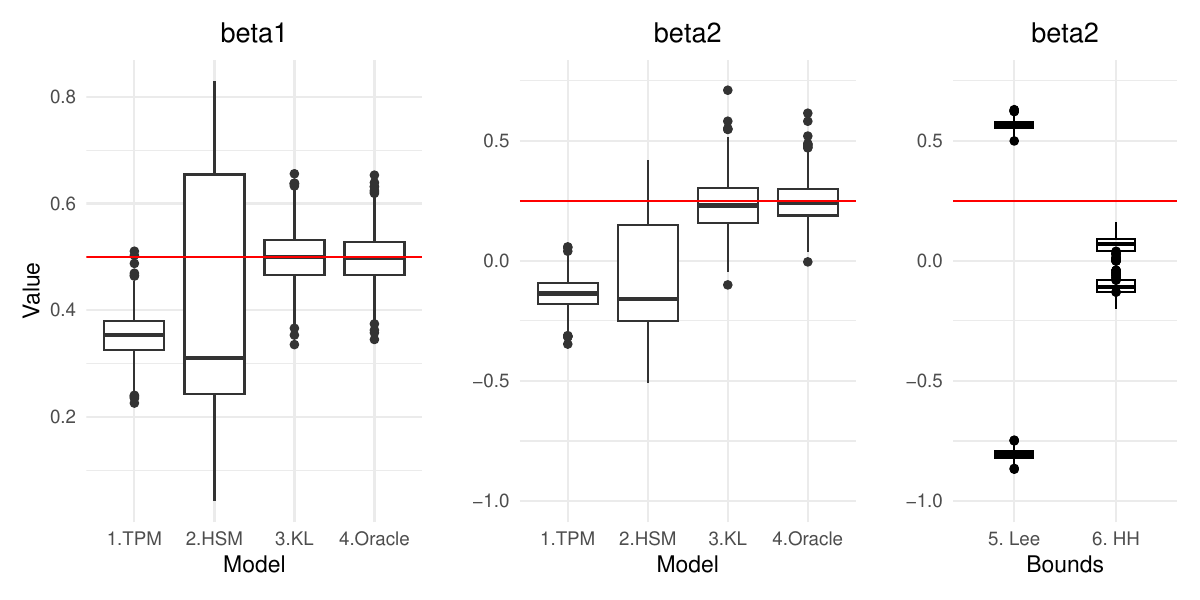}
  \end{minipage}

  \par\footnotesize\renewcommand{\baselineskip}{11pt}\raggedright
\textbf{Note:} Boxplots mark the IQR with the median (thick bar); whiskers extend to 1.5~×~IQR and dots denote outliers. Red lines indicate true parameter values. TPM, HSM, KL, and Oracle denote OLS under random selection, Heckman’s MLE, our semiparametric sieve estimator, and the oracle estimator, respectively. Lee and HH bounds are computed using a larger sample size ($n=100,000$).
\end{figure}
\FloatBarrier

\subsection{Lee and Honor\'{e}--Hu bounds}\label{subsec: bounds results}

We also evaluate the bounds proposed by \cite{lee2009bounds} and \cite{honore2020selection} using DGP2. We do not consider DGP1 because there is no binary treatment variable for which Lee’s bounds are applicable. For the HH bounds, we closely follow the HH implementation (with the logistic first stage). We use $n=100{,}000$, rather than the $n=5{,}000$ used for the point estimators, because the HH bounds are not reliably estimated at the smaller sample size: they are empty in 93 of 1,000 iterations. With the larger sample size, both bounds are reliably estimated. The right panel of Figure~\ref{fig:simul-overall}(b) displays the box plots of the HH and Lee bounds. The Lee bounds consistently contain the true parameter value for $\beta_2$ because they are very wide in this setup. They are never informative about the sign of the treatment effect, as they include zero in every iteration under DGP2. Although much tighter, the HH bounds are uninformative in most iterations and never contain the true value because the model misspecifies the selection process.

\subsection{Simulation results on weak nonlinearity}\label{appendix: weak nonlinearity}

\begin{table}[!htbp]
\centering\footnotesize
\setlength{\tabcolsep}{4pt}
\caption{Finite-sample performance of estimators}\label{tab: finite sample performance - weak nonlinearity}
\begin{tabular}{lccccccccc}
\toprule
 & & \multicolumn{4}{c}{Weak nonlinearity} & \multicolumn{4}{c}{Very weak nonlinearity} \\
\cmidrule(lr){3-6} \cmidrule(lr){7-10}
 & & TPM & HSM & KL & Oracle & TPM & HSM & KL & Oracle \\
\midrule
 \multirow{2}{*}{RMSE} &$\beta_1$ &  0.159 & 0.034 &  0.051 &  0.037 &  0.125 &  0.035 &  0.048 &  0.035 \\
                       & $\beta_2$ &  0.305 & 0.040 &  0.075 &  0.053 &  0.287 &  0.040 &  0.075 &  0.054 \\
 \multirow{2}{*}{Bias} & $\beta_1$ & -0.156 & -0.009 & -0.011 & -0.003 & -0.121 & -0.017 & -0.011 & -0.003 \\
                       & $\beta_2$ & -0.303 & 0.003 & -0.018 & -0.004 & -0.286 & -0.001 & -0.020 & -0.004 \\
\multirow{2}{*}{Coverage (uncorrected)} & $\beta_1$ & 0.000 & 0.946 & 0.926 & 0.939 &  0.015 & 0.919 & 0.924 & 0.942  \\
                       & $\beta_2$ & 0.000 & 0.949 & 0.924 & 0.948 & 0.000 & 0.948 & 0.916 & 0.943 \\
\multirow{2}{*}{Coverage (corrected)} & $\beta_1$ & -- & -- & 0.944 & -- & -- & -- & 0.942 & -- \\
                       & $\beta_2$ & -- & -- & 0.942 & -- & -- & -- & 0.929 & -- \\
\bottomrule
\end{tabular}
\par\footnotesize\renewcommand{\baselineskip}{11pt}\justifying
\textbf{Note:} TPM, HSM, KL, and Oracle denote OLS under random selection, Heckman’s MLE, our semiparametric sieve estimator, and the oracle estimator, respectively. For the KL estimator, `uncorrected' coverage uses the robust standard errors of the second-stage regression, which treat $\hat{p}$ as known, and `corrected' uses the first-stage-corrected standard errors of equation~\eqref{eq: corrected variance estimator} of the paper; TPM, HSM, and Oracle use their conventional standard errors.
\end{table}
\FloatBarrier

\begin{figure}[!htbp]
  \centering
  \caption{Finite-sample performance with weak nonlinearity}
  \label{fig:simul-weak-nonlinearity}
  \begin{minipage}[t]{0.95\textwidth}
    \centering
    \textbf{(a) DGP1: weak nonlinearity
    ($\alpha_2,\alpha_3,\alpha_6$ divided by 10)}\par\smallskip
    \includegraphics[width=\linewidth]{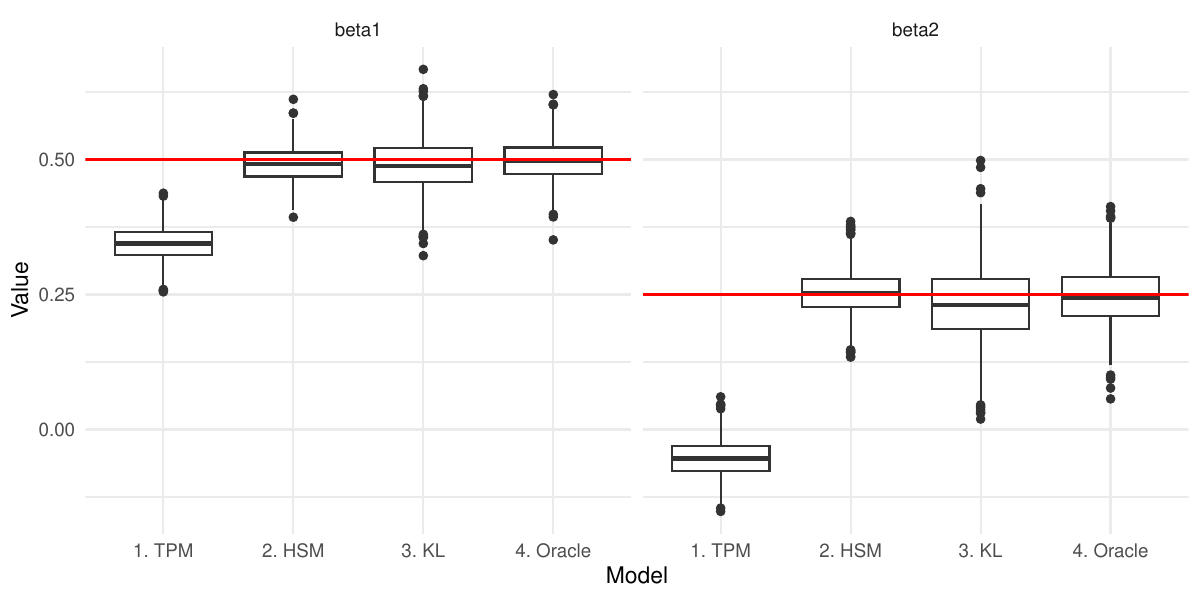}
  \end{minipage}

  \vspace{1em}

  \begin{minipage}[t]{0.95\textwidth}
    \centering
    \textbf{(b) DGP1: very weak nonlinearity
    ($\alpha_2,\alpha_3,\alpha_6$ divided by 100)}\par\smallskip
    \includegraphics[width=\linewidth]{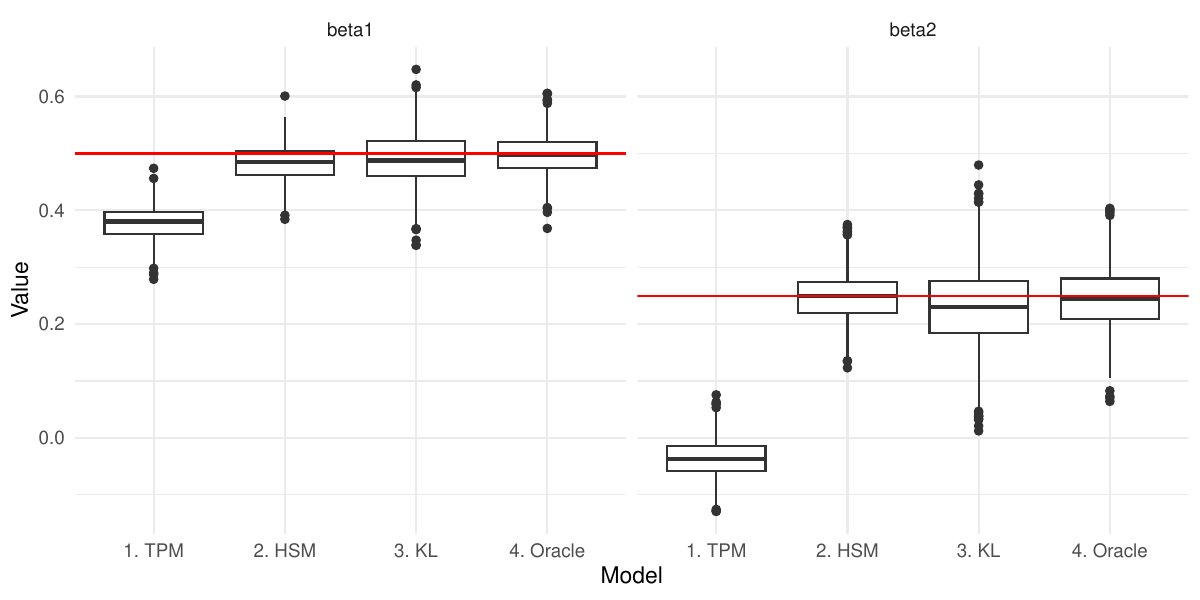}
  \end{minipage}

  \par\footnotesize\renewcommand{\baselineskip}{11pt}\raggedright
\textbf{Note:} Boxplots mark the IQR with the median (thick bar); whiskers extend to 1.5~×~IQR and dots denote outliers. Red lines indicate true parameter values. TPM, HSM, KL, and Oracle denote OLS under random selection, Heckman’s MLE, our semiparametric sieve estimator, and the oracle estimator, respectively.
\end{figure}
\FloatBarrier

\section{Nonlinearity in the selection equation}\label{appendix: nonlinearity in selection}

We estimate the selection index by the unpenalized sieve probit MLE (Assumption~\ref{assu:convergence} of the paper), implemented with the \texttt{gam} command of the \texttt{mgcv} package in R with fixed degrees of freedom (\texttt{fx = TRUE}). Age and experience enter through group-specific cubic regression splines with five basis functions, so each smooth term has exactly four degrees of freedom after the identifiability constraint. Because identification rests on nonlinearity of the selection index, the relevant empirical question is not whether age and experience matter in selection but whether their effects are nonlinear, and a significance test of a whole smooth term cannot distinguish the two. We therefore report the nested likelihood-ratio tests of Remark~\ref{rem: nonlinearity test} of the paper, comparing restricted index specifications with the sieve; all models are unpenalized maximum likelihood fits of the same binomial probit likelihood, so the classical chi-square theory applies. Table~\ref{tab:nonlinearity} reports four tests for each subsample: the null of a (group-specific) linear index in age and experience against the sieve (twelve degrees of freedom, three per smooth); the nulls that age and that experience enter linearly, one at a time, keeping the other smooth unrestricted (six degrees of freedom each); and the identification-relevant null that the selection index is spanned by the outcome-equation regressors, $g(X)=X'\gamma$ with $X$ including the quadratic terms in age and experience, the case in which the model is unidentified. For the last test the sieve is augmented with the two quadratic terms so that the null is nested, because the natural cubic regression spline space does not contain quadratic polynomials exactly (fourteen degrees of freedom).

Every null is overwhelmingly rejected in all four subsamples: the linear index is rejected with likelihood-ratio statistics between 418.5 and 1555.2 on twelve degrees of freedom, age and experience each exhibit significant curvature ($p<0.001$ in every case), and the outcome-equation index is rejected with statistics between 155.4 and 1066.3 on fourteen degrees of freedom. The selection index is thus strongly nonlinear, and nonlinear beyond the quadratic terms already present in the outcome equation, which is the source of identification in our model (Remark~\ref{rem: nonlinearity test} of the paper). We further perturb the number of knots in both the selection and outcome regressions, and the results remain stable. We report the wage regression results with seven knots in the next section (Table~\ref{tab:regression-hetero-7-knots}).

\begin{table}[!htbp]
\centering\footnotesize
\caption{Likelihood-ratio tests of nonlinearity in the selection index}
\label{tab:nonlinearity}
\begin{tabular}{lccccc}
\toprule
& & \multicolumn{2}{c}{Racial selection eq.} & \multicolumn{2}{c}{Gender selection eq.} \\
\cmidrule(lr){3-4}\cmidrule(lr){5-6}
Null hypothesis (index restriction) & df & Men & Women & White & Mexican \\
\midrule
Linear in age and experience         & 12 & 571.9 & 1363.4 & 1555.2 & 418.5 \\
Age enters linearly                  & 6  & 273.2 & 426.5  & 531.8  & 133.8 \\
Experience enters linearly           & 6  & 51.8  & 286.8  & 327.7  & 22.6  \\
Outcome-equation index $F(X'\gamma)$ & 14 & 155.4 & 656.4  & 1066.3 & 155.7 \\
\bottomrule
\end{tabular}

\vspace{0.5em}
\par\footnotesize\renewcommand{\baselineskip}{11pt}\justifying
\textbf{Note:} Each cell reports the likelihood-ratio statistic $2(l_{\mathrm{unres}}-l_{\mathrm{res}})$ of the restricted index specification in the row against the group-specific cubic-regression-spline sieve of the first stage; all models are unpenalized probit maximum likelihood fits, and every statistic is significant at $p<0.001$. In the first three rows the restricted models impose group-specific linear terms for the indicated covariate(s) and the unrestricted model is the sieve exactly as used in estimation. In the last row the null is the index spanned by the outcome-equation regressors, age, age$^{2}$, experience, experience$^{2}$, and the remaining controls, with common coefficients, under which the model is unidentified (Remark~\ref{rem: nonlinearity test} of the paper); because the natural cubic regression spline space does not contain quadratic polynomials, the unrestricted model for this test is the sieve augmented with age$^{2}$ and experience$^{2}$, so that the null is nested. df denotes the degrees of freedom of the $\chi^{2}$ null distribution.%
\end{table}
\FloatBarrier
\newpage

\section{Full wage regression tables}\label{appendix: full wage regression tables}


\begin{table}[!htbp]
\centering\footnotesize
\setlength{\tabcolsep}{3pt}
\caption{Wage regression and racial wage disparity}
\label{tab:new-results}

\begin{tabular}{l *{3}{d{3.3}}c *{3}{d{3.3}}c}
\toprule
& \multicolumn{4}{c}{Men} & \multicolumn{4}{c}{Women} \\
\cmidrule(lr){2-5}\cmidrule(lr){6-9}
Variable
& \multicolumn{1}{c}{TPM}
& \multicolumn{1}{c}{HSM}
& \multicolumn{1}{c}{KL}
& \multicolumn{1}{c}{HH}
& \multicolumn{1}{c}{TPM}
& \multicolumn{1}{c}{HSM}
& \multicolumn{1}{c}{KL}
& \multicolumn{1}{c}{HH}\\
\midrule
mexican & -0.113 & -0.113 & -0.089 & [-0.114, -0.103]
             & -0.078 & -0.078 & -0.066 & [-0.089, -0.066]\\
             & (0.005)& (0.005)& (0.009)&
             & (0.005)& (0.005)& (0.006)& \\[3pt]

age & 0.078 & 0.079 & 0.106 & [0.078, 0.091]
             & 0.111 & 0.113 & 0.132 & [0.095, 0.132]\\
             & (0.006)& (0.006)& (0.010)&
             & (0.007)& (0.007)& (0.009)& \\[3pt]

age$^2$ & 0.000 & 0.000 & -0.001 & [-0.001, 0.000]
             & 0.000 & 0.000 & -0.001 & [-0.001, 0.000]\\
             & (0.000)& (0.000)& (0.000)&
             & (0.000)& (0.000)& (0.000)& \\[3pt]

exp & -0.025 & -0.025 & -0.041 & [-0.032, -0.024]
             & -0.069 & -0.070 & -0.081 & [-0.082, -0.061]\\
             & (0.005)& (0.005)& (0.008)&
             & (0.006)& (0.006)& (0.007)& \\[3pt]

exp$^2$ & 0.000 & 0.000 & 0.000 & [0.000, 0.000]
             & 0.000 & 0.000 & 0.000 & [0.000, 0.000]\\
             & (0.000)& (0.000)& (0.000)&
             & (0.000)& (0.000)& (0.000)& \\[3pt]

less hs & -0.169 & -0.170 & -0.213 & [-0.191, -0.174]
             & -0.174 & -0.177 & -0.225 & [-0.235, -0.177]\\
             & (0.012)& (0.012)& (0.017)&
             & (0.014)& (0.015)& (0.017)& \\[3pt]

some college & 0.052 & 0.051 & 0.045 & [0.047, 0.050]
             & 0.033 & 0.033 & 0.027 & [0.022, 0.028]\\
             & (0.009)& (0.009)& (0.010)&
             & (0.011)& (0.011)& (0.011)& \\[3pt]

college & 0.235 & 0.235 & 0.211 & [0.220, 0.230]
             & 0.156 & 0.155 & 0.134 & [0.119, 0.143]\\
             & (0.022)& (0.023)& (0.025)&
             & (0.025)& (0.025)& (0.026)& \\[3pt]

adv degree & 0.258 & 0.257 & 0.208 & [0.230, 0.252]
             & 0.200 & 0.199 & 0.171 & [0.148, 0.174]\\
             & (0.031)& (0.031)& (0.036)&
             & (0.034)& (0.034)& (0.035)& \\[3pt]

veteran & -0.001 & -0.001 & 0.012 & [-0.001, 0.005]
             & 0.030 & 0.030 & 0.032 & [0.029, 0.031]\\
             & (0.006)& (0.006)& (0.007)&
             & (0.016)& (0.016)& (0.015)& \\[3pt]

married & 0.135 & 0.136 & 0.175 & [0.134, 0.153]
             & 0.034 & 0.033 & 0.011 & [0.013, 0.053]\\
             & (0.004)& (0.005)& (0.012)&
             & (0.004)& (0.005)& (0.007)& \\[3pt]

calif & 0.151 & 0.151 & 0.142 & [0.147, 0.151]
             & 0.204 & 0.204 & 0.199 & [0.199, 0.208]\\
             & (0.007)& (0.007)& (0.008)&
             & (0.007)& (0.007)& (0.008)& \\[3pt]

arizo & 0.042 & 0.042 & 0.050 & [0.042, 0.045]
             & 0.098 & 0.098 & 0.099 & [0.097, 0.098]\\
             & (0.009)& (0.009)& (0.010)&
             & (0.009)& (0.009)& (0.009)& \\[3pt]

texas & 0.015 & 0.015 & 0.039 & [0.014, 0.026]
             & 0.030 & 0.031 & 0.037 & [0.025, 0.037]\\
             & (0.007)& (0.008)& (0.011)&
             & (0.008)& (0.008)& (0.008)& \\

\bottomrule
\end{tabular}

\vspace{0.5em}
\par\footnotesize\renewcommand{\baselineskip}{11pt}\justifying
\textbf{Note:} Standard errors appear in parentheses below coefficients. For the KL estimator, the reported standard errors are the first-stage-corrected standard errors of Proposition~\ref{prop: asymptotics} of the paper, computed via equation~\eqref{eq: corrected variance estimator} of the paper. TPM, HSM, KL, and HH denote OLS under random selection, Heckman’s MLE, our semiparametric two-step estimator, and \cite{honore2020selection} bounds, respectively.
\end{table}
\FloatBarrier

\begin{table}[!htbp]
\centering\footnotesize
\setlength{\tabcolsep}{3pt}
\caption{Wage regression and gender wage disparity}
\label{tab:wage-results}

\begin{tabular}{l *{3}{d{3.3}}c *{3}{d{3.3}}c}
\toprule
& \multicolumn{4}{c}{White} & \multicolumn{4}{c}{Mexican} \\
\cmidrule(lr){2-5}\cmidrule(lr){6-9}
Variable
& \multicolumn{1}{c}{TPM}
& \multicolumn{1}{c}{HSM}
& \multicolumn{1}{c}{KL}
& \multicolumn{1}{c}{HH}
& \multicolumn{1}{c}{TPM}
& \multicolumn{1}{c}{HSM}
& \multicolumn{1}{c}{KL}
& \multicolumn{1}{c}{HH}\\
\midrule
female & -0.209 & -0.159 & -0.211 & [-0.214, -0.179]
             & -0.193 & -0.195 & -0.180 & [-0.219, -0.145]\\
             & (0.003)& (0.004)& (0.005)&
             & (0.006)& (0.007)& (0.012)& \\[3pt]

age & 0.103 & 0.045 & 0.109 & [0.072, 0.116]
             & 0.056 & 0.058 & 0.044 & [0.013, 0.078]\\
             & (0.006)& (0.006)& (0.007)&
             & (0.009)& (0.010)& (0.012)& \\[3pt]

age$^2$ & -0.001 & 0.000 & -0.001 & [-0.001, 0.000]
             & 0.000 & 0.000 & 0.000 & [0.000, 0.000]\\
             & (0.000)& (0.000)& (0.000)&
             & (0.000)& (0.000)& (0.000)& \\[3pt]

exp & -0.053 & -0.019 & -0.057 & [-0.062, -0.035]
             & -0.019 & -0.020 & -0.013 & [-0.031, 0.003]\\
             & (0.005)& (0.005)& (0.006)&
             & (0.007)& (0.007)& (0.008)& \\[3pt]

exp$^2$ & 0.000 & 0.000 & 0.000 & [0.000, 0.000]
             & 0.000 & 0.000 & 0.000 & [0.000, 0.000]\\
             & (0.000)& (0.000)& (0.000)&
             & (0.000)& (0.000)& (0.000)& \\[3pt]

less hs & -0.159 & -0.042 & -0.164 & [-0.164, -0.099]
             & -0.186 & -0.188 & -0.170 & [-0.220, -0.131]\\
             & (0.012)& (0.013)& (0.013)&
             & (0.014)& (0.015)& (0.018)& \\[3pt]

some college & 0.033 & 0.047 & 0.031 & [0.022, 0.032]
             & 0.078 & 0.078 & 0.074 & [0.066, 0.080]\\
             & (0.009)& (0.010)& (0.009)&
             & (0.012)& (0.012)& (0.012)& \\[3pt]

college & 0.171 & 0.217 & 0.165 & [0.143, 0.177]
             & 0.323 & 0.323 & 0.316 & [0.317, 0.318]\\
             & (0.021)& (0.023)& (0.022)&
             & (0.029)& (0.029)& (0.028)& \\[3pt]

adv degree & 0.201 & 0.268 & 0.194 & [0.170, 0.219]
             & 0.408 & 0.406 & 0.404 & [0.387, 0.422]\\
             & (0.029)& (0.032)& (0.030)&
             & (0.040)& (0.040)& (0.039)& \\[3pt]

veteran & 0.000 & -0.024 & 0.002 & [-0.013, 0.003]
             & 0.069 & 0.068 & 0.076 & [0.057, 0.093]\\
             & (0.006)& (0.007)& (0.006)&
             & (0.013)& (0.013)& (0.014)& \\[3pt]

married & 0.081 & 0.097 & 0.080 & [0.079, 0.092]
             & 0.095 & 0.095 & 0.091 & [0.084, 0.100]\\
             & (0.003)& (0.004)& (0.003)&
             & (0.006)& (0.006)& (0.007)& \\[3pt]

calif & 0.185 & 0.209 & 0.184 & [0.183, 0.199]
             & 0.147 & 0.147 & 0.146 & [0.144, 0.149]\\
             & (0.006)& (0.007)& (0.006)&
             & (0.012)& (0.012)& (0.013)& \\[3pt]

arizo & 0.074 & 0.066 & 0.075 & [0.071, 0.075]
             & 0.052 & 0.053 & 0.047 & [0.035, 0.061]\\
             & (0.007)& (0.008)& (0.007)&
             & (0.015)& (0.015)& (0.016)& \\[3pt]

texas & 0.035 & 0.001 & 0.038 & [0.015, 0.039]
             & -0.029 & -0.028 & -0.039 & [-0.060, -0.013]\\
             & (0.006)& (0.007)& (0.006)&
             & (0.012)& (0.012)& (0.014)& \\

\bottomrule
\end{tabular}

\vspace{0.5em}
\par\footnotesize\renewcommand{\baselineskip}{11pt}\justifying
\textbf{Note:} Standard errors appear in parentheses below coefficients. For the KL estimator, the reported standard errors are the first-stage-corrected standard errors of Proposition~\ref{prop: asymptotics} of the paper, computed via equation~\eqref{eq: corrected variance estimator} of the paper. TPM, HSM, KL, and HH denote OLS under random selection, Heckman’s MLE, our semiparametric two-step estimator, and \cite{honore2020selection} bounds, respectively.
\end{table}
\FloatBarrier

\begin{table}[!htbp]
\centering\footnotesize
\setlength{\tabcolsep}{5pt}
\caption{Semiparametric wage regression: robust versus first-stage-corrected standard errors}
\label{tab:regression-hetero}
\begin{tabular}{lcccccccc}
\hline
 & \multicolumn{2}{c}{Men} & \multicolumn{2}{c}{Women} & \multicolumn{2}{c}{White} & \multicolumn{2}{c}{Mexican} \\
Variable & s.e.$_{rob}$ & s.e.$_{corr}$ & s.e.$_{rob}$ & s.e.$_{corr}$ & s.e.$_{rob}$ & s.e.$_{corr}$ & s.e.$_{rob}$ & s.e.$_{corr}$ \\
\hline
wage gap      & (0.0085) & (0.0089) & (0.0060) & (0.0061) & (0.0045) & (0.0046) & (0.0118) & (0.0120) \\
age           & (0.0098) & (0.0103) & (0.0083) & (0.0085) & (0.0066) & (0.0066) & (0.0121) & (0.0122) \\
age$^2$       & (0.0001) & (0.0001) & (0.0001) & (0.0001) & (0.0000) & (0.0000) & (0.0001) & (0.0001) \\
exp           & (0.0072) & (0.0076) & (0.0064) & (0.0065) & (0.0055) & (0.0055) & (0.0078) & (0.0079) \\
exp$^2$       & (0.0001) & (0.0001) & (0.0001) & (0.0001) & (0.0000) & (0.0000) & (0.0001) & (0.0001) \\
less hs       & (0.0160) & (0.0168) & (0.0164) & (0.0167) & (0.0131) & (0.0132) & (0.0179) & (0.0180) \\
some college  & (0.0098) & (0.0103) & (0.0103) & (0.0106) & (0.0091) & (0.0091) & (0.0120) & (0.0121) \\
college       & (0.0238) & (0.0251) & (0.0247) & (0.0255) & (0.0222) & (0.0222) & (0.0284) & (0.0285) \\
adv degree    & (0.0340) & (0.0358) & (0.0337) & (0.0347) & (0.0304) & (0.0305) & (0.0386) & (0.0389) \\
veteran       & (0.0068) & (0.0071) & (0.0145) & (0.0147) & (0.0059) & (0.0059) & (0.0136) & (0.0137) \\
married       & (0.0117) & (0.0123) & (0.0069) & (0.0070) & (0.0034) & (0.0034) & (0.0068) & (0.0068) \\
calif         & (0.0076) & (0.0081) & (0.0074) & (0.0076) & (0.0058) & (0.0058) & (0.0127) & (0.0128) \\
arizo         & (0.0091) & (0.0097) & (0.0091) & (0.0093) & (0.0070) & (0.0070) & (0.0160) & (0.0161) \\
texas         & (0.0100) & (0.0105) & (0.0076) & (0.0077) & (0.0061) & (0.0061) & (0.0139) & (0.0140) \\
\hline
\end{tabular}
\vspace{0.3em}
\par\footnotesize\renewcommand{\baselineskip}{11pt}\justifying
\textbf{Note:} s.e.$_{rob}$ is the heteroskedasticity-robust standard error from the second-stage regression, which ignores the first stage; s.e.$_{corr}$ is the first-stage-corrected standard error of Proposition~\ref{prop: asymptotics} of the paper, computed via equation~\eqref{eq: corrected variance estimator} of the paper.
\end{table}

\begin{table}[!htbp]
\centering\footnotesize
\setlength{\tabcolsep}{5pt}
\caption{Wage regression with seven knots}
\label{tab:regression-hetero-7-knots}
\begin{tabular}{lcccccccc}
\hline
 & \multicolumn{2}{c}{Men} & \multicolumn{2}{c}{Women} & \multicolumn{2}{c}{White} & \multicolumn{2}{c}{Mexican} \\
Variable & Coef & s.e. & Coef & s.e. & Coef & s.e. & Coef & s.e. \\
\hline
wage gap      & -0.090 & (0.0087) & -0.067 & (0.0062) & -0.212 & (0.0046) & -0.176 & (0.0114) \\
age           & 0.105 & (0.0101) & 0.130 & (0.0084) & 0.108 & (0.0065) & 0.042 & (0.0119) \\
age$^2$       & -0.001 & (0.0001) & -0.001 & (0.0001) & -0.001 & (0.0000) & 0.000 & (0.0001) \\
exp           & -0.041 & (0.0074) & -0.080 & (0.0065) & -0.057 & (0.0055) & -0.012 & (0.0077) \\
exp$^2$       & 0.000 & (0.0001) & 0.000 & (0.0001) & 0.000 & (0.0000) & 0.000 & (0.0001) \\
less hs       & -0.212 & (0.0167) & -0.223 & (0.0166) & -0.164 & (0.0132) & -0.166 & (0.0176) \\
some college  & 0.045 & (0.0102) & 0.027 & (0.0106) & 0.030 & (0.0091) & 0.074 & (0.0122) \\
college       & 0.212 & (0.0249) & 0.133 & (0.0256) & 0.164 & (0.0222) & 0.315 & (0.0287) \\
adv degree    & 0.209 & (0.0354) & 0.170 & (0.0348) & 0.193 & (0.0304) & 0.405 & (0.0390) \\
veteran       & 0.011 & (0.0071) & 0.032 & (0.0147) & 0.003 & (0.0059) & 0.078 & (0.0137) \\
married       & 0.174 & (0.0120) & 0.014 & (0.0070) & 0.081 & (0.0034) & 0.090 & (0.0068) \\
calif         & 0.142 & (0.0081) & 0.199 & (0.0075) & 0.185 & (0.0058) & 0.146 & (0.0128) \\
arizo         & 0.050 & (0.0096) & 0.099 & (0.0092) & 0.075 & (0.0070) & 0.045 & (0.0160) \\
texas         & 0.038 & (0.0104) & 0.037 & (0.0077) & 0.037 & (0.0061) & -0.041 & (0.0136) \\
\hline
\end{tabular}
\vspace{0.3em}
\par\footnotesize\renewcommand{\baselineskip}{11pt}\justifying
\textbf{Note:} Reported standard errors are the first-stage-corrected standard errors of Proposition~\ref{prop: asymptotics} of the paper, computed via equation~\eqref{eq: corrected variance estimator} of the paper.
\end{table}
\FloatBarrier

\bibliographystyle{chicago}
\bibliography{ref}

\let\section\ArxivSection

\clearpage
\let\bibliography\ArxivBibliography
\let\bibliographystyle\ArxivBibliographystyle
\bibliographystyle{chicago}
\bibliography{ref}

@article{ichimura1993semiparametric,
  title={Semiparametric least squares ({SLS}) and weighted {SLS} estimation of single-index models},
  author={Ichimura, Hidehiko},
  journal={Journal of Econometrics},
  volume={58},
  pages={71--120},
  year={1993},
  publisher={Elsevier}
}

@article{vytlacil2002independence,
  title={Independence, monotonicity, and latent index models: An equivalence result},
  author={Vytlacil, Edward},
  journal={Econometrica},
  volume={70},
  pages={331--41},
  year={2002}
}

@article{donald1994series,
  title={Series estimation of semilinear models},
  author={Donald, Stephen G. and Newey, Whitney K.},
  journal={Journal of Multivariate Analysis},
  volume={50},
  pages={30--40},
  year={1994},
  publisher={Elsevier}
}

@book{li2007nonparametric,
  title={Nonparametric Econometrics: Theory and Practice},
  author={Li, Qi and Racine, Jeffrey Scott},
  year={2007},
  publisher={Princeton University Press},
  address = {Princeton, NJ}
}

@book{mincer1974schooling,
  title={Schooling, Experience, and Earnings},
  author={Mincer, Jacob},
  year={1974},
  publisher={National Bureau of Economic Research},
  address={New York, NY}
}

@article{murphy1992structure,
  title={The structure of wages},
  author={Murphy, Kevin M and Welch, Finis},
  journal={Quarterly Journal of Economics},
  volume={107},
  pages={285--326},
  year={1992},
  publisher={MIT Press}
}

@article{murphy1990empirical,
  title={Empirical age-earnings profiles},
  author={Murphy, Kevin M. and Welch, Finis},
  journal={Journal of Labor Economics},
  volume={8},
  pages={202--29},
  year={1990},
  publisher={University of Chicago Press}
}

@incollection{card1999causal,
  title={The causal effect of education on earnings},
  author={Card, David},
  editor={Ashenfelter, Orley and Card, David},
  booktitle={Handbook of Labor Economics},
  volume={3A},
  pages={1801--63},
  year={1999},
  publisher={Elsevier},
  address={Amsterdam}
}

@incollection{lemieux2006mincer,
  title={The ``{M}incer equation'' thirty years after schooling, experience, and earnings},
  author={Lemieux, Thomas},
  editor={Grossbard, Shoshana},
  booktitle={Jacob Mincer: A Pioneer of Modern Labor Economics},
  pages={127--45},
  year={2006},
  publisher={Springer},
  address={Boston, MA}
}

@article{manning2001estimating,
  title={Estimating log models: to transform or not to transform?},
  author={Manning, Willard G. and Mullahy, John},
  journal={Journal of Health Economics},
  volume={20},
  pages={461--94},
  year={2001},
  publisher={Elsevier}
}

@article{chernozhukov2018double,
  title={Double/debiased machine learning for treatment and structural parameters},
  author={Chernozhukov, Victor and Chetverikov, Denis and Demirer, Mert and Duflo, Esther and Hansen, Christian and Newey, Whitney and Robins, James},
  journal={Econometrics Journal},
  volume={21},
  pages={C1--C68},
  year={2018},
  doi={10.1111/ectj.12097}
}

@article{robinson1988root,
  title={Root-{N}-consistent semiparametric regression},
  author={Robinson, Peter M.},
  journal={Econometrica},
  volume={56},
  pages={931--54},
  year={1988}
}

@techreport{pan2024locally,
  title={Locally robust semiparametric estimation of sample selection models without exclusion restrictions},
  author={Pan, Zhewen and Zhang, Yifan},
  type={arXiv preprint},
  number={arXiv:2412.01208},
  institution={arXiv},
  year={2024}
}

@article{mullahy2024transform,
  title={Why transform {Y}? {T}he pitfalls of transformed regressions with a mass at zero},
  author={Mullahy, John and Norton, Edward C.},
  journal={Oxford Bulletin of Economics and Statistics},
  volume={86},
  pages={417--47},
  year={2024}
}

@article{chen2024logs,
  title={Logs with zeros? Some problems and solutions},
  author={Chen, Jiafeng and Roth, Jonathan},
  journal={Quarterly Journal of Economics},
  volume={139},
  pages={891--936},
  year={2024},
  publisher={Oxford University Press}
}

@techreport{kroft2024lee,
  title={{H}orowitz-{M}anski-{L}ee bounds with multilayered sample selection},
  author={Kroft, Kory and Mourifi{\'e}, Ismael and Vayalinkal, Atom},
  type={NBER Working Paper},
  number={32952},
  institution={National Bureau of Economic Research},
  address={Cambridge, MA},
  year={2024}
}

@article{semenova2023generalized,
  title={Generalized {L}ee bounds},
  author={Semenova, Vira},
  journal={Journal of Econometrics},
  volume={251},
  pages={106055},
  year={2025},
  doi={10.1016/j.jeconom.2025.106055}
}

@article{heckman1990varieties,
  title={Varieties of selection bias},
  author={Heckman, James J.},
  journal={American Economic Review},
  volume={80},
  number={2},
  pages={313--18},
  year={1990}
}

@article{andrews1998semiparametric,
  title={Semiparametric estimation of the intercept of a sample selection model},
  author={Andrews, Donald W. K. and Schafgans, Marcia M. A.},
  journal={Review of Economic Studies},
  volume={65},
  pages={497--517},
  year={1998},
  publisher={Wiley-Blackwell}
}

@article{mora2008nonparametric,
  title={A nonparametric decomposition of the {M}exican {A}merican average wage gap},
  author={Mora, Ricardo},
  journal={Journal of Applied Econometrics},
  volume={23},
  pages={463--85},
  year={2008}
}

@article{newey1993efficiency,
  title={Efficiency bounds for some semiparametric selection models},
  author={Newey, Whitney K. and Powell, James L.},
  journal={Journal of Econometrics},
  volume={58},
  pages={169--84},
  year={1993},
  publisher={Elsevier}
}

@article{chen2010integrated,
  title={An integrated maximum score estimator for a generalized censored quantile regression model},
  author={Chen, Songnian},
  journal={Journal of Econometrics},
  volume={155},
  pages={90--98},
  year={2010},
  publisher={Elsevier}
}

@article{chen2010non,
  title={Non-parametric identification and estimation of truncated regression models},
  author={Chen, Songnian},
  journal={Review of Economic Studies},
  volume={77},
  pages={127--53},
  year={2010},
  publisher={Wiley-Blackwell}
}

@article{chen2011semiparametric,
  title={Semiparametric estimation of a bivariate {T}obit model},
  author={Chen, Songnian and Zhou, Xianbo},
  journal={Journal of Econometrics},
  volume={165},
  pages={266--74},
  year={2011},
  publisher={Elsevier}
}

@article{chen2012semiparametric,
  title={Semiparametric estimation of a truncated regression model},
  author={Chen, Songnian and Zhou, Xianbo},
  journal={Journal of Econometrics},
  volume={167},
  pages={297--304},
  year={2012},
  publisher={Elsevier}
}

@article{pan2022semiparametric,
  title={Semiparametric Estimation of a Censored Regression Model Subject to Nonparametric Sample Selection},
  author={Pan, Zhewen and Zhou, Xianbo and Zhou, Yahong},
  journal={Journal of Business and Economic Statistics},
  volume={40},
  pages={141--51},
  year={2022}
}

@article{newey2009two,
  title={Two-step series estimation of sample selection models},
  author={Newey, Whitney K.},
  journal={Econometrics Journal},
  volume={12},
  pages={S217--S229},
  year={2009}
}

@article{newey1997convergence,
  title={Convergence rates and asymptotic normality for series estimators},
  author={Newey, Whitney K.},
  journal={Journal of Econometrics},
  volume={79},
  pages={147--68},
  year={1997},
  publisher={Elsevier}
}

@article{belloni2015some,
  title={Some new asymptotic theory for least squares series: pointwise and uniform results},
  author={Belloni, Alexandre and Chernozhukov, Victor and Chetverikov, Denis and Kato, Kengo},
  journal={Journal of Econometrics},
  volume={186},
  pages={345--66},
  year={2015},
  publisher={Elsevier}
}

@article{ackerberg2012practical,
  title={A practical asymptotic variance estimator for two-step semiparametric estimators},
  author={Ackerberg, Daniel A. and Chen, Xiaohong and Hahn, Jinyong},
  journal={Review of Economics and Statistics},
  volume={94},
  pages={481--98},
  year={2012},
  publisher={MIT Press}
}

@article{ahn1993semiparametric,
  title={Semiparametric estimation of censored selection models with a nonparametric selection mechanism},
  author={Ahn, Hyungtaik and Powell, James L.},
  journal={Journal of Econometrics},
  volume={58},
  pages={3--29},
  year={1993},
  publisher={Elsevier}
}

@article{chamberlain1986asymptotic,
  title={Asymptotic efficiency in semi-parametric models with censoring},
  author={Chamberlain, Gary},
  journal={Journal of Econometrics},
  volume={32},
  pages={189--218},
  year={1986},
  publisher={Elsevier}
}

@article{hay1984let,
  title={Let them eat cake: a note on comparing alternative models of the demand for medical care},
  author={Hay, Joel W. and Olsen, Randall J.},
  journal={Journal of Business and Economic Statistics},
  volume={2},
  pages={279--82},
  year={1984}
}

@article{manning1987monte,
  title={{M}onte {C}arlo evidence on the choice between sample selection and two-part models},
  author={Manning, Willard G. and Duan, Naihua and Rogers, William H.},
  journal={Journal of Econometrics},
  volume={35},
  pages={59--82},
  year={1987},
  publisher={Elsevier}
}

@article{duan1984choosing,
  title={Choosing between the sample-selection model and the multi-part model},
  author={Duan, Naihua and Manning, Willard G. and Morris, Carl N. and Newhouse, Joseph P.},
  journal={Journal of Business and Economic Statistics},
  volume={2},
  pages={283--89},
  year={1984}
}

@article{escanciano2016identification,
  title={Identification and estimation of semiparametric two-step models},
  author={Escanciano, Juan Carlos and Jacho-Ch{\'a}vez, David and Lewbel, Arthur},
  journal={Quantitative Economics},
  volume={7},
  pages={561--89},
  year={2016}
}

@article{lee2009bounds,
    author = {Lee, David S.},
    title = "{Training, Wages, and Sample Selection: Estimating Sharp Bounds on Treatment Effects}",
    journal = {Review of Economic Studies},
    volume = {76},
    pages={1071--102},
    year = {2009},
}

@article{honore2020selection,
  title={Selection without exclusion},
  author={Honor{\'e}, Bo E. and Hu, Luojia},
  journal={Econometrica},
  volume={88},
  pages={1007--29},
  year={2020}
}

@article{heckman1979sample,
  title={Sample selection bias as a specification error},
  author={Heckman, James J.},
  journal={Econometrica},
  volume={47},
  pages={153--61},
  year={1979}
}

@article{das2003nonparametric,
  title={Nonparametric estimation of sample selection models},
  author={Das, Mitali and Newey, Whitney K. and Vella, Francis},
  journal={Review of Economic Studies},
  volume={70},
  pages={33--58},
  year={2003},
  publisher={Wiley-Blackwell}
}

@article{newey1994asymptotic,
  title={The asymptotic variance of semiparametric estimators},
  author={Newey, Whitney K.},
  journal={Econometrica},
  volume={62},
  pages={1349--82},
  year={1994}
}

@article{klein1993efficient,
  title={An efficient semiparametric estimator for binary response models},
  author={Klein, Roger W. and Spady, Richard H.},
  journal={Econometrica},
  volume={61},
  pages={387--421},
  year={1993}
}

@book{bickel1993efficient,
  title={Efficient and Adaptive Estimation for Semiparametric Models},
  author={Bickel, Peter J. and Klaassen, Chris A. J. and Ritov, Ya'acov and Wellner, Jon A.},
  year={1993},
  publisher={Johns Hopkins University Press},
  address={Baltimore, MD}
}

@article{newey1990semiparametric,
  title={Semiparametric efficiency bounds},
  author={Newey, Whitney K},
  journal={Journal of Applied Econometrics},
  volume={5},
  pages={99--135},
  year={1990}
}

@book{vaart1996weak,
  title={Weak Convergence and Empirical Processes: With Applications to Statistics},
  author={van der Vaart, Aad W. and Wellner, Jon A.},
  year={1996},
  publisher={Springer},
  address={New York, NY}
}

@article{tropp2012user,
  title={User-friendly tail bounds for sums of random matrices},
  author={Tropp, Joel A.},
  journal={Foundations of Computational Mathematics},
  volume={12},
  pages={389--434},
  year={2012}
}

@article{bousquet2002bennett,
  title={A {B}ennett concentration inequality and its application to suprema of empirical processes},
  author={Bousquet, Olivier},
  journal={Comptes Rendus Mathematique},
  volume={334},
  pages={495--500},
  year={2002}
}

\end{document}